%% file: main.tex
\documentclass[11pt]{article}
\usepackage[top=1in, bottom=1in, left=1in, right=1in]{geometry}
\usepackage{tgpagella}
\usepackage[utf8]{inputenc} % allow utf-8 input
\usepackage[T1]{fontenc}    % use 8-bit T1 fonts
\usepackage[dvipsnames]{xcolor}
\usepackage{fancyhdr}
\usepackage{amsmath}
\usepackage{authblk}
\usepackage{booktabs}
\usepackage{cite}
\usepackage[toc, page]{appendix}
\usepackage{hyperref}      % hyperlinks
\usepackage{url}            % simple URL typesetting
\usepackage{booktabs}       % professional-quality tables
\usepackage{amsfonts}       % blackboard math symbols
\usepackage{nicefrac}       % compact symbols for 1/2, etc.
\usepackage{microtype}      % microtypography
\usepackage{graphicx}

\usepackage{tcolorbox}
\usepackage{listings}
\usepackage{xcolor}
\usepackage{wrapfig}
\usepackage{cleveref}
\usepackage{makecell}
\usepackage{subcaption} 
\usepackage{multirow}
\usepackage{amsthm}
\usepackage{caption}
\usepackage{floatrow}
\usepackage{natbib}

\usepackage{enumitem}
 \usepackage{amssymb}

\input{header}

\input{math_commands}
\author{%
Sheng Liu\textsuperscript{*}$^{\dagger}$
 \;\; Zihan Wang\textsuperscript{*}$^{\ddag}$\;\;Yuxiao Chen$^{\S}$\;\;Qi Lei$^\ddag$
}
\title{Data Reconstruction Attacks and Defenses: A Systematic Evaluation}
\date{}

\begin{document}
\maketitle
\vspace{-15mm}
\begin{center}
\text{$^{\dagger}$ Stanford University}\quad \quad 
\text{$^\ddag$ New York University}\quad \quad \text{$\S$ Peking University}\\
  \vspace{2mm}
   \texttt{shengl@stanford.edu,\quad \{zw3508, ql518\}@nyu.edu} 
\end{center}
\footnotetext{\textsuperscript{*}Equal contribution.}
\setcounter{footnote}{0}
%\vspace{-5mm}

%\vspace{-10mm}
\begin{abstract}
  Reconstruction attacks and defenses are essential in understanding the data leakage problem in machine learning. However, prior work has centered around empirical observations of gradient inversion attacks, lacks theoretical grounding, and cannot disentangle the usefulness of defending methods from the computational limitation of attacking methods.  
In this work, we propose to view the problem as an inverse problem, enabling us to theoretically and systematically evaluate the data reconstruction attack. On various defense methods, we derived the algorithmic upper bound and the matching (in feature dimension and architecture dimension) information-theoretical lower bound on the reconstruction error for two-layer neural networks. To complement the theoretical results and investigate the utility-privacy trade-off, we defined a natural evaluation metric of the defense methods with similar utility loss among the strongest attacks. We further propose a strong reconstruction attack that helps update some previous
understanding of the strength of defense methods under our proposed evaluation metric. 
\end{abstract}
\section{Introduction}

Machine learning research has transformed the technical landscape across various domains but also raises privacy concerns potentially~\citep{li2023multi,papernot2016towards}. Federated learning~\citep{konečný2016federated,mcmahan2017communication}, a collaborative multi-site training framework, aims to uphold user privacy by keeping user data local while only exchanging model parameters and updates between a central server and edge users. However, recent studies on reconstruction attacks~\citep{yin2021see,huang2021evaluating} highlight privacy vulnerabilities even within federated learning. Attackers can eavesdrop on shared trained models and gradient information and reconstruct training data using them. Even worse, honest but curious servers can inadvertently expose training data by querying designed model parameters. Some defending methods are also proposed and analyzed.

However, previous research on reconstruction attacks and defenses centered around empirical studies and lacked theoretical guarantees. Moreover, those experiments over defenses usually focus only on one specific attack, making the evaluation of defenses' strength untrustworthy. Empirical understanding of attacking and defending methods is affected by factors like optimization challenges of non-convex objectives. Under the heuristic attacking methods from prior work, it is hard to disentangle the following two possibilities: 1) the algorithmic or computational barrier of the specific attack and 2) the defense's success. Thus, a systematical and information-theoretical evaluation of defenses that is independent to the attack method is needed.

A line of theoretical work analyzing privacy attacks is differential privacy (DP)~\citep{dwork2006differential}, which guarantees privacy in data re-identification or membership inference. However, there are scenarios when data identity information is not sensitive, but data itself is, and DP is unsuitable for interpreting data reconstruction attacks~\citep{guo2022bounding,hayes2024bounding}. An algorithm can yield no DP guarantee but prevent the threat model from data reconstruction (Details deferred to Section~\ref{sec:preliminary}). 
Although some previous works discussed the reconstruction attack with more tailored Renyi-DP \citep{guo2022bounding,stock2022defending}, their setting is when all samples are known, and the attacker's goal is to reconstruct the last data point, which doesn't apply to the federated learning setting. We will discuss how (Renyi-)DP is too stringent for the data leakage problem under federated learning in Section~\ref{sec:preliminary}. Some underlying structures of the observation can yield no (Renyi-)DP guarantee but
prevent the attacker from reconstructing the data.

\begin{wrapfigure}{r}{0.5\textwidth}
    \centering
   % \vspace{-0.2cm}
    \scalebox{1}{
    \centering
\includegraphics[width=1\columnwidth]{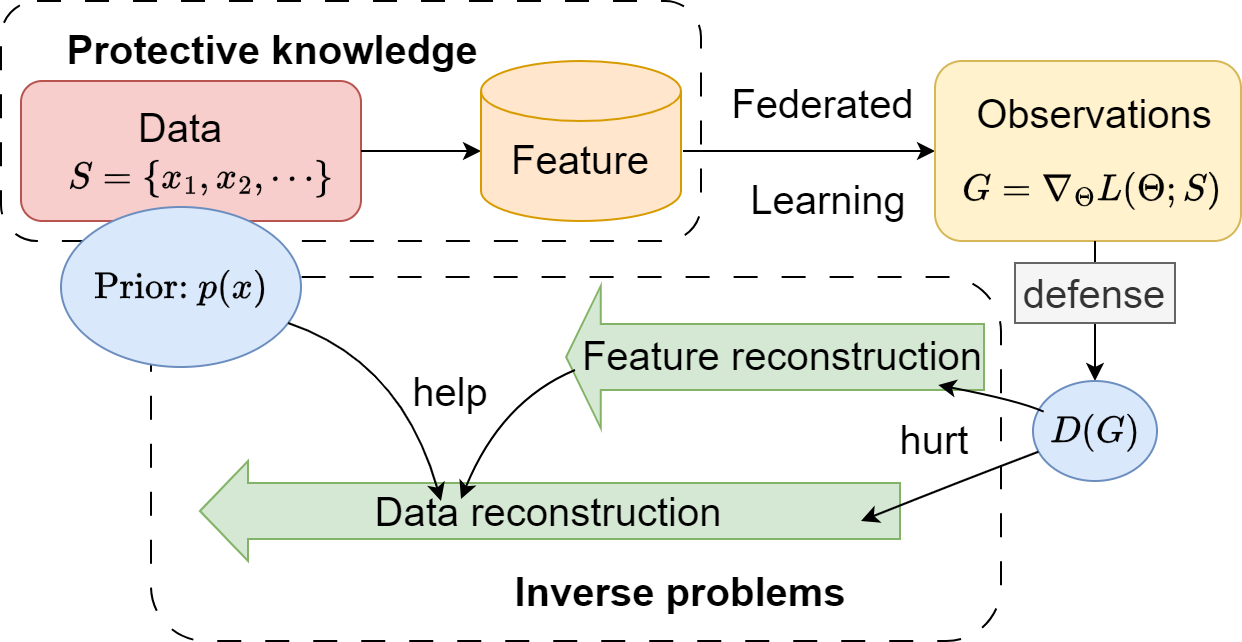}
    \caption{\small An illustration of the key components of data reconstruction studied in this paper.}
    \label{fig:illustration}
    }
    \vspace{-0.5cm}
\end{wrapfigure}

In our work, we propose a framework for evaluating defenses that view the reconstruction attack as an inverse problem (Fig. \ref{fig:illustration}): from the gradient $G$ that is generated from the unknown data $S$ and a known function (objective's gradient), the adversary intends to reconstruct the input data, and a defender $D$ is to prevent this from perturbing the observations. Our theoretical analysis focuses on 2-layer fully connected networks and gives both upper and lower bounds on the error of this inverse problem under various defenses. We are especially interested in obtaining a quantitative relationship between reconstruction error and the key factors of the learning pipeline, such as data dimension, model architecture (width), and defense strength, which are derived in the form of both algorithmic upper bounds and (matching) information-theoretic lower bound (that is independent of the attack algorithm). To generalize the analysis to real-world models and study the utility-privacy trade-offs, we empirically conduct a comprehensive evaluation of different defenses' performances. Our comparisons are based on the strongest attack for each defense to exclude the case that some defenses are only effective for a specific attack.

Our paper is organized as follows. 
In Section \ref{sec:preliminary}, we introduce the exact setups of both theoretical and empirical analysis of defenses. We also explain why DP is not suitable for our setting.
 In Section \ref{sec:theory}, we summarize the defenses used in reconstruction attacks and analyze the error bounds theoretically for 2-layer networks.
 In Section \ref{sec:background}, we introduce several previous attacking methods used in our empirical analysis and propose a strong attack based on an intermediate feature reconstruction method \citep{wang2023reconstructing}. 
In Section \ref{sec:experiments}, we show the results of experiments and evaluate the defenses based on utility loss and best reconstruction error.

\subsection{Related Work}
\input{related_work}
\input{defending}

\section{Empirical Analysis}
\subsection{Attacks for Defense Methods Evaluation}
\label{sec:background}
\input{background}

\subsection{Description of the proposed attack algorithm} 
\label{sec:method}
\input{method}
\input{experiments_arxiv}
\input{results}
%\vspace{5mm}
\section{Conclusion} 
In this paper, we systematically evaluate various defense methods, both theoretically and empirically. We established the algorithmic upper bound and matching information-theoretic lower bound of the reconstruction error for various defenses on two-layer random nets, thereby quantitatively analyzing the effect of defense strength, data dimension, and model width. To extend our theoretical insights to more general network architectures and explore the utility-privacy trade-off, we introduced a robust attack method capable of overcoming a broad spectrum of defense strategies in a setting with an honest but curious server. %To complement our theoretical findings to general networks and to investigate utility-privacy trade-off, we propose a stronger attack method that is robust to a wide range of defense methods in the setting of honest-but-curious server. 
The method enhances the conventional gradient inversion attack by combining it with the powerful feature reconstruction attack that achieves our algorithmic upper bound. 
We propose to evaluate each attack based on their impact on utility and effectiveness against the strongest corresponding attack. 
Our evaluation indicates gradient pruning is the strongest among our considered defenses and against our considered attacks, which updates some previous evaluation results purely based on gradient inversion attacks~\citep{huang2021evaluating}. Our work establishes a more comprehensive understanding of the data reconstruction problem in federated learning. We anticipate that our proposed framework will encourage further research and additional evaluations building on our findings.
%Our theoretical and empirical findings indicate that gradient pruning is the most effective defense method among many others, including gradient clipping, dropout, adding noise, local aggregation, etc. 
% {\red Zihan: need modify the conclusion (done) and also the checklist at the end of the paper.}

\section*{Acknowledgments}
This material is based upon work supported by the U.S. Department of Energy,
Office of Science Energy Earthshot Initiative as part of the project ``Learning reduced models under extreme data conditions for design and rapid decision-making in complex systems" under Award
\#DE-SC0024721.

%\newpage
\bibliographystyle{abbrvnat}
\bibliography{ref}
%%%%%%%%%%%%%%%%%%%%%%%%%%%%%%%%%%%%%%%%%%%%%%%%%%%%%%%%%%%%

\appendix

\input{appendix}

%%%%%%%%%%%%%%%%%%%%%%%%%%%%%%%%%%%%%%%%%%%%%%%%%%%%%%%%%%%%

\end{document}

%% file: header.tex
\usepackage{algorithm}% http://ctan.org/pkg/algorithm
\usepackage{algpseudocode}% http://ctan.org/pkg/algorithmicx
\usepackage{amsmath,amssymb}
\usepackage{enumitem}
\usepackage{comment} 
\usepackage{bm}

\usepackage{thmtools}
\usepackage{thm-restate}

\usepackage{amssymb}
\usepackage{amsmath}
\usepackage{amsfonts}
\usepackage{amsthm}

\usepackage{url}            % simple URL typesetting
\usepackage{booktabs}       % professional-quality tables
\usepackage{amsfonts}       % blackboard math symbols
\usepackage{nicefrac}       % compact symbols for 1/2, etc.
\usepackage{microtype}      % microtypography
  \usepackage{mathtools}

\usepackage{tikz}
\usepackage{pgfplots}
\usetikzlibrary{pgfplots.groupplots}

\usepackage{caption}
\usepackage[bottom,hang,flushmargin]{footmisc}

\newcommand{\tsn}[1]{{\left\vert\kern-0.25ex\left\vert\kern-0.25ex\left\vert #1 
    \right\vert\kern-0.25ex\right\vert\kern-0.25ex\right\vert}}

\newenvironment{itemize*}%
{\begin{itemize}[leftmargin=*,topsep=0pt]%
		\setlength{\itemsep}{0pt}%
		\setlength{\parskip}{0pt}}%
	{\end{itemize}}
\newenvironment{enumerate*}%
{\begin{enumerate}[leftmargin=*,topsep=0pt]%
		\setlength{\itemsep}{0pt}%
		\setlength{\parskip}{0pt}}%
	{\end{enumerate}}

\allowdisplaybreaks

\newtheorem{theorem}{Theorem}[section]

\newtheorem{assumption}{Assumption}[section]

\newtheorem{lemma}[theorem]{Lemma}
\newtheorem{corollary}[theorem]{Corollary}
\newtheorem{proposition}[theorem]{Proposition}
\newtheorem{definition}[theorem]{Definition}

\newtheorem{remark}[theorem]{Remark}

\newcommand{\vct}{\boldsymbol }

\newcommand{\tr}{\mathrm{tr}}

\def\R{\mathbb{R}}

\def\cA{\mathcal{A}}

\def\cD{\mathcal{D}}

\def\cS{\mathcal{S}}

\def\cM{\mathcal{M}}
\def\cN{\mathcal{N}}

\def\cS{\mathcal{S}}

\def\cX{\mathcal{X}}

\def\va{\mathbf{a}}

\newcommand{\norm}[1]{\left\|#1\right\|}

\newtheorem{prop}{Proposition}[section]

\def\approxcorrect{\cmark\kern-1.4ex\raisebox{.30ex}{$\xmark$}}

\newcommand{\idxn}[1][]{\ifthenelse{\equal{#1}{}}{\mathbb{INDQ}_n}{\mathbb{INDQ}_{#1}}}

\newcommand{\beq}{\begin{equation}}

\newcommand{\eeq}{\end{equation}}
\newcommand{\bx}{{\bm{x}}}

\newcommand{\bg}{{\bm{g}}}
\newcommand{\bw}{{\bm{w}}}

\newcommand{\bE}{{\bm{E}}}

\newcommand{\bT}{{\bm{T}}}

\newcommand{\ba}{{\bm{a}}}
\newcommand{\bu}{{\bm{u}}}

\newcommand{\bz}{{\bm{z}}}
\newcommand{\bZ}{{\bm{Z}}}

\newcommand{\z}{{\vct{z}}}

\newcommand{\vh}{\vct{\hat{v}}}

\newcommand{\bb}{\bm{b}}

\newcommand{\vepsilon}{{\bm{\epsilon}}}

\definecolor{emmanuel}{RGB}{255,127,0}

\newcommand{\E}{\operatorname{\mathbb{E}}}
\def \endprf{\hfill {\vrule height6pt width6pt depth0pt}\medskip}

%% file: math_commands.tex
\usepackage{amsmath,amsfonts,bm}

\def\eqref#1{equation~\ref{#1}}
\def\1{\bm{1}}
\def\va{{\bm{a}}}

\def\vg{{\bm{g}}}
\def\vh{{\bm{h}}}

\def\vx{{\bm{x}}}

\def\mI{{\bm{I}}}

\def\mM{{\bm{M}}}

\DeclareMathAlphabet{\mathsfit}{\encodingdefault}{\sfdefault}{m}{sl}
\SetMathAlphabet{\mathsfit}{bold}{\encodingdefault}{\sfdefault}{bx}{n}

\def\R{\mathbb{R}}

\def\cA{\mathcal{A}}

\def\cD{\mathcal{D}}

\def\cS{\mathcal{S}}

\def\cM{\mathcal{M}}
\def\cN{\mathcal{N}}

\def\cS{\mathcal{S}}

\def\cX{\mathcal{X}}

\def\va{\mathbf{a}}

%% file: related_work.tex
In federated learning \citep{konečný2016federated,mcmahan2017communication}, a line of work in gradient inversion aims to reconstruct training data from gradients. Early work by \citealp{aono2017privacy} theoretically showed that reconstruction is possible under a single neuron setting. \citealp{zhu2019deep} next proposed the gradient matching framework, where the training data and labels are recovered by minimizing the $l_2$ distance between gradients generated by dummy variables and real data. Instead of using $l_2$ distance, \citealp{wang2020sapag} proposed a Gaussian kernel-based distance function, and \citealp{geiping2020inverting} used cosine similarity.

Besides using different distance measures, recent works introduced image priors to gradient matching. \citealp{geiping2020inverting} added a total variation regularization term since pixels of a real image are usually similar to their neighbor. \citealp{wei2020framework} introduced a label-based regularization to fit the prediction results. \citealp{balunovic2021bayesian} formulated the gradient inversion attack in a Bayesian framework. \citealp{yin2021see} used a batch-normalization regularization to fit the BN statistics and a group consistency regularization to indicate the spatial equivariance of CNN. Instead of using regularization terms, \citealp{jeon2021gradient,xue2023fast,li2022auditing} represented dummy training data by a generative model to preserve the image prior. %This method trains the generative model's latent variables and parameters simultaneously.

%Besides the framework of gradient matching, there are other methods to reconstruct training data. 
Some recent works investigated new data reconstruction methods outside the framework of gradient matching. 
\citealp{wang2023reconstructing} and \citealp{kariyappa2023cocktail} respectively use tensor-based method and independent component analysis to construct data from gradient information, where optimization is not needed. \citealp{haim2022reconstructing} proposed a reconstruction method using the stationary property for trained neural networks induced by the implicit bias of training. Other works investigate partial data reconstructions with fishing parameters~\citep{wen2022fishing,boenisch2023curious,fowl2021robbing}. 

Defense strategies in this setting, however, are relatively less studied. The original setting of federated learning proposed by \citealp{mcmahan2017communication} used local aggregation to update multiple steps instead of gradients directly. \citealp{bonawitz2016practical,bonawitz2017practical} proposed a secure aggregation that can collect the gradients from multiple clients before updating. \citealp{geyer2017differentially,wei2020federated} introduced differential privacy into federated learning, perturbing the gradients. \citealp{aono2017privacy} demonstrated an encryption framework in the setting of federated learning. Moreover, many tricks designed to improve the performance of neural networks are effective in defending against privacy attacks, such as Dropout \citep{hinton2012improving}, gradient pruning \citep{sun2017meprop}, and MixUp \citep{zhang2017mixup}. The effect of defending strategies against gradient matching attacks are discussed in \citep{zhu2019deep,geiping2020inverting,huang2021evaluating}.

%{\red QL: need to add discussions of the papers mentioned by the reviewers}

%% file: defending.tex
\section{Preliminary}
\label{sec:preliminary}
% {\red QL: the flow is weird here. We need to add a section \\
% - introducing the setting: essential it's an inverse problem;
% For an unknown dataset $S$, we observe $O(S)=\nabla L(\Theta; S)$. \\
% - introduce our definition of reconstruction error, with minimax risk (for lower bound).
% \begin{align*}
%     \mathsf{R} = \min_{\widehat{S} = \widehat{S}(O)}\max_{S\subset \mathcal{X}^B} \mathbb{E}\left[d(S, \widehat{S})\right], 
%     \end{align*}
% here $d(S,\widehat S)=\min_{\text{permutation }\pi}\sum_{i=1}^B \|S_i-\widehat S_{\pi(i)}\|^2
% $
% - introduce the problem of DP analysis (too loose, because it's worst-case analysis, and not tailored to this problem)\\

% }

% {\red QL: 

% - introduce how we're gonna conduct the evaluation both theoretically and empirically.
% First we will analyze the problem by studying the upper and lower bound of the reconstruction error connecting to the model architecture, data dimension and defense mechanism.\\

% The problem with only theoretical analysis is the lack of extension to general networks, and the lack of utility-privacy trade-off. 
% Therefore we solve the problem and introduce a systematic way to further evaluate these defense methods in the empirical section.
% $$ \min_{D\in\mathcal{D_{U}}}\max_{A\in\mathcal{A}} d(S,A(D(O(S)))),$$
% $\mathcal{D_{U}}$ is the set of defense methods with the same level of utility loss, indexed by $U$.  $\mathcal{A}$ is the set of data reconstruction attacks that can be utilized in the considered setting. 
% }

In reconstruction attacks, we observe the model iterations in training and recover data from the gradient $\nabla L(\Theta;S)$, where $L$ is the loss function, $\Theta$ is the model and $S$ is an unknown dataset. Then, it can be regarded as an inverse problem to identify unknown signal $S$ from observation $G=\nabla L(\Theta;S)$, where $\nabla L$ is the forward function. We denote the data space by $\cX$ and the batch size of data by $B$. Then $S\in \cX^B$. Under the inverse problem framework, the effectiveness of attacks and defenses becomes the feasibility and hardness of recovering $S$. Formally, we define a minimax risk of reconstruction error:
\begin{equation}
\label{eq:lower}
    R_L=\left(\min_{\hat{S}=\hat{S}(G)}\max_{S\subset \cX^B}\min_{\pi}\E\left[d(S,\pi(\hat{S}))\right]\right)^{1/2}.
\end{equation}
Here $d(S,\pi(\hat{S}))=\frac{1}{B}\sum_{i=1}^B\|S_i-\hat{S}_{\pi(i)}\|^2$, where $\pi$ is a permutation of $[B]$. We will study in Section \ref{sec:theory} the reconstruction error lower bounded in the setting of a two-layer network with random weights
\footnote{Throughout the paper, we use a unified scaling for data and network weights similar to Mean-field view~\citep{mei2018mean}.}.
For upper bounding reconstruction errors, we consider the error for specific algorithm $A$.
% \begin{equation}
% \label{eq:upper}
%     \max_{S\subset\cX^B}d(S, A(O(S))).
% \end{equation}
%In this sense, we can upper bound the reconstruction with two layer networks if we adopt the tensor based attack proposed by \citep{wang2023reconstructing}. 
We define the reconstruction error by
\begin{equation}
\label{eq:upper error}
    R_U=\max_{S\subset \cX^B}\min_{\pi}\left(d(S,\pi(\hat S))\right)^{1/2}. %\left(\frac{1}{B}\sum_{i=1}^B \|S_i-\hat{S}_i\|^2\right)^{1/2}.
\end{equation}
Note that in both lower and upper bounds, we consider the minimum over permutation. It is because batched gradient descent adds gradient together in an unordered manner, so only identifying the set of data points without their correspondence is important.

Our studies are meant to complement the esteemed DP \citep{dwork2006differential} to analyze the reconstruction as DP is too strong and sometimes unnecessary for our studied settings. On one hand, scenarios exist when reconstruction is impossible, even when DP does not hold
\footnote{For instance, one can easily verify that with linear net and $l_2$ loss, the observation becomes $\sum_{i=1}^B x_i$, where no DP is satisfied but individual recovery of $x_i$ is impossible for $B>1$ unless specific prior knowledge is assumed, for example natural images.}. 
In addition, the price of achieving DP in the data reconstruction setting is high. For DP-SGD \citep{abadi2016deep}, the well-established Gaussian mechanism requires a large variance that will destroy the utility of gradient information:
\begin{prop}[Short version of Proposition \ref{prop:DP guarantee}]
\label{prop:DP guarantee short}
    For a two-layer neural network with $m$ hidden nodes and random weights, we denote the gradient by $G$. Under mild assumptions, the randomized mechanism $\cM=G+\cN(0,\sigma^2I)$ is $(\epsilon, \delta)-$ DP for any $\epsilon,\delta>0$ if $\sigma^2=\Omega(\frac{m\log(1/\delta)}{\epsilon})$. 
\end{prop}
In this work, we will theoretically and empirically evaluate defenses against reconstruction. For theoretical analysis, we bound the reconstruction error under various defenses $D$ by estimating Eq. (\ref{eq:lower}) and (\ref{eq:upper error}), where the observation is $D(G)$ instead of $G$. The theoretical bounds will connect the model architecture, defenses' strength, and data dimension. However, only theoretical analysis lacks an extension to general networks and utility-privacy trade-off. Therefore we introduce a systematic way to evaluate defenses by empirical analysis further. Within a set of defenses $\cD_{\mathcal{U}}$ of the same level of utility loss $\mathcal{U}$, we will measure the following criterion for each $D\in \cD_{\mathcal{U}}$:
\begin{equation}
\label{eq:empirical}
\mathcal{S}_D:=\max_{A\in\cA}d(S,A(D(G))),
\end{equation}
which evaluates their strength against the most effective attacker in $\cA$, the set of considered reconstruction attacks. In experiments, we conduct various attacks under different defenses to estimate Eq. (\ref{eq:empirical}). Then, we compare the best reconstruction error with the same utility loss for each defense.

\begin{table*}[t]
\centering
  \scriptsize
  \resizebox{.99\textwidth}{!}{
  \begin{tabular}{l|c|c|c|c|c|c|c}
  \toprule
   &  No defense & Local aggregation & Gradient noise & Gradient clipping & DP-SGD & Dropout & Gradient pruning \\
  \midrule
   {\bf Upper bound}  & $B\sqrt{\frac{d}{m}}$ & $B\sqrt{\frac{d}{m}}$ & $(B+\sigma_0)\sqrt{\frac{d}{m}}$ & $B\sqrt{\frac{d}{m}}$ & $(B+\sigma_0\max\{1,\frac{\|G\|}{C}\})\sqrt{\frac{d}{m}}$ & $B\sqrt{\frac{d}{(1-p)m}}$ & Unknown \\
   \midrule
   {\bf Lower bound} & $\sigma\sqrt{\frac{d}{m}}$ & $\sigma\sqrt{\frac{d}{m}}$  & $\sigma\sqrt{\frac{d}{m}}$ & $\sigma\max\{1,\frac{\|G\|}{C}\}\sqrt{\frac{d}{m}}$ & $\sigma\max\{1,\frac{\|G\|}{C}\}\sqrt{\frac{d}{m}}$ & $\sigma\sqrt{\frac{d}{(1-p)m}}$ & $\sigma\sqrt{\frac{d}{(1-\hat p)m}}$ \\
  \bottomrule
  \end{tabular}}
  \caption{The algorithmic upper bound and information-theoretic lower bound of the reconstruction error against different defenses. The bound here is the order with respect to different factors. Some parameters are defined later in the subsection of each defense.  %The error bound of tensor based feature reconstruction method against different defenses. The reconstruction error can be upper bounded for most of the defending methods except for gradient pruning, indicating that gradient pruning is a strong defense against the attack.
  } 
  \label{tab:bounds}
\end{table*}

\section{Theoretical Analysis}
\label{sec:theory}
% {\red - introduce the attack method first (move some content in the attack part here) and then the defense methods. make it shorter. }

In theoretical analysis, we mainly focus on two-layer neural networks $f(\bx;\Theta)=\sum_{j=1}^ma_j\sigma(\bw_j^\top \bx)$ with $\ell$ be the square loss. Here $a_j\sim \cN(0,1/m^2)$ and $\bw_j\sim\cN(0,I)$.
For the upper bound, we examine a specific attack method proposed by \citep{wang2023reconstructing}. The attack method involves the computation of gradient-related tensor $\sum_{j=1}^mg(\bw_j)H_p(\bw_j)$, where $g(\bw_j):=\frac{\partial \ell}{\partial a_j}$ and $H_p$ is the $p$-th Hermite function. By Stein's lemma \citep{stein1981estimation, mamis2022extension}, it is approximate to a tensor product $\sum_{i=1}^B c_i\bx_i^{\otimes p}$ when $\bw_j$'s are random Gaussian and we can conduct tensor decomposition to recover data, where $B$ is the batch size.
\begin{theorem}[Informal, Theorem 5.1 in \citep{wang2023reconstructing}]
\label{thm:upper short}
    For a 2-layer network under mild assumptions, the reconstruction error of the tensor-based attack has an upper bound $R_U\le \tilde{O}(B\sqrt{\frac{d}{m}})$ with high probability.
\end{theorem}

For the lower bound, to avoid pure combinatorial analysis, we formulate data reconstruction as a statistical estimation problem: we observe the noisy gradient $\nabla_\Theta L(\Theta;S)+\epsilon$ and identify the input data $S$. Here $\epsilon\sim\cN(0,\sigma^2)$ is a Gaussian random vector. We consider the minimax lower bound in Eq. (\ref{eq:lower}), where the expectation is over the random Gaussian noise and the high probability in the following proposition is derived with random weights.
\begin{theorem}[Short version for Theorem \ref{batchedmain}]
\label{thm:batchedmain short}
    For a 2-layer network with noisy gradient, the statistical minimax risk has a lower bound $R_L\ge\tilde{\Omega}(\sigma\sqrt{\frac{d}{m}})$ with high probability.
\end{theorem}
In the following, we analyze the upper and lower bound for the reconstruction under different defenses. The bound is shown in Table \ref{tab:bounds}. Our analysis provides a tight and quantitative analysis of how each defense's defending strength (such as pruning ratio or noise variance) affects the data reconstruction error (which is inversely proportional to the degree of privacy). 

\subsection{Local aggregation}
In a realistic setting of federated learning, the observed local update can be multiple steps of gradient descent \citep{mcmahan2017communication}. 
% Specifically, under the setting of \textit{federated averaging} \citep{mcmahan2017communication}, a client with $N$ images trains $E$ epochs for each update, where a single epoch contains $N/B$ gradient descent steps with batch size $B$. 
In our analysis, we mainly consider the most simple cases, where local devices train two steps and the observation is $\Theta^{(2)}-\Theta^{(0)}$. Since one step of training does not change the parameters too much, we can approximately assume that the two steps share a same set of parameters. Then the analysis is similar to that without defense.
\begin{prop}[Short version for Proposition \ref{prop:twosteps}]
\label{prop:twosteps short}
    For a 2-layer network with mild assumptions, the reconstruction error of tensor based attack under local aggregation defense for 2 steps has a upper bound $R_U\le \tilde{O}(B\sqrt{\frac{d}{m}})$ with high probability.
\end{prop}
\begin{prop}[Short version for Proposition \ref{prop:aggre lower bound}]
\label{prop: aggre lower bound short}
    For a 2-layer network with noisy gradient and local aggregation defense with two steps, the statistical minimax risk has a lower bound $R_L\ge\tilde{\Omega}(\sigma\sqrt{\frac{d}{m}})$ with high probability.
\end{prop}

\subsection{Differential privacy stochastic gradient descent}
Differential privacy \citep{dwork2006differential} algorithms hides the exact local gradients by adding noise to gradient descent steps \citep{shokri2015privacy,abadi2016deep,song2013stochastic}.
% In this way, it hides the exact local gradient from the global model. 
In differential private federated learning \citep{geyer2017differentially,wei2020federated}, each client clips the gradient and introduces a random Gaussian noise before updates to the cluster. In the setting of differential privacy stochastic gradient descent (DP-SGD), the update of $a_j$ is
$$\Tilde{G}_j=G_j/\max\left\{1,\frac{\|G\|}{C}\right\}+ \epsilon_0,
\text{ where }\epsilon_0\sim \cN(0,\sigma_0^2 I).$$

Here $C$ is a constant threshold for gradient clipping.
Note that under this setting, the information contained in the gradient is disrupted by gradient clipping and random noise. 
% For the upper bound, we first analyze the effect of the two strategies separately and then combine them to bound the error of DP-SGD. 
% \begin{proposition}[Short version for Proposition \ref{prop:clipping}]
% \label{prop:clipping short}
%     The observation is clipped gradient descent steps $\tilde{g}$ with clipping threshold $S$. Then under mild assumptions, we can reconstruct input data with the error bound
% \begin{equation}
% \sqrt{\frac{1}{B}\sum_{i=1}^B\|\bx_i-\hat\bx_i\|^2 } \leq  \tilde O(BK\sqrt{\frac{d}{m}} )
% \end{equation}
% with high probability.
% \end{proposition}
% \begin{proposition}[Short version for Proposition \ref{prop:noise}]
% \label{prop:noise short}
%     The observation is noisy gradient descent steps $g+\epsilon$ with $\epsilon\sim\cN(0,\sigma^2I)$, where $\sigma\le O(\frac{B}{m})$. Then under mild assumptions, we can reconstruct input data with the error bound
%     \begin{equation}
% \sqrt{\frac{1}{B}\sum_{i=1}^B\|\bx_i-\hat\bx_i\|^2 } \leq  \tilde O(BK\sqrt{\frac{d}{m}} )\
% \end{equation}
% \end{proposition}
For the upper bound, the noisy gradient $G+\epsilon_0$ will lead to a $\tilde{O}((B+\sigma_0)\sqrt{\frac{d}{m}})$ error bound, which is the same order to original case if $\sigma_0\le O(\frac{B}{m})$. Gradient clipping will enhance the strength noise though itself has no defensive effect.
% If we conduct the tensor-based attack on clipped gradient $\tilde{G}$, there is no defensive effect no matter the clipping threshold and the reconstruction error bound $\tilde O(B\sqrt{\frac{d}{m}} )$ is the same as that with no defense (even for the hidden terms). For noisy gradient $G+\epsilon_0$, we have the error bound $\tilde{O}((B+\sigma_0)\sqrt{\frac{d}{m}})$, which is the same order to original case if $\sigma_0\le O(\frac{B}{m})$. 
% By combining two defenses, gradient clipping enhances the performance of gradient noise by increasing the relative scaling of the noise. The error bound is $O(\max\{1,\frac{\|g\|}{S}\})$ times to the bound of gradient noise defense only (see Appendix \ref{sec:DP proof}).
\begin{prop}[Short version for Proposition \ref{prop:DP}]
\label{prop:DP short}
    For a 2-layer network with mild assumptions, the reconstruction error of tensor-based attack under defense DP-SGD with clipping threshold $S$ and Gaussian noise with variance $\sigma_0^2$ has a upper bound $R_U\le \tilde{O}((B+\sigma_0\max\{1,\frac{\|G\|}{C}\})\sqrt{\frac{d}{m}})$ with high probability.
\end{prop}

In the analysis of statistical lower bound, there is already a Gaussian noise in the observed gradient so we only need to analyze the effect of gradient clipping. Similarly, the clipping changes the scaling of the random noise and the bound changes to $\tilde{\Omega}(\sigma\max\{1,\frac{\|G\|}{C}\}\sqrt{\frac{d}{m}})$. Therefore, gradient noise will be effective if the scaling of noise is large. However, this hurts the utility a lot.

% In the defense of differential privacy gradient descent, we combine the two strategies and conduct a tensor method based on $\tilde{g}+\epsilon$. The defense is more effective and has an error bound $O(\max\{1,\frac{\|g\|}{S}\})$ times to the bound of gradient noise defense only. Gradient clipping enhances the defensive performance of the gradient noise though itself has no effect in defending the attack. However, our attack method is still enough to recover the feature since we can modify the initial weights to control $\|g\|$.
% \begin{theorem}[Short version for Proposition \ref{thm:DP}]
% \label{thm:DP short}
%     The observation is clipped noisy gradient descent steps $\tilde{g}+\epsilon$ with clipping threshold $S$ and $\epsilon\sim\cN(0,\sigma^2I)$, where $\sigma\le O(\frac{B}{m})$. Then under mild assumptions, we can reconstruct input data with the error bound
%     \begin{equation}
% \sqrt{\frac{1}{B}\sum_{i=1}^B\|\bx_i-\hat\bx_i\|^2 } \leq  \tilde O(\max\{1,\frac{\|g\|}{S}\}B\sqrt{\frac{d}{m}} )\
% \end{equation}
% with high probability.
% \end{theorem}

% In conclusion, with DP-SGD, differential privacy gradient descent can be successful only when the noise level is relatively large. However, the noise variance $\sigma^2$ in the application is usually small since a large random noise will cause bad performance in federated learning \citep{wei2020federated}. We also show some results in Section \ref{sec:experiments}. See detailed proofs of the error bounds in Appendix \ref{sec:DP proof}.

\subsection{Secure aggregation}
In federated learning, local devices update gradients to the server individually. Secure aggregation \citep{bonawitz2016practical,bonawitz2017practical} is a method that clients can aggregate their gradients before the cluster, and the cluster can only access to aggregated gradient:
$\Tilde{G}=\frac{1}{B}\sum_{l=1}^L G_l B_l,$
where $G_l$ and $B_l$ is the gradient update and batch size for $l$-th user respectively and $B=\sum_{l=1}^L B_l$.
In this way, the global server is blocked from knowing each gradient, which prevents privacy leakage. Though there is no extra defensive strength compared with directly reconstructing a large batch with size $B$, it cannot identify which user has the data.
% If the batch size is to large, say $B>d$, feature recovery cannot work since the tensor decomposition will not be unique.

\subsection{Dropout}
Another method to improve the defending effect is \textit{dropout} \citep{hinton2012improving,srivastava2014dropout}, a mechanism designed to prevent overfitting. It has been empirically discovered in \citep{zhu2019deep,huang2021evaluating} that it can defend against reconstruction attack. It randomly drops nodes in a network with a certain probability of $1-p$. In this way, some of the entries in the gradient turn to zero, introducing randomness to the gradient and prevent data leakage to some degree.
% The effect of dropout against gradient matching is empirically discussed in \citep{zhu2019deep,huang2021evaluating}.
% A two-layer fully connected network with a dropout layer can be written as
% $$f(\bx;\Theta)=\sum_{j=1}^m s_j a_j\sigma(\bw_j\cdot \bx)=\sum_{j:s_j=1}a_j \sigma(\bw_j\cdot\bx),$$
% where $s_j\sim \mathrm{Bernoulli}(p)$. Thus, the model has an effective width equal to the number of nodes $m'\approx (1-p)m$ that has not been dropped. 
A two-layer fully connected network with a dropout layer can be seen as a model with an effective width equal to the number of nodes $m'\approx (1-p)m$ that has not been dropped. Then $m$ changes into $(1-p)m$ for both upper and lower bounds.
% We can compute the sum $\frac{1}{m'}\sum_{j:s_j=1}g_j H_p(\bw_j)$ and conduct tensor decomposition.
% Since the effective width $m'\approx pm$, the defense is \textbf{not effective} when $p$ is not too small.

\subsection{Gradient pruning}
Gradient pruning \citep{sun2017meprop,ye2020accelerating} is a method that accelerates the computation in training by setting the small entries in gradient to zero. Thus, the gradient will be sparse, which is similar to dropout. However, the key difference is that the remaining entries of the gradient in dropout are chosen randomly yet gradient pruning drops small entries. With dropping nodes by the pruning rules, the distribution of effective weights is not Gaussian, which violates a key assumption in upper bound analysis. Therefore, there is no evidence showing that the tensor-based method can reconstruct input data with small errors under gradient pruning defense. In contrary, we can prove a lower bound with the same order when there is no defense.
\begin{prop}[Short version for Proposition \ref{prop:prune}]
\label{prop:prune short}
    For a 2-layer network with noisy gradient and gradient pruning defense with pruning ratio $p$, the statistical minimax risk has a lower bound $R_L\ge\tilde{\Omega}(\sigma\sqrt{\frac{d}{(1-\hat p)m}})$ with high probability. Let $J=\nabla_{\bx_1,\cdots \bx_B}\nabla_\Theta L$. $\hat p=\|J_{\symbol{92} p}\|_F^2/\|J\|_F^2 \in [0,1]$, where $J_{\symbol{92} p}$ is $J$ dropping the columns corresponding to the pruned coordinates. 
\end{prop}
No upper bound derived indicated gradient pruning is more effective for the tensor-based data reconstruction method. The lower bound also demonstrates its potential for better privacy-utility trade-off, at least compared to dropout. Note that gradient pruning only mildly hurts utility by pruning out the least important coordinates; it does not necessarily hurt reconstruction error in the same way.  %This analysis indicates that gradient pruning may be an effective defense against data reconstruction.

% If we denote the set of entries that is not pruned by $S$ and conduct the tensor-based attack, we have $$\hat{\bT}_p=\frac{1}{m'}\sum_{j\in S}g_j H_p(\bw_j),$$
% where $g_j=\nabla_{a_j}L$ and $m'=|S|$. This formula has a similar form to the case in dropout but the key difference is that the remaining $\bw_j$'s are more likely to make $g_j$ larger so they do not follow the normal distribution. In this case, the concentration property $\sum_{j\in S}g(\bw_j)H_p(\bw_j)\approx m'\E_{\bw\sim \cN(0,I)}[g(\bw)H_p(\bw)]$ may not hold. Moreover, Lemma \ref{lemma:stein} cannot be applied with $\bw_j$'s are not normal random vectors so $\E_{\bw\sim \cN(0,I)}[g(\bw)H_p(\bw)]\neq \E_\bw[\nabla^p_\bw g(\bw)]$. Thus, the $\hat{\bT}_p$ is not a good estimation of the real tensor product, indicating gradient pruning defense is \textbf{effective} against feature reconstruction.

%% file: background.tex
\begin{table*}[t]
\centering
\resizebox{0.99\linewidth}{!}{
\begin{tabular}{cc|ccc|lll|lll|cc}
\toprule
 &  & \textbf{} & \textbf{Gradprune} &  &  & \textbf{GradClipping} &  & \multicolumn{1}{c}{\textbf{}} & \multicolumn{1}{c}{\textbf{GradNoise}} & &\multicolumn{2}{c}{\textbf{Local Aggregation}} \\ \hline
\textbf{Batchsize} & \textbf{w/o defense} & \textbf{$p=0.3$} & \textbf{$p=0.5$} & \textbf{$p=0.7$} & \multicolumn{1}{c}{\textbf{$C=2$}} & \multicolumn{1}{c}{\textbf{$C=4$}} & \multicolumn{1}{c|}{\textbf{$C=8$}} & \multicolumn{1}{c}{{0.001}} & \multicolumn{1}{c}{{0.01}} & \multicolumn{1}{c|}{{0.1}} & \multicolumn{1}{c}{{step=3}} & \multicolumn{1}{c}{{step=5}}\\ \cline{1-8} \cline{9-13} 
1 & 0.171 & 0.173 (0.002) & 0.173 (0.002) & 0.248 (0.006) & 0.173 (0.002) & 0.173 (0.002) & 0.173 (0.002) & 0.172 (0.003) & 0.174 (0.002) & 0.259 (0.012) & 0.181 (0.067) & 0.213 (0.071) \\
2 & 0.171 & 0.171 (0.003) & 0.218 (0.011) & 0.443 (0.120) & 0.169 (0.003) & 0.169 (0.003) & 0.170 (0.003) & 0.166 (0.003) & 0.252 (0.082) & 0.850 (0.108) & 0.186 (0.127)  & 0.210 (0.087)\\
4 & 0.174 & 0.186 (0.114) & 0.277 (0.106) & 0.483 (0.075) & 0.174 (0.116) & 0.175 (0.127) & 0.177 (0.113) & 0.190 (0.021) & 0.252 (0.106) & 0.714 (0.102) & 0.192 (0.079)  & 0.218 (0.081) \\
8 & 0.175 & 0.188 (0.044) & 0.266 (0.071) & 0.425 (0.043) & 0.179 (0.041) & 0.179 (0.041) & 0.179 (0.041) & 0.198 (0.114) & 0.499 (0.104) & 0.953 (0.108) & 0.202 (0.073)  & 0.223 (0.069)\\ \bottomrule
\end{tabular}}
\caption{\small Feature reconstruction error under different defenses. When using gradient pruning defense and gradient noise defense with large $\sigma$, the feature reconstruction quality degrades significantly.} 
\label{tab:feature_summary}
\end{table*}
To conduct a fair comparison between different defense $D$, we evaluate their strength in response to the strongest attack: $\mathcal{S}_D:=\max_{A\in\cA}d(S,A(D(G)))$, where $\cA$ is the set of attacking methods considered in this evaluation. The ideal case is for $\cA$ to contain the strongest possible attack; to complement existent attacks, we propose a new attack demonstrated strong for various defenses. We also include a range of attacks that are either related to our proposed attack method or widely acknowledged to be strong and robust to different defenses:
\begin{itemize}
    \item \textbf{GradientInversion} \citep{geiping2020inverting}. We consider the attack proposed by \citep{geiping2020inverting}, which achieves good recovery results using gradient inversion with known BatchNorm statistics, as knowing such statistics often results in better recovery results and thus stronger attack~\citep{huang2021evaluating}. 
    \item \textbf{Our proposed attack}. Our proposed attack is an improved version of GradientInversion attack. There are two modules we introduced to the original GradientInversion framework. a) a randomly initialized image prior network is introduced to generate the images, and b) the latent features of the fully connected layer are recovered first based on~\citep{wang2023reconstructing}, and then inverting these features using feature inversion to recover the input images. More details will be introduced in Section~\ref{sec:method}.
    \item \textbf{CPA} \citep{kariyappa2023cocktail}. We consider this attack method due to its similarity to our method: GradientInversion and feature inversion are both adopted to recover the input images. In contrast to our proposed attack, CPA forms the feature recovery problem as a blind source separation problem and recovers features by finding the unmixing matrix.
    \item \textbf{Robbing The Fed} \citep{fowl2021robbing}. As both the proposed method and CPA require minimal malicious modifications of the shared model, we are also interested in attacks on explicitly malicious servers.  
\end{itemize}
We apply various defense techniques to the above attack algorithms, and for each defense, we consider the best attack performance across different attack algorithms following Eq. (\ref{eq:empirical}).

%% file: method.tex
% {\red QL: move a brief intro of feature reconstruction part to the theoretical section. 
% Make this section much shorter; no need to repeat anything appeared in previous work. only emphasize the hardness to combine feature reconstruction to gradient inversion in a single query. }

Notice that the tensor-based reconstruction attack is provably strong, with a matching lower bound in the setting of two-layer neural networks. However, the original design in \citep{wang2023reconstructing} was impractical, and we propose an attack that incorporates the strength of the tensor-based method and the generality of gradient inversion attacks for a comprehensive evaluation of different defenses.
On top of the gradient matching term $\mathcal{L}_{\rm grad}$ in \citep{geiping2020inverting} and the prior knowledge utilized in \citep{yin2021see} grounded in batch normalization, the key innovation of our method is an additional regularization that integrates feature reconstruction to gradient inversion. However, the integration is challenging due to different parameter requirements for gradient inversion attacks and tensor-based feature reconstruction. %This novel incorporation seeks to strength the attack strategy, ensuring a more precise and efficient evaluation of defense methods.

\paragraph{Regularization on feature matching.\;} 
By treating the intermediate (last but one) layer $\z_i$ as the model input, we approximately reconstruct $\hat z_i$ as a preprocessing step using the tensor-based method in~\citep{wang2023reconstructing}. However, the reconstructed quality, the ordering of $\hat z_i$, or whether it is unique is unclear. 
To address such challenges, we introduce a gradient-matching regularization term to align intermediate network features with reconstructed feature $\hat z_i$. This regularization, denoted as \(\mathcal{R}_{\text{feature}}(x)\), employs squared cosine similarity to mitigate the issue of sign ambiguity inherent in tensor-based reconstruction:
%\begin{equation}
$\mathcal{R}_{\rm feature}({x}) = 1-\left( \frac{\langle f({x},\phi), \hat{z}\rangle}{\left\|f({x},\phi)\right\|\left\| \hat{z}\right\|}\right)^2,$
%\label{eq:feature_reg}
%\end{equation}
where \(f(x,\phi)\) are the intermediate features of input \(x\) and \(\hat{z}\) are the reconstructed features from the tensor-based method. This approach aims to refine gradient matching by ensuring the generated input \(x\) yields intermediate features similar to \(\hat{z}\), providing a more accurate and specific solution space. We notice the reconstruction quality of individual samples drops drastically when they fall into the same class. However, the reconstruction of their spanned space is still accurate. We further refine the regularize as $\mathcal{R}_{\rm feature}({x})=\|P_{\text{span}(\hat z_i)}^\perp f(x,\phi)\|^2 $ when the sample size is large.

% issues from the sign of the features
% to encourage the intermediate features $f(\hat{x},\phi)$ be similar to recovered feature $\hat{x}$ from the tensor based method. we square the cosine-similarity. 

% To make it compatible with gradient inversion \SL{should we mention the half weights random, half weights set to make it compatible gradient inversion? Good side: make our contribution less trivial, bad side, make people feel our method is strange. }
\paragraph{Other regularizations.\;}
Following ~\citep{geiping2020inverting}, we adopt total variation and regularization on batch norm to exploit prior knowledge of natural images and training data. Moreover, inspired by deep image prior \citep{ulyanov2018deep}, we employ an untrained Convolutional Neural Network (CNN) to generate an image, which can serve as a sufficient image prior. This idea is rather similar to~\citep{jeon2021gradient} where a network is pre-trained for image generation. The final objective of the introduced attack method is
\begin{align*}
    \arg \min _{x}&  \mathcal{L}_{\rm grad}(x; \theta, \nabla_{\theta} \mathcal{L}_{\theta}(x^{*}, y^{*})) 
     + \alpha_{TV} \mathcal{R}_{\rm TV} \\
     &+ \alpha_{BN} \mathcal{R}_{\rm BN}
     + \alpha_{f} \mathcal{R}_{\rm feature}(x) .
    \label{eq:final_objective}
\end{align*}

%% file: experiments_arxiv.tex
\section{Experiments}
\label{sec:experiments}

\begin{figure*}[t]
    \centering
  \includegraphics[width=0.99\linewidth]{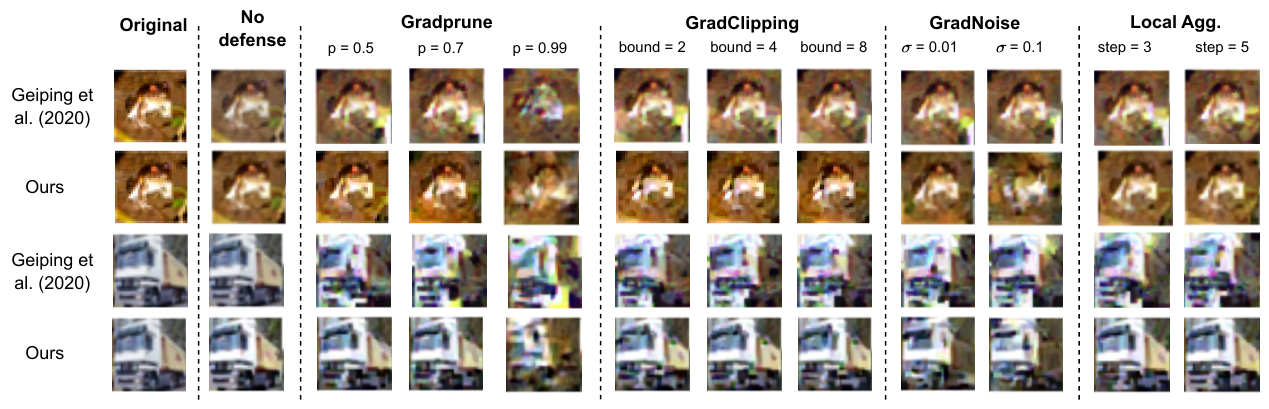}
    \caption{\small Comparison of the reconstruction results from the gradient inversion method~\citep{geiping2020inverting} and our proposed attack method on different defenses with batch size equal to 4. Our method achieves more robust reconstruction across various defenses. Gradient pruning ($p=0.99$) makes reconstructions from both methods almost unrecognizable.}
    \label{fig:attack result}
\end{figure*}

\subsection{Experiment setup}

\paragraph{Key parameters of defenses.\;} We evaluate defense methods on CIFAR-10 dataset ~\citep{krizhevsky2009learning} with a ResNet-18 (trained for one epoch), which is the default backbone model for federated learning. The details are in the Appendix~\ref{sec:exp_setup}. For \textit{GradPrune}: gradient pruning sets gradients of small magnitudes to zero. We vary the pruning ratio $p$ in $\{0.3, 0.5, 0.7, 0.9, 0.99\}$. For \textit{GradClipping}: gradient clipping set gradient to have a bounded norm, which is ensured by applying a clipping
operation that shrinks individual model gradients when their norm exceeds a given threshold. We set this threshold as $\{2, 4, 8\}$. We also adopt defense by \textit{adding noise to gradient}, the noise is generated from random Gaussian with different standard deviations $\{0.001, 0.01, 0.05, 0.1 \}$. For \textit{GradDropout}, we consider randomly dropout the entries of the gradient with probability in $\{0.3, 0.5,0.7, 0.9\}$. Local gradient aggregation aggregates local gradient by performing a few steps of gradient descent to update the model. We locally aggregate $3$ or $5$ steps.

\paragraph{Key experimental setup for the attack.\; }
We use a stratified sampled subset of 50 CIFAR-10 images to evaluate the attack performance. For all attack algorithms, we assume the batch norm statistics are available following~\citep{geiping2020inverting}. We observe the difficulty of optimization when using cosine-similarity as the gradient inversion loss for the proposed attack algorithm and therefore address this issue by reweighting the gradients by their $\ell_0$ norm. See details in the Appendix~\ref{sec:exp_setup}.~\citealp{geiping2020inverting, zhu2019deep} have shown that small batch size is vital for the success of the attack. We intentionally evaluate the attack with two small batch sizes $4$ and $8$ and two small but realistic batch sizes, 16 and 32.

\paragraph{Metrics for reconstruction quality.\;} We visualize reconstructions obtained under different defenses. Following~\citep{yin2021see}, we also use the RMSE, PSNR, and learned perceptual image patch similarity (LPIPS) score~\citep{zhang2018unreasonable} to measure the mismatch between reconstruction and original images. For the evaluation of feature reconstruction, we use the average norm of projection from original features to the orthogonal complement of the space spanned by reconstructed features. This metric measures the difference between the two spaces spanned by original and reconstructed features. Compared with the average cosine similarity between features, this metric is more robust when some features are similar.
\vspace{3mm}
\subsection{Experiment results}

% Feature reconstruction error
% $\|P_{\hat V}^{\perp} V\|_F$

\paragraph{Verification of theoretical analysis. } In Section \ref{sec:theory}, we analyze the performance of different defenses on 2-layer networks. In the experiments, our attack method involves intermediate feature matching, which requires reconstructing feature from a 2-layer network. This setting aligns with our theoretical analysis. Therefore, we can verify the theory by checking the feature reconstruction error in our attack. The results is shown in Table \ref{tab:feature_summary}. The increase of the reconstruction error along with the batch size matches the theoretical bounds. Moreover, gradient pruning and gradient noise with large scaling has significantly large reconstruction error, which also aligns with our theoretical evaluation. 
% For feature recovery, we compared the reconstruction error under different defending strategies, and the results are shown in Table \ref{tab:feature_summary}. We also use the final training loss of the federated learning model to measure the level of interference from defenses toward the training (details see Table \ref{tab:final loss} in Appendix~\ref{sec:add_results}). For an effective defense, it should increase the feature reconstruction error while disturbing little to the training. We show the relation between reconstruction error and final training loss with different defenses in Figure \ref{fig:enter-label}. Gradient pruning has a large reconstruction error without increasing training loss too much so it has the best defending effect under this criteria. Adding noise also prevents feature leakage when noise is large enough but it will affect training. Gradient clipping is barely effective for defending feature recovery while dropout has a mild effect but it is harmful to training. The feature reconstruction errors are consistent with the theoretical analysis in Section \ref{sec:defend}.

\begin{table}[] 
\centering
  \setlength{\tabcolsep}{2.8pt}
  \renewcommand{\arraystretch}{0.95}
  \resizebox{0.78\columnwidth}{!}{
  \begin{tabular}{l|c|c|c|c|c|c}
  \toprule
   Batch Size & \multicolumn{3}{c|}{{\bf 16 }} & \multicolumn{3}{c}{{\bf 32 }}\\
  \midrule
   {\bf Method}  & Geiping & Jeon & Ours & Geiping& Jeon & Ours\\
   \midrule
    {\bf LPIPS $\downarrow$}  & 0.41 (0.09)  & 0.17 (0.12)   & \textbf{0.14 (0.11)}  & 0.45 (0.11)  & 0.24 (0.13)  & \textbf{0.15 (0.11)} \\
  \bottomrule
  \end{tabular}}
  \caption{\small Comparison of our methods with other methods, Geiping, and Jeon refer to \citep{geiping2020inverting} and \citep{jeon2021gradient} respectively. We highlight the best performances in bold.}
  \label{tab:sota}
  \end{table}

\paragraph{Our proposed attack.}
~According to the framework of evaluation, we need to consider strong attacks against various defenses. A key motivation of our proposed attack is its good performance among different defense, which we will verify in the following results.
Our method is based on gradient inversion and can be easily added to previous methods~\citep{geiping2020inverting, zhu2019deep}. In Table ~\ref{tab:sota}, we compare the state-of-the-art gradient inversion methods with the proposed attacking method, our method outperforms previous methods. We visualize reconstructed images in Figure~\ref{fig:attack result}. Without defense, both~\citep{geiping2020inverting} and our attack method can recover images well from the gradient, and our method produces higher-quality images. 
% Increasing the batch size makes all methods more difficult to recover well. 
With defenses, our method shows better robustness -- the recovered images are
visually more similar to the original image. See Table~\ref{tab:w/ feature} and Table~\ref{tab:w/o feature} in Appendix~\ref{sec:add_results} for summary of reconstruction results.

% This result shows that our attack is a stronger attack under these defense methods. Figure~\ref{fig:defense} suggests that the feature matching loss is the main reason behind such robustness, showing that the analysis of feature reconstruction is reasonable. 

\begin{table*}[t]
    \scriptsize
    \centering
  \setlength{\tabcolsep}{2.8pt}
  \renewcommand{\arraystretch}{0.95}
  \resizebox{.99\textwidth}{!}{
\begin{tabular}{c|c|c|c|c|c|c|c|c|c|c|c|c|c|c|c|c}
\hline
& & \multicolumn{3}{|c|}{\bf GradClip ($C$)} & \multicolumn{4}{c|}{\bf GradDrop ($p$)} & \multicolumn{3}{c|}{\bf GradNoise ($\sigma_0$)} & \multicolumn{5}{c}{\bf GradPrune ($p$)} \\
\hline
Parameter &  & 2 & 4 & 8 & 0.3 & 0.5 & 0.7 & 0.9 & 0.001 & 0.01 & 0.1 & 0.3 & 0.5 & 0.7 & 0.9 & 0.99 \\
\hline
\multirow{4}{*}{\bf $B=2$} & Ours & \textbf{0.17} & \textbf{0.17} & \textbf{0.19} & \textbf{0.16} & \textbf{0.16} & \textbf{0.18} & \textbf{0.18} & \bf0.17 & \bf0.22 & 0.28 & \bf0.16 & \bf0.17 &\bf0.20  & \bf0.26 & \bf0.27 \\
 & GradientInversion & 0.19 & 0.23 & 0.25 & 0.19 & 0.20 & 0.20 & 0.21 & 0.30 & 0.32 & 0.35 & 0.19 & 0.20 & 0.24 & 0.28 & 0.27 \\
 & Robbing The Fed & 0.23 & 0.23 & 0.23 & 0.28 & 0.32 & 0.32 & 0.30 & 0.26 & 0.25 & \bf0.27 & 0.22 & 0.25 & 0.32 & 0.37 & 0.46 \\
 & CPA & 0.21  & 0.22 & 0.23  & 0.23  & 0.22 & 0.22 & 0.24  & 0.22  & 0.23  & 0.28  & 0.21  & 0.23 & 0.29 & 0.28 & 0.32\\
\hline
\multirow{4}{*}{\bf $B=4$} & Ours & \bf0.16 & \bf0.16 & \bf0.16 & \bf0.16 & \bf0.16 & \bf0.19 & \bf0.19 & \bf0.19 & 0.24 & 0.27 & \bf0.16 & \bf0.16 & \bf0.21 & \bf0.27 & \bf0.27 \\
 & GradientInversion & 0.21 & 0.22 & 0.23 & 0.21 & 0.21 & 0.21 & 0.22 & 0.31 & 0.31 & 0.31 & 0.19 & 0.18 & 0.23 & 0.29 & 0.28 \\
 & Robbing The Fed & 0.18 & 0.18 & 0.18 & 0.27 & 0.31 & 0.32 & 0.31 & \bf0.19 & \bf0.19 & \bf0.22 & 0.20 & 0.25 & 0.32 & 0.39 & 0.42 \\
 & CPA & 0.21  & 0.21 & 0.20 &  0.21 & 0.23 & 0.22 & 0.24  & 0.23  & 0.23 & 0.25 & 0.18 & 0.21 & 0.24 & 0.29  & 0.31 \\
\hline
\multirow{4}{*}{\bf $B=8$} & Ours & \bf0.16 & \bf0.16 & \bf0.16 & \bf0.16 & \bf0.17 & \bf0.19 & \bf0.19 & 0.19 & 0.29 & 0.30 & \bf0.15 & \bf0.16 & \bf0.20 & \bf0.29 & \bf0.29 \\
 & GradientInversion & 0.22 & 0.21 & 0.21 & 0.22 & 0.22 & 0.22 & 0.23 & 0.30 & 0.31 & 0.31 & 0.20 & 0.21 & 0.25 & 0.30 & 0.30 \\
 & Robbing The Fed & \bf0.16 & \bf0.16 & \bf0.16 & 0.29 & 0.31 & 0.31 & 0.31 & \bf0.16 & \bf0.16 & 0.25 & 0.22 & 0.30 & 0.30 & 0.36 & 0.43 \\
 & CPA & 0.21  & 0.21 & 0.22 & 0.22 & 0.22 & 0.23 & 0.25 & 0.20 & 0.22  & \bf 0.24 & 0.19 & 0.22 & 0.26 & 0.32 & 0.33 \\
\hline
\end{tabular}}
\caption{\small We run attacks under different defenses with various parameters and batch sizes. For each setting, we compared the attacks and select the one with lowest RMSE (the highlighted numbers). The best attack for each setting represents the degree of data leakage of that specific defense. Compared with the utility loss of this defense, we can systematically evaluate defenses.}
\label{tab:defense}

\end{table*}
% \vspace{8mm}
\begin{table*}[t]
\centering
  \scriptsize
  \setlength{\tabcolsep}{2.8pt}
  \renewcommand{\arraystretch}{0.95}
  \resizebox{.99\textwidth}{!}{
  \begin{tabular}{l|c|c|c|c|c|c|c|c|c|c|c|c|c|c|c}
  \toprule
   & \multicolumn{3}{c|}{\bf GradClip ($C$)} & \multicolumn{4}{c|}{\bf GradDrop ($p$)} &  \multicolumn{3}{c|}{\bf GradNoise ($\sigma_0$)} & \multicolumn{5}{c}{\bf GradPrune ($p$)}   \\
  \midrule
   {\bf Parameter}  & 2 & 4 & 8 & 0.3 & 0.5 & 0.7 & 0.9 & 0.001 & 0.01 & 0.1 & 0.3 & 0.5 & 0.7 & 0.9 & 0.99   \\
   \midrule
   % \multicolumn{11}{c}{\bf Attack batch size $= 1$} \\
   % \midrule
   % {\bf RMSE $\downarrow$}  & 0.16(0.02)  & 0.19  & 0.24 (0.01)  & 0.25 (0.00)  & 0.42 & 0.22 (0.11)  & 0.21 (0.11)  & 0.22 (0.10)  & 0.16 (0.11)  & 0.17 (0.11)  & 0.17 (0.11) \\
   %  {\bf PSNR $\uparrow$}  & 12.18 (0.59)  & 0.19  & 12.01 (0.52)  & 11.58 (0.59)  & 0.42 & 19.55 (6.08) & 19.23 (5.49)  & 18.04 (5.12)  & 23.63 (5.47)  & 23.18 (5.52)  & 23.57 (5.94)\\
   % \midrule
   \multicolumn{14}{c}{\bf Attack batch size $= 2$} \\
   \midrule
   {\bf{Final training loss} $\downarrow$ }  & 0.390 & 0.481  & 0.377  & 0.466  & 0.363  & 0.564  & 0.864 & 0.445 & 0.769 & 1.889 & 0.585 & 0.430 & 0.621 & 0.763 & 1.020 \\
   \midrule
   \multicolumn{14}{c}{\bf Attack batch size $= 4$} \\
   \midrule
   {\bf{Final training loss} $\downarrow$ }  & 0.408  & 0.353  & 0.329  & 0.552  & 0.316  & 0.275  & 0.460 & 0.527 & 0.551 & 1.764 & 0.253 & 0.389 & 0.561 & 0.505 & 0.935 \\
   \midrule
   \multicolumn{14}{c}{\bf Attack batch size $= 8$} \\
   \midrule
   {\bf{Final training loss} $\downarrow$ }  & 0.377  & 0.201  & 0.357  & 0.259  & 0.178  & 0.231  & 0.461 & 0.365 & 0.343 & 1.540 & 0.256 & 0.161 & 0.257 & 0.344 & 0.725 \\
  \bottomrule
  \end{tabular}}
  \caption{\small The final training loss of the model with defenses, which measures the utility loss from defenses. Gradient noise produces much larger interference than gradient pruning.} 
  \label{tab:final loss}
\end{table*}

% \begin{figure}[t]
%     \centering
%     \begin{subfigure}{0.48\linewidth}
%     \centering
%     \includegraphics[width=\textwidth]{figures/effectiveofdefense.png}
%     \caption{Feature reconstruction error vs final training loss.}
%     \label{fig:enter-label}
%     \end{subfigure}
% \hfill
%     \begin{subfigure}{0.49\linewidth}
%     \centering
%     \includegraphics[width=\textwidth]{figures/average_PSNR_different_defense.pdf}
%     \caption{PSNR of our method against different defenses.}
%     \label{fig:defense}
%     \end{subfigure}
    
% % \subfigure[Feature reconstruction error vs final training loss.]{
% %         \centering
% %     \includegraphics[width=0.43\linewidth]{figures/effectiveofdefense.png}
% %     \label{fig:enter-label}
% %     }
% % \hfill
% % \subfigure[PSNR of our method against different defenses.]{
% %     \centering
% %     \includegraphics[width=0.45\linewidth]{figures/average_PSNR_different_defense.pdf}
% %     \label{fig:defense}
% %     }

%     \caption{The results of defenses against our attack. (a) The final training loss measures the level of interference with the training of defenses. The relation between feature error and interference reflects the effect of defenses. (b) PSNR of reconstructed data using the method with and without feature matching. The difference indicates that feature matching improves robustness against defenses.}
%     \label{fig:feature}
% \end{figure}

\paragraph{Systematic evaluation of defense.} 
% \zihan{Focus more on the comparison between defenses.}
In our systematic evaluation of different defenses, we measure their strength with the strongest attack: $\mathcal{S}_D=\max_{A\in\cA}d(S,A(D(G)))$. In order to estimate the maximum reconstruction error, we selected four different attacks with various parameters and batch sizes to conduct the experiment and record the reconstruction RMSE in Table \ref{tab:defense}. For each setting, we select the attack with the smallest RMSE as the strongest attack. Therefore, among all defenses, gradient pruning with large $p$ and gradient noise with large variance $\sigma_0^2$ has the best defensive effect.

In order to evaluate the defenses, another key criteria is the utility loss of the defense methods. Most of the defenses prevent models from data leakage by perturbing the gradients in training, which may hurt the training task itself. In our experiment, we measure the utility by the final loss of the training task and we show it in Table \ref{tab:final loss}. Gradient pruning, as one of the defenses having strongest effect towards attacks, causes smaller utility loss than gradient noise. This result indicates that gradient pruning is a better defense than gradient noise. In Figure \ref{fig:evaluation}, we plot the relation between the reconstruction error and the utility loss, where each dot represents a method under the setting of a specific batch size and parameter of the defense. With the same level of utility loss, gradient pruning has the largest RMSE, indicating it is the best defense to have $\min_{D\in \mathcal{D}_\mathcal{U}}\cS_D$.
\\

\begin{wrapfigure}{r}{0.5\textwidth}
    \centering
    \vspace{-.3cm}
    \scalebox{1}{
        \centering
\includegraphics[width=\columnwidth,height=0.75\columnwidth]{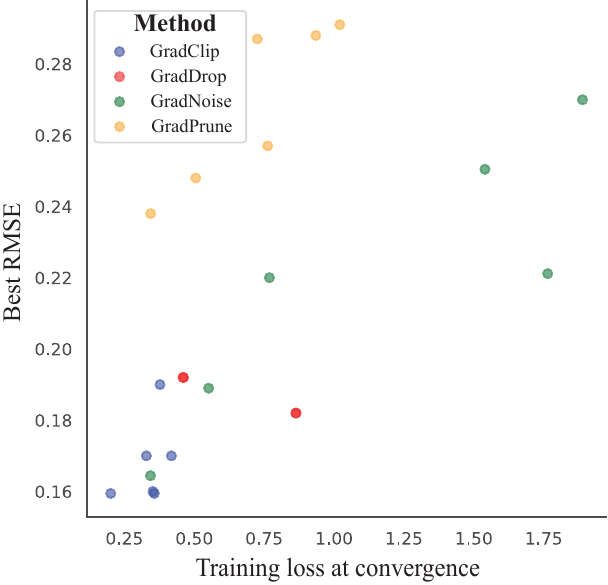}
    \caption{\small The relation between the strength of the defense and the utility loss, where we use the best RMSE and the final training loss of the original task respectively. Each dot represents a defense method with a different batch size and strength. For the same level of utility loss, gradient pruning has the best defending effect.\vspace{-1.5cm} } 
    \label{fig:evaluation} }
        
\end{wrapfigure}
Note that there are various of attacks to reconstruction input data and we cannot try all of them on defenses. However, within our framework of evaluation, newly proposed attacks can be added in the experiment and will improve the strength of defenses $\cS_D$. 
Therefore, we call for more empirical analysis on top of our results with more attack methods, in order to produce a more accurate evaluation.

%% file: results.tex
\newcommand{\best}[1]{\textcolor{green}{\textbf{#1}}}
   % 57  % \captionsetup[table]{font=small}

% \begin{table}
%     \centering
%     \begin{tabular}{c|c|c}
%     \toprule
%         Batch size & \textbf{16} & \textbf{32} \\
%         \midrule
%         \citep{geiping2020inverting} & 0.41(0.09) & 0.45(0.11) \\
%         \midrule
%         \cite{jeon2021gradient} & 0.17(0.12) & 0.24(0.13)\\
%         \midrule
%         Ours & \textbf{0.14(0.11)} & \textbf{0.15(0.11)} \\
%         \bottomrule
%     \end{tabular}
%     \caption{Caption}
%     \label{tab:sota}
% \end{table}

%% file: appendix.tex
\appendix
\input{proofs}

\input{lowerbound}
\input{exp_details}

%% file: proofs.tex
\section{Analysis of Differential Privacy}
Differential privacy (DP) \citep{dwork2006differential} is a measure of privacy mainly in membership inference attack. A random algorithm $\cM: \cD\rightarrow \mathcal{R}$ satisfies $(\epsilon,\delta)$-DP if for any $\mathcal{Q}\subset \mathcal{R}$ and any adjacent inputs $S,S'\in \cD$ holds that
$$
\mathbb{P}(\cM(S)\in \mathcal{Q})\le e^\epsilon\mathbb{P}(\cM(S')\in \mathcal{Q})+\delta.
$$
Here adjacent inputs means only one of the samples is different.

In reconstruction attack, DP is too strong that an algorithm may prevent data leakage without DP guarantees. Moreover, the price to guarantee DP is high. A popular method to guarantee DP is DP-SGD \citep{abadi2016deep}. However, it requires extremely large Gaussian noise.
\begin{proposition}[Full version of Proposition \ref{prop:DP guarantee short}]
\label{prop:DP guarantee}
     We define a two-layer neural network $f(\bx;\Theta)=\sum_{j=1}^ma_j\sigma(\bw_j^\top \bx)$ with input dimension $d$, random Gaussian weights $a_j\sim\cN(0,\frac{1}{m^2})$ and $\bw_j\sim\cN(0,I_d)$ and activation function $\sigma$ is 1-Lipschitz. The loss function $\ell$ is square loss, the data $\|\bx\|=1$ and the label $y\in \{\pm 1\}$. We denote the gradient of $\ell$ by $G$. The randomized mechanism $\cM=G+\cN(0,\sigma^2I)$ is $(\epsilon, \delta)-$ DP for any $\epsilon,\delta>0$ if $\sigma^2=\Omega(\frac{m\log(1/\delta)}{\epsilon})$ with high probability with respect to random weights.
\end{proposition}
In order to prove the proposition, we first define
$$\alpha_\cM(\lambda):=\max_{S,S'}\log\E_{z\sim\cM(S)}\left(\frac{p(\cM(S)=z)}{p(\cM(S')=z)}\right)^\lambda,$$
where $p$ is the density function. In DP-SGD setting, $\cM(S)=G(S)+\cM(0,\sigma^2I)$ so $\cM(S)\sim \cN(G(S),\sigma^2I)$ for $S\in \R^{d\times B}$. The following lemma is a key step in the proof of Proposition \ref{prop:DP guarantee}.
\begin{lemma}[Theorem 2 in \citep{abadi2016deep}]
\label{lemma:DP tail}
    For any $\epsilon>0$, the mechanism $\cM$ is $(\epsilon,\delta)$-DP for $\delta=\min_{\lambda}(\alpha_\cM(\lambda)-\lambda\epsilon)$.
\end{lemma}
\begin{proof}[Proof of Proposition \ref{prop:DP guarantee}]
    We denote $\Delta:=\max_{\bx,\bx'\in \R^d}\|G(\bx)-G)(\bx')\|^2$. Then we compute $\alpha_\cM(\lambda)$ with adjacent $S,S'\in \R^{d\times B}$. We denote the only different element by $\bx$ and $\bx'$.
    \begin{align*}
        \alpha_\cM(\lambda)=&\max_{S,S'}\log\E_{Z\sim\cN(G(S),\sigma^2I)}\left[\left(\frac{p_S(Z)}{p_Y(Z)}\right)^\lambda\right]\\
        =&\max_{\bx,\bx'}\log\E_{\bz\sim\cN(G(\bx),I)}\exp\left(\frac{\lambda}{2\sigma^2}\sum_{i=1}^d\left(\left(G(\bx)_i-z_i\right)^2-\left(G(\bx')_i-z_i\right)^2\right)\right)\\
        =&\max_{\bx,\bx'}\log\frac{1}{\left(\sqrt{2\pi}\sigma\right)^d}\int \exp\left(\frac{1}{2\sigma^2}\sum_{i=1}^d\left((\lambda+1)\left(G(\bx)_i-z_i\right)^2-\lambda\left(G(\bx')_i-z_i\right)^2\right)\right)\\
        =&\max_{\bx,\bx'}\log\frac{1}{\left(\sqrt{2\pi}\sigma\right)^d}\int \exp\left(\frac{1}{2\sigma^2}\sum_{i=1}^d\left(\left(z_i-\left((\lambda+1)G(\bx)_i-\lambda G(\bx')_i\right)\right)^2-\lambda(\lambda+1)\left(G(\bx)_i-G(\bx')_i\right)^2\right)\right)\\
        =&\max_{\bx,\bx'}\log\exp\left(\frac{\lambda(\lambda+1)}{2\sigma^2}\sum_{i=1}^d\left(G(\bx)_i-G(\bx')_i\right)^2\right)\\
        =&\max_{\bx,\bx'}\frac{\lambda(\lambda+1)}{2\sigma^2}\left\|G(\bx)-G(\bx')\right\|^2=\frac{\lambda(\lambda+1)}{w\sigma^2}\Delta.
    \end{align*}
    By Lemma \ref{lemma:DP tail}, for any 
    \begin{align*}
        \delta\ge&\min_{\lambda}\exp\left(\frac{\lambda(\lambda+1)}{2\sigma^2}\Delta-\lambda\epsilon\right)\\
        &=\exp\left(-\frac{\Delta}{2\sigma^2}\left(\frac{\sigma^2\epsilon}{\Delta}-\frac{1}{2}\right)^2\right),
    \end{align*}
    $\cM$ is $(\epsilon,\delta)$-DP. Then for $\sigma^2\ge\Omega(\frac{\Delta\log(1/\delta)}{\epsilon})$, $\cM$ is $(\epsilon,\delta)$-DP.

    Now we only need to bound $\Delta$. We denote $r=y-f(\bx)$ and $r'=y'-f(\bx')$. We consider $\nabla_{a_j}$ and $\nabla_{\bw_j}$ separately. For $\nabla_{a_j}$, we have
    \begin{align*}
        \left|\nabla_{a_j}\ell(S)-\nabla_{a_j}\ell(S')\right|^2=&\left|\nabla_{a_j}\ell(\bx)-\nabla_{a_j}\ell(\bx')\right|^2\\
        =&\left|r\sigma(\bw_j^\top\bx)-r'\sigma(\bw_j^\top\bx)\right|^2\\
        \lesssim&r^2\left|\sigma(\bw_j^\top\bx)\right|^2+r'^2\left|\sigma(\bw_j^\top\bx')\right|^2\\
        =&\tilde{O}(1)
    \end{align*}
    with high probability for any $j$.
    For $\nabla_{\bw_j}$, we have
    \begin{align*}
        \left\|\nabla_{\bw_j}\ell(S)-\nabla_{\bw_j}\ell(S')\right\|^2=&\left\|\nabla_{\bw_j}\ell(\bx)-\nabla_{\bw_j}\ell(\bx')\right\|^2\\
        =&\left\|r\sigma'(\bw_j^\top\bx)\bx-r'\sigma'(\bw_j^\top\bx')\bx'\right\|^2\\
        \lesssim&r^2\left\|\bx\right\|^2+r'^2\left\|\bx'\right\|^2\\
        =&\tilde(O)(1)
    \end{align*}
    with high probability for any $j$.
    Then 
    \begin{align*}
        \Delta&=\max_{S,S'}\left\|G(S)-G(S')\right\|^2\\
        &=\max_{S,S'}\sum_{j=1}^m\left(\left|\nabla_{a_j}\ell(S)-\nabla_{a_j}\ell(S')\right|^2+\left\|\nabla_{\bw_j}\ell(S)-\nabla_{\bw_j}\ell(S')\right\|^2\right)\\
        &\le \Tilde{O}(m)
    \end{align*}
    with high probability. 
    Therefore, for $\sigma^2\ge \Tilde{\Omega}(\frac{m\log(1/\delta)}{\epsilon})$, $\cM$ is $(\epsilon,\delta)$-DP with high probability.
\end{proof}

\section{Analysis of Reconstruction Upper Bound}
\subsection{Tensor Method}
Various of noisy tensor decomposition can be used in tensor based attack \citep{wang2023reconstructing}. We select the method proposed by \citep{zhong2017recovery} in our theoretical analysis since it can provably achieve a relatively small error with least assumptions.
\begin{assumption}
\label{assump:2layer}
We make the following assumptions:
    \begin{itemize}
	\item \textbf{Data:} Let data matrix $S:=[\bx_1,\cdots,\bx_B]\in\mathbb{R}^{d\times B}$, we denote the $B$-th singular value by $\pi_{\min}>0$. Training samples are normalized: $\|\bx_i\|=1,\forall i \in [B]$.
	\item \textbf{Activation:} $\sigma$ is 1-Lipschitz and $\E_{z\sim\cN(0,1)}[\sigma''(z)]<\infty$. Let $$k_2=\min\{k\ge 2:|\E_{z\sim \cN(0,1)}[\sigma^{(k)}(z)]|\neq 0\}$$ and $$k_3=\min\{k\ge 3:|\E_{z\sim \cN(0,1)}[\sigma^{(k)}(z)]|\neq 0\}.$$ Then $\nu= |\E_{z\sim \cN(0,1)}[\sigma^{(k_2)}(z)]|$ and $\lambda=| \E_{z\sim \cN(0,1)}[\sigma^{(k_3)}(z)]|$ are not zero. We assume $k_2\le 3$ and $k_3\le 4.$.
\end{itemize}
\end{assumption}
We first introduce the method in \citep{zhong2017recovery} when the activation function satisfies $k_2=2$ and $k_3=3.$
In addition to $\hat{\bT}=\sum_{j=1}^m g(\bw_j)H_3(\bw_j)$, we also need the matrix $\hat{P}=\sum_{j=1}^m g(\bw_j)H_2(\bw_j)$ in the tensor decomposition, where $H_p$ is the $p$-th Hermite function. We first conduct the power method to $\hat{P}$ to estimate the orthogonal span $U$ of training samples $\bx_1,\dots,\bx_B$ and denote it by $V$, where $V\in \mathbb{R}^{B\times d}$ is an orthogonal matrix. Then we conduct tensor decomposition with $\hat{\bT}(V,V,V)$ instead of $T$ and have the estimation of $\{V^\top \bx_i\}_{i=1}^B$. By multiplying column orthogonal matrix $V$, we can reconstruct training data. The advantage of this method is that the dimension of $\hat{\bT}(V,V,V)$ is $B<d$. Then the error will depend on $B$ instead of $d$.

For activation functions such that $k_2=3$, we define $\hat{P}=\sum_{j=1}^m g(\bw_j)H_3(\bw_j)(I,I,\ba)$ and for activation functions such that $k_3=4$, we define $\hat{\bT}=\sum_{j=1}^m g(\bw_j)H_4(\bw_j)(I,I,I,\ba)$, where $\ba$ is any unit vector. In the proofs below, we only consider the case $k_2=2$ and $k_3=3$ while the proofs of other settings are similar.

Note that with this method only, we cannot identify the norm of recovered samples without the assumption of $\|\bx_i\|=1$ for all $i$. However, in our stronger reconstruction attack, the feature matching term is the cosine similarity between the reconstructed feature and the dummy feature, where knowing the norm is not necessary. Thus, our assumption on $\|bx_i\|=1$ is reasonable and can simplify the method and the proof.

\subsection{Error Bound}
For the noisy tensor decomposition introduced above, we have the following error bound:
\begin{theorem}[Adapted from the proof of Theorem 5.6 in \citep{zhong2017recovery}.]
\label{prop:tensorerror}
    Consider matrix $\hat{P}=P+S$ and tensor $\hat{\bT}=\bT+\bE$ with rank-$B$ decomposition
    $$P=\sum_{i=1}^B\nu_i\bx_i\bx_i^\top,\ \bT=\sum_{i=1}^B\lambda_i\bx_i^{\otimes 3},$$
    where $\bx_i\in \mathbb{R}^d$ satisfying Assumption \ref{assump:2layer}. Let $V$ be the output of Algorithm 3 in \citep{zhong2017recovery} with input $P$ and $\{s_i\bu_i\}_{i=1}^B$ be the output of Algorithm 1 in \citep{kuleshov2015tensor} with input $\bT(V,V,V)$, where $\{s_i\}$ are unknown signs. Suppose the perturbations satisfy
    $$\|S\|\le\mu\le O(\nu_{\min}\pi_{\min}),\ \|\bE(V,V,V)\|\le\gamma.$$
    Let $N=\Theta(\log\frac{1}{\epsilon})$ be the iteration numbers of Algorithm 3 in \citep{zhong2017recovery}, where $\epsilon=\frac{\mu}{\nu_{\min}}$.
    Then with high probability, we have
    $$\left\|\bx_i-s_i V \bu_i\right\|\le \Tilde{O}(\frac{\mu}{\nu_{\min}\pi_{\min}})+\Tilde{O}(\frac{\kappa\gamma\sqrt{B}}{\lambda_{\min}\pi_{\min}^2}),$$
    where $\kappa=\frac{\lambda_{\max}}{\lambda_{\min}}.$
\end{theorem}

With Theorem \ref{prop:tensorerror}, we only need to bound the error $\|P-\hat{P}\|$ and $\|\bT-\hat{\bT}\|$, where the bound varies under different settings. \citealp{wang2023reconstructing} proved the error bound under the square loss setting. In the following section, we will propose error bounds with cross-entropy loss, which is more widely used in classification problems. We will also give error bounds when defending strategies are used against privacy attacks.

\subsection{Matrix Bernstein's Inequality}

The following lemma is crucial to the concentration bounds $\|P-\hat{P}\|$ and $\|\bT-\hat{\bT}\|$.
\begin{lemma}[Matrix Bernstein for unbounded matrices; adapted from Lemma B.7 in \citep{zhong2017recovery}]
\label{bernstein}
Let $\mathcal{Z}$ denote a distribution over $\mathbb{R}^{d_1 \times d_2}$. Let $d=d_1+d_2$. Let $Z_1, Z_2, \cdots, Z_m$ be i.i.d. random matrices sampled from $\mathcal{Z}$. Let $\bar{Z}=\mathbb{E}_{Z \sim \mathcal{Z}}[Z]$ and $\widehat{Z}=\frac{1}{m} \sum_{i=1}^m Z_i$. For parameters $\delta_0 \in(0,1), M=M(\delta_0,m)\ge 0, \nu>0, L>0$,  if the distribution $\mathcal{B}$ satisfies the following four properties,
\begin{center}
$$
\begin{aligned}
        (I)& \quad \mathbb{P}_{Z \sim \mathcal{Z}}\left\{\|Z\| \leq M\right\} \geq 1-\frac{\delta_0}{m}\\
(II)& \quad \max \left(\left\|\underset{Z \sim \mathcal{Z}}{\mathbb{E}}\left[Z Z^{\top}\right]\right\|,\left\|\underset{Z \sim
\mathcal{Z}}{\mathbb{E}}\left[Z^{\top} Z\right]\right\|\right) \leq \nu\\
(III)&  \quad \max _{\|\ba\|=\|\bb\|=1}\left(\underset{Z \sim \mathcal{Z}}{\mathbb{E}}\left[\left(\ba^{\top} Z \bb\right)^2\right]\right)^{1 / 2} \leq L
\end{aligned}
$$
\end{center}

Then we have for any $0<\delta_1<1$, if $\delta_1\le \frac{1}{d}$ and $m\gtrsim M\log(1/\delta_1)$,
with probability at least $1-\delta_1-\delta_0$,
$$
\|\widehat{Z}-\bar{Z}\| \lesssim \sqrt{\frac{\log(1/\delta_1)(\nu+\|\bar{Z}\|^2+M+\delta_0L^2)}{m}}
$$
\end{lemma}

\subsection{Reconstruction Upper Bounds}
\label{sec:feature reconstruction}
\citealp{wang2023reconstructing} proposed a privacy attack based on tensor decomposition. We adopt the setting of that work.
Let the input dimension be $d$  For a two layer neural network $f(\bx;\Theta)=\sum_{j=1}^m a_j\sigma(\bw_j^\top\bx)$ with $m$ hidden nodes and $\sigma$ is point-wise. Then square loss with $B$ samples is $\ell(f(\bx;\Theta),y)=\sum_{i=1}^B (y_i-f(\bx_i;\Theta))^2$. Then we have the gradient 
$$\frac{\partial\ell}{\partial \bw_j}= \sum_{i=1}^B r_ia_j\sigma'(\bw_j\top\bx_i)\bx_i$$
and 
$$\frac{\partial\ell}{\partial a_{j}}=\sum_{i=1}^B r_i\sigma(\bw_j^\top \bx_i),$$
where $r_i=2(f(\bx_i;\Theta)-y_i)$.

To reconstruct training data, we set the weights to be random Gaussian.
\begin{assumption}
\label{assump:weight}
The network $f(\bx;\Theta)=\sum_{j=1}^m a_j\sigma(\bw_j^\top\bx)$.
    The parameters of the network satisfy $a_j\sim\cN(0,\frac{1}{m^2})$ and $\bw_j\sim\cN(0,I_d)$ and are independent.
\end{assumption}
% we can modify the initialization weights arbitrarily. Let $\ba_k=c_k\Tilde{\ba}$, where $c_k$ is a scalar and $\Tilde{\ba}\sim\mathcal{N}(0,I_m)$. Then each rows of $A$ are proportional so we can write $z_k=c_k z$. We can let $c_k=b_1$ for $k\in U $, where $U$ is a subset of $\{1,\dots,K\}$ and $\|U\|=K/2$ (assuming $K>2$ is even) and $c_k=b_2$ for $k\notin {U}$, where $b_1\neq b_2$. Then all $p_k$'s are equal to some $\Tilde{p}$ for $k\in U$ and all $p_k$'s are equal to another $\Tilde{q}$ for $k\notin U$. Without loss of generality, we can assume that $1\in U$ and $2\notin U$, then $c_1=b_1$ and $c_2=b_2$. Note that $$p_1+p_2=\tilde{p}+\Tilde{q}=\frac{e^{b_1 z}+e^{b_2 z}}{\sum_{k\in U}e^{b_1 z}+\sum_{k\notin U}e^{b_2 z}}=\frac{2}{K}$$
% is a constant, no matter the input samples.
% % Without loss of generality, we can assume that $\Tilde{p}<\Tilde{q}$. Thus, $\Tilde{p}=\Tilde{p}\sum_{k=1}^K c_k<\sum_{k=1}^K p_kc_k<\Tilde{q}\sum_{k=1}^Kc_k=\Tilde{q}.$
% If we use the gradient of $A$, we construct $g(\bw_j)$ with the first two classes: 
% \begin{equation}
% \begin{split}
% \label{eq:g-def}
%     g(\bw_j):&=\sum_{i=1}^B \left(\frac{\partial \ell(f(\bx_i;\Theta),\by_i)}{\partial a_{1j}}+\frac{\partial \ell(f(\bx_i;\Theta),\by_i)}{\partial a_{2j}}\right)\\
%     &=\frac{1}{m}\sum_{i=1}^B \sigma(\bw_j^\top \bx_i)(p_1-y_{i1}+p_2-y_{i2})\\
%     &=\frac{1}{m}\sum_{i=1}^B\left(\frac{2}{K}-(y_{i1}+y_{i2})\right) \sigma(\bw_j^\top \bx_i).
%     \end{split}
% \end{equation}
% Let $r_i=\frac{2}{K}-(y_{i1}+y_{i2})$, then $r_i=\frac{2}{K}$ or $\frac{2-K}{K}$, which is non-zero. 
For each $\bw_j$, we define $g(\bw_j)=\sum_{i=1}^B r_i\sigma(\bw_j^\top\bx_i)$.
We can use the same attacking method in \citep{wang2023reconstructing} that conduct noisy tensor decomposition to $\sum_{j=1}^m g(\bw_j)H_3(\bw_j)$, where $H_p$ is the $p$-th Hermite function. In our method, we will mainly use $H_2(\bw)=\bw\bw^\top-I$ and $H_3(\bw)=\bw^{\otimes3}-\bw\tilde{\otimes}I$, where $\bw\tilde{\otimes}I(i,j,k)=w_i\delta_{jk}+w_j\delta{ki}+w_k\delta{ij}$.

\begin{lemma}[Stein's Lemma]
	\label{lemma:stein}
	Let $X$ be a standard normal random variable. Then for any function $g$, we have
	\begin{equation}
	\mathbb{E}[g(X)H_p(X)]=\mathbb{E}[g^{(p)}(X)],
	\end{equation}
 if both sides of the equation exist. Here $H_p$ is the $p$th Hermite function and $g^{(p)}$ is the $p$th derivative of $g$.
\end{lemma} 
We define 
\begin{equation}
\label{eq:start}
    \hat{P}=\frac{1}{m}\sum_{j=1}^m g(\bw_j)H_2(\bw_j)
\end{equation}
\begin{equation}
    P=\E \left[\sum_{i=1}^B r_i^*\sigma''(\bw^\top\bx_i)\bx_i\bx_i^\top\right]
\end{equation}
\begin{equation}
    \hat{\bT}=\frac{1}{m}\sum_{j=1}^m g(\bw_j)H_3(\bw_j)
\end{equation}
\begin{equation}
\label{eq:end}
    \bT=\E \left[\sum_{i=1}^B r_i^*\sigma^{(3)}(\bw^\top\bx_i)\bx_i^{\otimes 3}\right].
\end{equation}
Then with Lemma \ref{lemma:stein} and concentration bounds, we have the following lemmas.
\begin{lemma}[Proposition 5.5 in \citep{wang2023reconstructing}]
\label{lemma:P bound}
    Under Assumption \ref{assump:2layer} and \ref{assump:weight}, $|y_i|\le 1$, then for $\delta\le \frac{2}{d}$ and $m\gtrsim \log(8/\delta)$, we have
    $$\|\hat{P}-P\|
    \le \Tilde{O}(\frac{B\sqrt{d}}{\sqrt{m}})$$
    with probability $1-\delta$.
\end{lemma}
\begin{lemma}[Proposition 5.6 in \citep{wang2023reconstructing}]
\label{lemma:T bound}
    Under Assumption \ref{assump:2layer} and \ref{assump:weight}, $|y_i|\le 1$, and $\|VV^\top-UU^\top\|\le 1/4$, then for $\delta\le\frac{2}{B}$ and $m\gtrsim \log(6/\delta)$ we have
    $$\|\bar{\bT}(V,V,V)-\bT(V,V,V)\|\le\Tilde{O}(\frac{B^{5/2}}{\sqrt{m}})$$
    with probability $1-\delta$.
\end{lemma}
Combining with Theorem \ref{prop:tensorerror}, we have a upper bound for data reconstruction error.
\begin{theorem}[Theorem 5.1 in \citep{wang2023reconstructing}, full version of Theorem \ref{thm:upper short}]
    Under Assumption \ref{assump:2layer} and \ref{assump:weight}, $y_i\in \{\pm 1\}$, if we have $B\le \Tilde{O}(d^{1/4})$ and $m \geq \tilde \Omega(\frac{d}{\min\{\nu^2,\lambda^2\}\pi_{\min}^4}) $, then for $\delta\le \frac{2}{d}$ and $m\gtrsim\log(8/\delta)$, the output of tensor based reconstruction satisfies:
    $$\sqrt{\frac{1}{B}\sum_{i=1}^B\|\bx_i-\hat\bx_i\|^2 } \leq  \frac{1}{\min\{|\nu|,|\lambda|\} \pi_{\min}^2}\tilde O(\sqrt{\frac{d}{m}} )$$
    with probability $1-\delta$.
\end{theorem}

\section{Proofs of Upper Bounds with Defenses}
\input{aggregation}

\subsection{Differential Privacy}
\label{sec:DP proof}
In differential private federated learning, the gradient update for any parameter is 
$$\Tilde{G}=G/\max\left\{1,\frac{\|G\|}{C}\right\}+ \mathcal{E}_G,$$
where $\|G\|$ is the norm of the gradient of all parameters and $\mathcal{E}_g\sim\cN(0,\sigma_0^2 C^2 I_d)$, $d$ is the dimension of parameter. First, we give the proof of the case with no random noise that only gradient clipping is used, i.e. $\tilde{G}=G/\max{\left\{1,\frac{\|G\|}{C}\right\}}$. The error bound is the same as the case with no defense.
\begin{proposition}
\label{prop:clipping}
    Under Assumption \ref{assump:2layer} and Assumption \ref{assump:weight}, $y_i\in\{\pm 1\}$, we observe clipped gradient descent steps $\tilde{G}$ with batch size $B$, where $B^4\le \Tilde{O}(d)$ and clipping threshold $C$. If $m \geq \tilde \Omega(\frac{B^2d}{\nu^2\pi_{\min}^2}) $, then with appropriate tensor decomposition methods and proper weights, we can reconstruct input data with the error bound
$$
\sqrt{\frac{1}{B}\sum_{i=1}^B\|\bx_i-\hat\bx_i\|^2 } \leq  \frac{1}{\min\{|\nu|,|\lambda|\} \pi_{\min}^2}\tilde O(B\sqrt{\frac{d}{m}} )\
$$
with high probability.
\end{proposition}
\begin{proof}
    For any $j$, the observed gradient updates for $a_j$ are $\tilde{g}(\bw_j)=g(\bw_j)/\max\left\{1,\frac{\|G\|}{C}\right\}$.
    We define $\Tilde{P}=\frac{1}{m}\sum_{j=1}^m \tilde{g}(\bw_j)(\bw_j\bw_j^\top-I)$ and $\tilde{\bT}=\frac{1}{m}\sum_{j=1}^m \tilde{g}(\bw_j)(\bw_j^{\otimes 3}-\bw_j \tilde{\otimes}I)$. $\hat{P}$, $\hat{\bT}$, $P$ and $\bT$ are defined in Eq. (\ref{eq:start}) to (\ref{eq:end}). Then let $R=\min\{1,\frac{C}{\|G\|}\}$, we have $\tilde{P}=R\hat{P}$ and $\tilde{\bT}=R\hat{\bT}$. Then by Lemma \ref{lemma:P bound} and \ref{lemma:T bound} we have $$\|\Tilde{P}-RP\|=R\|\hat{P}-P\|\le \Tilde{O}(\frac{RB\sqrt{d}}{\sqrt{m}})$$ and $$\|\Tilde{\bT}(V,V,V)-R\bT(V,V,V)\|=R\|\hat{\bT}(V,V,V)-\bT(V,V,V)\|\le\tilde{O}(\frac{RB^{5/2}}{\sqrt{m}})$$
    with high probability.
    
    On the other hand, the smallest components of $RP$ and $R\bT$ are $R\nu_{\min}$ and $R\lambda_{\min}$ respectively. We have $\nu_{\min}\ge|\nu|$, $\lambda_{\min}\ge|\lambda|$ and $\kappa=1$. Then by Theorem \ref{prop:tensorerror}, 
    $$
    \|\bx_i-\hat{\bx}_i\|\le \tilde{O}(\frac{RB\sqrt{d/m}}{R|\mu|\pi_{\min}^2})+\tilde{O}(\frac{R \sqrt{B^6/m}}{R|\lambda|\pi_{\min}^2})\le\frac{1}{\min\{|\nu|,|\lambda|\}\pi_{\min}^2} \tilde{O}(B\sqrt{\frac{d}{m}})
    $$
    for all $i=1,\dots,B$ with high probability.
\end{proof}

Then we give the formal statement and proof of the error bound of the case without gradient clipping that the only defense is random noise, i.e. $\tilde{G}=G+\mathcal{E}_G$, where $\mathcal{E}_G\sim\cN(0,\sigma_0^2I)$.
\begin{proposition}
\label{prop:noise}
    Under Assumption \ref{assump:2layer} and Assumption \ref{assump:weight}, $y_i\in\{\pm 1\}$, we observe noisy gradient descent steps $G+\epsilon_0$ with batch size $B$, where $K^2B^4\le \Tilde{O}(d)$ and $\epsilon\sim\cN(0,\sigma_0^2I)$. If $m \geq \tilde \Omega(\frac{B^2d}{\nu^2\pi_{\min}^2}) $, then with appropriate tensor decomposition methods and proper weights, we can reconstruct input data with the error bound
    $$
\sqrt{\frac{1}{B}\sum_{i=1}^B\|\bx_i-\hat\bx_i\|^2 } \leq  \frac{1}{\min\{|\nu|,|\lambda|\} \pi_{\min}^2}\tilde O((B+\sigma_0)\sqrt{\frac{d}{m}} )\
$$
with high probability.
\end{proposition}
\begin{proof}
     For any $j$, the observed gradient updates for $a_j$ are $\tilde{g}(\bw_j)=g(\bw_j)+\epsilon_j$, where $\epsilon_j\sim\cN(0,\sigma^2)$.
    We define $\Tilde{P}=\frac{1}{m}\sum_{j=1}^m \tilde{g}(\bw_j)(\bw_j\bw_j^\top-I)$ and $\tilde{\bT}=\frac{1}{m}\sum_{j=1}^m \tilde{g}(\bw_j)(\bw_j^{\otimes 3}-\bw_j \tilde{\otimes}I)$. $\hat{P}$, $\hat{\bT}$, $P$ and $\bT$ are defined in Eq. (\ref{eq:start}) to (\ref{eq:end}). Then we have $\Tilde{P}=\hat{P}+\frac{1}{m}\sum_{j=1}^m\epsilon_j(\bw_j\bw_j^\top-I)$ and $\Tilde{\bT}=\hat{\bT}+\frac{1}{m}\sum_{j=1}^m \epsilon_j(\bw_j^{\otimes3}+\bw_j\tilde{\otimes}I)$. 
    Now we bound $\|\Tilde{P}-P\|$ and $\|\Tilde{\bT}(V,V,V)-\bT(V,V,V)\|$. 
    
    \textbf{Error bound of $\Tilde{P}$.} We have $\|\Tilde{P}-P\|\le \|\hat{P}-P\|+\|P_\epsilon\|,$ where $P_\epsilon=\frac{1}{m}\sum_{j=1}^m \epsilon_j(\bw_j\bw_j^\top-I)$. Since $\|\hat{P}-P\|\le \tilde{O}(\frac{B\sqrt{d}}{\sqrt{m}})$ with high probability by Lemma \ref{lemma:P bound}, we only have to bound $\|P_\epsilon\|.$ We let $Z_j=(\epsilon_{1j}+\epsilon_{2j})(\bw_j\bw_j^\top-I)$ and check the conditions of Theorem \ref{bernstein}: 

    (I) We first bound the norm of $Z_j$: $$\left\|Z_j\right\|\le\left|\epsilon_j\right|\left\|(\bw_j\bw_j^\top-I)\right\|\lesssim\sigma_0 d(\log(16m/\delta))^2$$
    with probability $1-\frac{\delta}{4m}$.

    (II) We have
    $$
    \max\left\{\left\|\mathbb{E}[Z^\top Z]\right\|,\left\|\mathbb{E}[ZZ^\top]\right\|\right\}=\left\|\E\epsilon^2 \E[(\bw\bw^\top-I)^2]\right\|.
    $$
    Let $Q=(\bw\bw^\top-I)^2$. Since $Q_{ij}=\sum_{k\neq i,j} w_iw_jw_k^2+(w_i^2+w_j^2-2)w_iw_j$ for $i\neq j$ and $Q_{ii}=\sum_{k\neq i} w_i^2w_k^2+(w_i^2-1)^2$, $\E(\bw\bw^\top-I)^2=(d+1)I$. Then $\max\left\{\left\|\mathbb{E}[Z^\top Z]\right\|,\left\|\mathbb{E}[ZZ^\top]\right\|\right\}\le O(\sigma_0^2d).$

    (III) For $\max_{\|\ba\|=\|\bb\|=1}\left(\E\left[(\ba^\top Z\bb)^2\right]\right)^{1/2}$, it reaches the maximal when $\ba=\bb$ since $Z$ is a symmetric matrix. Thus, we have
    \begin{equation}
\begin{split}
    \mathbb{E}(\ba^\top (\bw\bw^\top-I)Z \ba)^2&=\mathbb{E}\left(\sum_{i=1}^d a_i^2(w_i^2-1)+\sum_{i\neq j}a_ia_jw_iw_j\right)^2\\
    &=\mathbb{E}\left[\sum_{i=1}^d a_i^4(w_i^2-1)^2+\sum_{i\neq j}a_i^2a_j^2w_i^2w_j^2\right]\\
    &=2\sum_{i=1}^d a_i^4+\sum_{i\neq j}a_i^2a_j^2\\
    &\le 2 \left(\sum_{i=1}^d a_i^2\right)^2=2.
\end{split}
\end{equation}
    Then, $\max_{\|\ba\|=\|\bb\|=1}\left(\E\left[(\ba^\top Z\bb)^2\right]\right)^{1/2}=\le O(\sigma_0^2).$

    Moreover, $\E[Z]=0.$ Then by Theorem \ref{bernstein},
    $$
    \left\|\frac{1}{m}\sum_{j=1}^mZ_j\right\|\le\log(16m/\delta)\sqrt{\frac{\sigma_0^2 d\log(4/\delta)}{m}}\le \tilde{O}(\sigma_0\sqrt{\frac{d}{m}})
    $$
    with probability $1-\frac{\delta}{2}$. Thus, 
    $$
    \|\tilde{P}-P\|\le \tilde{O}((B+\sigma_0)\sqrt{\frac{d}{m}})
    $$
    with probability $1-\delta$.

    \textbf{Error bound of $\tilde{\bT}(V,V,V)$.} We have $\|\tilde{\bT}(V,V,V)-\bT(V,V,V)\|\le \|\hat{\bT}(V,V,V)-\bT(V,V,V)\|+\|\bT_\epsilon(V,V,V)\|,$ where $\bT_\epsilon=\frac{1}{m}\sum_{j=1}^m \epsilon_j(\bw_j^{\otimes3}-\bw_j\tilde{\otimes}I)$. Since $\|\hat{\bT}(V,V,V)-\bT(V,V,V)\|\le \tilde{O}(\frac{B^{5/2}}{\sqrt{m}})$ with high probability by Lemma \ref{lemma:T bound}, we only need to bound $\|\bT_\epsilon(V,V,V)\|$. Note that $\|\bT_\epsilon\|\le \|T_\epsilon^{(1)}\|$, where $T_\epsilon^{(1)}$ is the flatten of $\bT_\epsilon$ along the first dimension, so we can bound $\|T_\epsilon^{(1)}(V,V,V)\|$ instead. We check the conditions of Theorem \ref{bernstein}. For any $\bx$
    
    (I) We have
    $$
\left\|Z_j\right\|\lesssim|\epsilon_j|\|V^\top\bw_j\|^3\le\Tilde{O}(\frac{\sigma_0 B^{3/2}}{\sqrt{m}})
$$
with probability $1-\frac{\delta}{4m}$.

(II) We have
$$
    \max\left\{\left\|\mathbb{E}\left[Z^\top Z\right]\right\|,\left\|\mathbb{E}\left[ZZ^\top\right]\right\|\right\}\le\mathbb{E}\left[\left\|Z\right\|^2\right]\lesssim \mathbb{E}\left[\epsilon^2\right]\mathbb{E}\left[\|V^\top\bw_j\|^6\right]\lesssim \sigma_0^2 B^3.
$$

(III) We have
$$
    \max _{\|\ba\|=\|\bb\|=1}\left({\mathbb{E}}\left[\left(\ba^{\top} Z \bb\right)^2\right]\right)^{1 / 2}\le \left(\mathbb{E}\left[\left\|Z\right\|^2\right]\right)^{1/2}\lesssim \sigma_0 B^{3/2}.
$$

Moreover, $\|\E[Z]\|=0$.
Then by Theorem \ref{bernstein}, we have for any $i$ that
$$
\|\frac{1}{m}\sum_{j=1}^m \bT^{(1)}_\epsilon(V,V,V)\|\le \tilde{O}(\frac{\sigma_0 B^{3/2}}{\sqrt{m}})
$$
with probability $1-\frac{\delta}{2}$. Thus, 
$$
\|\tilde{\bT}(V,V,V)-\bT(V,V,V)\|\le \tilde{O}((B+\sigma_0)\sqrt{\frac{B^3}{m}})
$$
with probability $1-\delta$.

Since $\nu_{\min}$, $\lambda_{\min}$ and $\kappa$ are not changed, by Theorem \ref{prop:tensorerror}, 
$$
\|\bx_i-\hat{\bx}_i\|\le \frac{K}{\min\{|\nu|,|\lambda|\}\pi_{\min}^2}\tilde{O}((B+\sigma_0)\sqrt{\frac{d}{m}})
$$
for all $i=1,\dots,B$ with high probability.
\end{proof}

Then we consider the case where both gradient clipping and gradient noise are used in the training. In this case, we found that gradient clipping will increase the error caused by gradient noise though clipping itself has no effect on the error bound. Here we reparameterize the observed gradient as $\tilde{G}=G/\max\left\{1,\frac{\|G\|}{C}\right\}+\epsilon_0$, where $\epsilon_0\sim\cN(0,\sigma^2I)$.
\begin{proposition}[Full version of Proposition \ref{prop:DP short}]
\label{prop:DP}
    Under Assumption \ref{assump:2layer} and Assumption \ref{assump:weight}, $y_i\in\{\pm 1\}$, we observe clipped noisy gradient descent steps $G/\max\left\{1,\frac{\|G\|}{C}\right\}+\epsilon$ with batch size $B$, where $ B^4\le \Tilde{O}(d)$ and $\epsilon\sim\cN(0,\sigma^2I)$. If $m \geq \tilde \Omega(\frac{ B^2d}{\nu^2\pi_{\min}^2}) $, then with appropriate tensor decomposition methods and proper weights, we can reconstruct input data with the error bound
    $$
\sqrt{\frac{1}{B}\sum_{i=1}^B\|\bx_i-\hat\bx_i\|^2 } \leq  \frac{1}{\min\{|\nu|,|\lambda|\} \pi_{\min}^2}\tilde O((B+\sigma_0\max\{1,\frac{\|G\|}{C}\})\sqrt{\frac{d}{m}} )\
$$
with high probability.
\end{proposition}

\begin{proof}
    For any $j$, the observed gradient updates for $a_j$ are $\tilde{g}(\bw_j)=g(\bw_j)/\max\left\{1,\frac{\|g\|}{C}\right\}+\epsilon_j$, where $\epsilon_j\sim\cN(0,\sigma_0^2)$. 
    We define $\Tilde{P}=\frac{1}{m}\sum_{j=1}^m \tilde{g}(\bw_j)(\bw_j\bw_j^\top-I)$ and $\tilde{\bT}=\frac{1}{m}\sum_{j=1}^m \tilde{g}(\bw_j)(\bw_j^{\otimes 3}-\bw_j \tilde{\otimes}I)$. $\hat{P}$, $\hat{\bT}$, $P$ and $\bT$ are defined in Eq. (\ref{eq:start}) to (\ref{eq:end}). Let $R=\min\{1,\frac{C}{\|G\|}\}$. Then we have $\Tilde{P}=R\hat{P}+\frac{1}{m}\sum_{j=1}^m \epsilon_j(\bw_j\bw_j^\top-I)$ and $\Tilde{\bT}=R\hat{\bT}+\frac{1}{m}\sum_{j=1}^m  \epsilon_j(\bw_j^{\otimes3}+\bw_j\tilde{\otimes}I)$. By the proof of Proposition \ref{prop:clipping} and Proposition \ref{prop:noise}, we have
    \begin{equation}
        \|\Tilde{P}-RP\|\le R\|\hat{P}-P\|+\|P_\epsilon\|\le \tilde{O}((RB+\sigma_0)\frac{d}{m})
    \end{equation}
    and
    \begin{equation}
    \begin{split}
        \|\Tilde{\bT}(V,V,V)-R\bT(V,V,V)\|&\le R\|\hat{\bT}(V,V,V)-\bT(V,V,V)\|+\|\bT_\epsilon(V,V,V)\|\\
        &\le \tilde{O}((RB+\sigma_0)\sqrt{\frac{B^3}{m}}).
    \end{split}
    \end{equation}
    Note that the smallest components of $RP$ and $R\bT$ are $R\nu_{\min}$ and $R\lambda_{\min}$ respectively. Then by Theorem \ref{prop:tensorerror}, 
    \begin{align*}
    \|\bx_i-\hat{\bx}_i\|&\le \tilde{O}(\frac{(RB+\sigma_0)\sqrt{d/m}}{R|\mu|\pi_{\min}^2/K})+\tilde{O}(\frac{(RB+\sigma_0)K\sqrt{B^4/m}}{R|\lambda|\pi_{\min}^2/K})\\
    &\le\frac{K}{\min\{|\nu|,|\lambda|\}\pi_{\min}^2} \tilde{O}((B+\sigma_0\max\{1,\frac{\|G\|}{C}\})\sqrt{\frac{d}{m}})
    \end{align*}
    for all $i=1,\dots,B$ with high probability.
\end{proof}

% \subsection{Dropout}
% When Dropout is used in the training, only a subset of weights are used. Since we can know exactly from gradient update which nodes are dropped, we can use the gradients of the nodes without dropping and know the exact number of these nodes. Denote $\hat{P}=\frac{m}{m'}\sum_{j=1}^m g(\bw_j)H_2(\bw_j)$ and $\hat{\bT}=\frac{m}{m'}\sum_{j=1}^m g(\bw_j)H_3(\bw_j)$. Then the error $\|P-\hat{P}\|\le \tilde{O}(\frac{B\sqrt{d}}{\sqrt{m'}})$ and $\|\bT-\hat{\bT}\|\le \tilde{O}(\frac{B\sqrt{d}}{\sqrt{m'}})$.

%% file: aggregation.tex
\subsection{Local Aggregation}
\label{sec:aggregation proof}
For local aggregation, the observation is $\Theta^{(2)}-\Theta^{(0)}$. In our attack method, we conduct tensor decomposition to $\tilde{g}=-\sum_{j=1}^m(a_j^{(2)}-a_j^{(0)})H_3(\bw_j)$. Here a notation with superscript $(k)$ means the corresponding value after $k$ steps of iteration. Then we have the following reconstruction error bounds under the setting of two gradient descent steps with a same batch of input and with two different batches input.
\begin{proposition}[Full version of Proposition \ref{prop:twosteps short}]
\label{prop:twosteps}
    Under Assumption \ref{assump:2layer} and \ref{assump:weight}, $y_i\in \{\pm 1\}$, the observed update $\Theta^{(2)}-\Theta^{(0)}$ is the result of 2 gradient descent steps trained a same batch with size $B$, where $B^4\le \Tilde{O}(d)$. If $m \geq \tilde \Omega(\frac{B^2d}{\nu^2\pi_{\min}^2}) $ and we set learning rate for two layers with different scales $\eta_a=O(\frac{1}{m^2})$ and $\eta_w=O(1)$, then with appropriate tensor decomposition methods and proper weights, we can reconstruct input data with the error bound:
\begin{equation}
\sqrt{\frac{1}{B}\sum_{i=1}^B\|\bx_i-\hat\bx_i\|^2 } \leq  \frac{1}{\min\{|\nu|,|\lambda|\} \pi_{\min}^2}\tilde O(B\sqrt{\frac{d}{m}} )\
\end{equation}
with high probability.
\end{proposition}

\begin{proposition}
\label{prop:diff batch}
    Under Assumption \ref{assump:2layer} and \ref{assump:weight}, $y_i\in \{\pm 1\}$, the observed update $\Theta^{(2)}-\Theta^{(0)}$ is the result of 2 gradient descent steps, where the first step is trained with $B$ data and the second step is trained with $N-B$ data. Here $B<N\le 2B$ and $B^4\le \Tilde{O}(d)$. If $m \geq \tilde \Omega(\frac{B^2d}{\nu^2\pi_{\min}^2}) $ and we set learning rate for two layers with different scales $\eta_a=O(\frac{1}{m^2})$ and $\eta_w=O(1)$, then with appropriate tensor decomposition methods and proper weights, we can reconstruct input data with the error bound:
\begin{equation}
\sqrt{\frac{1}{N}\sum_{i=1}^N\|\bx_i-\hat\bx_i\|^2 } \leq  \frac{1}{\min\{|\nu|,|\lambda|\} \pi_{\min}^2}\tilde O(N\sqrt{\frac{d}{m}} )\
\end{equation}
with high probability.
\end{proposition}

\begin{remark}
    In these propositions we assume that the learning rate for the two layers are different since $a_j\sim\cN(0,\frac{1}{m^2})$ but $\bw_j\sim\cN(0,I_d)$ have different scaling. If we use a single learning rate $\eta$ for all parameters, the results still hold if we assume $\eta=O(\frac{1}{m^2})$.
\end{remark}

To prove the propositions, we have to verify that the parameters after a gradient descent step will not change too much from the original ones. We first introduce some lemmas.

The first lemma show that $\bw^{(1)}$ has a similar effect to $\bw$ in the reconstruction.
\begin{lemma}
    \label{lemma:sigma}
    If the learning rate $\bw$ $\eta_w=O(1)$, we have
    \begin{equation}
        |\sigma(\bw_j^{(1)}\cdot \bx_i)-\sigma(\bw_j^\top \bx_i)|\le\tilde{O}(\frac{B}{m})
    \end{equation}
    for any $i$ and $j$ with probability $1-\frac{\delta}{Bm}$.
\end{lemma}
\begin{proof}
    Note that $\bw_j^{(1)}=\bw_j-{\eta_w}\sum_{i=1}^B a_{j}\sigma'(\bw_j^\top \bx_i)\bx_i$. Then 
    \begin{equation}
    \label{eq:sigma}
        \begin{split}
            &|\sigma(\bw_j^{(1)}\cdot \bx_i)-\sigma(\bw_j^\top \bx_i)|\\
            =&\left|\sigma\left(\bw_j^\top\bx_i-\eta_w\sum_{i=1}^B a_{j}\sigma'(\bw_j^\top \bx_i)\right)-\sigma(\bw_j^\top\bx_i)\right|\\
            \le &\eta_w\sum_{i=1}^Ba_{j}\sigma'(\bw_j^\top\bx_i)\le\tilde{O}(\frac{B}{m})
        \end{split}
    \end{equation}
    with probability $1-\frac{\delta}{Bm}$ for all $i,\ j$. 
\end{proof}

Then we give lemmas showing that $r_i^{(1)}$ is very close to $r_i$ and the difference only changes a little in the reconstruction.
\begin{lemma}
    \label{lemma:z}
    If the learning rate $\eta_w=O(1)$ and $\eta_a=O(\frac{1}{m^2})$, we have
    \begin{equation}
        \left|r_{i}^{(1)}-r_{i}\right|\le \tilde{O}(\frac{B}{m})
    \end{equation}
    with probability $1-\frac{\delta}{Bm}$ for any $i$.
\end{lemma}
\begin{proof}
    We have $r_i=\sum_{j=1}^m a_{j}\sigma(\bw_j^\top\bx_i)-y_i$ and $r_i^{(1)}=\sum_{j=1}^m a_{j}^{(1)}\sigma(\bw_j^{(1)}\cdot\bx_i)-y_i$. Note that $a_{kj}^{(1)}=a_{kj}-\eta_a\sum_{i=1}^B\sigma(\bw_j^\top\bx_i)$. Then by Lemma \ref{lemma:sigma} we have
    \begin{equation}
        \begin{split}
            \left|r_i^{(1)}-r_i\right|&=\left|\sum_{j=1}^m a_{j}\sigma(\bw_j^\top\bx_i)-\sum_{j=1}^m \left(a_{j}-{\eta_a}\sum_{i=1}^B\sigma(\bw_j^\top\bx_i)\right)\sigma(\bw_j^{(1)}\cdot\bx_i)\right|\\
            &\le \sum_{j=1}^m a_{j}\left|\sigma(\bw_j^\top\bx_i)-\sigma(\bw_j^{(1)}\cdot\bx_i)\right|+\eta_a\sum_{j=1}^m\sum_{i=1}^B(\sigma(\bw_j^\top\bx_i))^2\\
            &\quad\quad\quad\quad+{\eta_a}\sum_{j=1}^m\sum_{i=1}^B|\sigma(\bw_j^\top\bx_i)||\sigma(\bw_j^\top\bx_i)-\sigma(\bw_j^{(1)}\cdot\bx_i)|\\
            &\le \tilde{O}(\frac{B}{m})+\tilde{O}(\frac{B\log(2Bm^2/\delta)}{m})+\tilde{O}(\frac{B\log(2Bm^2/\delta)}{m^2})\le \tilde{O}(\frac{B}{m})
        \end{split}
    \end{equation}
    with probability $1-\frac{\delta}{Bm}$.
\end{proof}

\begin{corollary}
\label{cor:r_i}
    If the learning rate $\eta_w=O(1)$ and $\eta_a=O(\frac{1}{m^2})$, we have $r_i^{(1)}=\tilde{O}(1)$ with probability $1-\frac{\delta}{Bm}$.
\end{corollary}
\begin{proof}
    We have $r_i=\tilde{O}(1)$ with probability $1-\frac{\delta}{Bm}$. Then the result holds directly by Lemma \ref{lemma:z}.
\end{proof}

\begin{lemma}
\label{lemma:P1}
If the learning rate $\eta_w=O(1)$ and $\eta_a=O(\frac{1}{m^2})$, we have
    \begin{equation}
        \left\|\frac{1}{m}\sum_{j=1}^m r_i^{(1)}\left|\sigma(\bw_j^{(1)}\cdot\bx_i)-\sigma(\bw_j^\top\bx_i)\right|(\bw_j\bw_j^\top-I)\right\|\le\tilde{O}(\frac{Bd}{m})
    \end{equation}
    with probability $1-\frac{\delta}{2B}$ for any $i$.
\end{lemma}
\begin{proof}
    We define $Z_j=r_i^{(1)}|\sigma(\bw_j^{(1)}\cdot\bx_i)-\sigma(\bw_j^\top\bx_i)|(\bw_j\bw_j^\top-I)$, then we have
    $$
    \|\frac{1}{m}\sum_{j=1}^m Z_j\|\le \|\E Z_j\|+\|\frac{1}{m}\sum_{j=1}^m Z_j-\E Z_j\|.
    $$

    We first bound $\|\E Z_j\|$.
    By Lemma \ref{lemma:sigma} and Corollary \ref{cor:r_i} we have
    \begin{equation}
        \label{eq:P mean norm}
    \|\E Z_j\|\lesssim\E |\sigma(\bw_j^{(1)}\cdot \bx_i)-\sigma(\bw_j^\top \bx_i)|\|\bw_j\bw_j^\top-I\|\lesssim\frac{Bd}{m}.
    \end{equation}

    Next we bound $\|\frac{1}{m}\sum_{j=1}^m Z_j-\E Z_j\|$ with matrix Bernstein inequality. We first check the conditions of Theorem \ref{bernstein} with Lemma \ref{lemma:sigma} and Corollary \ref{cor:r_i}

    (I) We first bound the norm of $Z_j$. BWe have
    $$
    \|Z_j\|\le 2|\sigma(\bw_j^{(1)}\cdot \bx_i)-\sigma(\bw_j^\top \bx_i)|(\|\bw_j\|^2+1)\lesssim \frac{Bd\log(4Bm/\delta)}{m}
    $$
    with probability $1-\frac{\delta}{2Bm}$.

    (II) We have
\begin{align*}
    \max&\left\{\left\|\mathbb{E}\left[Z^\top Z\right]\right\|,\left\|\mathbb{E}\left[Z Z^\top\right]\right\|\right\}=\left\|\mathbb{E}\left[Z^2\right]\right\|\\
        &\lesssim \E |\sigma(\bw^{(1)}\cdot \bx_i)-\sigma(\bw^\top \bx_i)|^2\|(\bw\bw^\top-I)^2\|\lesssim \frac{B^2d^2}{m^2}.
\end{align*}

(III) For $\max_{\|\ba\|=\|\bb\|=1}(\mathbb{E}(\ba^\top Z \bb)^2)^{1/2}$, it reaches the maximal when $\ba=\bb$ since $Z$ is a symmetric matrix. Thus, we have
$$
    \mathbb{E}(\ba^\top Z \ba)^2=\mathbb{E}|\sigma(\bw^{(1)}\cdot \bx_i)-\sigma(\bw^\top \bx_i)|^2(\ba^\top(\bw\bw^\top-I)\ba)^2\lesssim\frac{B^2d^2}{m^2}.
$$
Then $\max_{\|\ba\|=\|\bb\|=1}(\mathbb{E}(\ba^\top Z \bb)^2)^{1/2}\lesssim\frac{Bd}{m}$.

Moreover, $\|\E Z_j\|\lesssim\frac{Bd}{m}$ by Eq. \eqref{eq:P mean norm}. Then by Theorem \ref{bernstein}, we have
        $$
        \|\frac{1}{m}\sum_{j=1}^m Z_j-\E Z_j\|\le\tilde{O}(\frac{Bd}{m})
        $$
        with probability $1-\frac{\delta}{2B}$.
\end{proof}

% \begin{lemma}
% \label{lemma:P1}
% If the learning rate $\eta=O(1)$ and $m\ge\Omega(K^2B^2d)$ we have
%     \begin{equation}
%     \begin{split}
%         &\left\|\frac{1}{m}\sum_{j=1}^m\sum_{i=1}^B\left[(p^{(1)}_{1i}-y_{1i})+(p^{(1)}_{2i}-y_{2i})\right]\left|\sigma(\bw_j^{(1)}\cdot\bx_i)-\sigma(\bw_j^\top\bx_i)\right|\left(\bw_j\bw_j^\top-I\right)\right\|\\
%         &\quad\quad\quad\quad\quad\quad\le\tilde{O}(B\sqrt{\frac{d}{m}})
%     \end{split}
%     \end{equation}
%     with probability $1-\frac{\delta}{2}$.
% \end{lemma}
% \begin{proof}
%     The result comes directly after Lemma \ref{lemma:P1 pre}.
% \end{proof}
\begin{lemma}
\label{lemma:T1}
If the learning rate $\eta_w=O(1)$ and $\eta_a=O(\frac{1}{m^2})$, we have
    \begin{equation}
        \left\|\frac{1}{m}\sum_{j=1}^mr_i^{(1)}\left|\sigma(\bw_j^{(1)}\cdot\bx_i)-\sigma(\bw_j^\top\bx_i)\right|(\bw_j^{\otimes 3}-\bw_j\tilde{\otimes}I)(V,V,V)\right\|\le\tilde{O}(\frac{B^{5/2}}{m})
    \end{equation}
    with probability $1-\frac{\delta}{2B}$ for any $i$.
\end{lemma}
\begin{proof}
    We define $\tilde{\bZ}_j=r_i^{(1)}|\sigma(\bw_j^{(1)}\cdot\bx_i)-\sigma(\bw_j^\top\bx_i)|(\bw_j^{\otimes 3}-\bw_j\tilde{\otimes}I)(V,V,V)$ and $Z_j$ be the flatten of $\tilde{\bZ}_j$ along the first dimension. Then we have
    $$
    \|\frac{1}{m}\sum_{j=1}^m \tilde{\bZ}_j\|\le \|\frac{1}{m}\sum_{j=1}^m Z_j\|\le \|\E Z_j\|+\|\frac{1}{m}\sum_{j=1}^m Z_j-\E Z_j\|.
    $$

    We first bound $\E Z_j$. By Lemma \ref{lemma:sigma} and Corollary \ref{cor:r_i}, we have
    \begin{equation}
        \label{eq:T mean norm}
        \|\E Z_j\| \lesssim \E |\sigma(\bw_j^{(1)}\cdot \bx_i)-\sigma(\bw_j^\top \bx_i)|\|V^\top \bw_j\|^3\lesssim\frac{B^{5/2}}{m}.
    \end{equation}

    Next we bound $\|\frac{1}{m}\sum_{j=1}^m Z_j-\E Z_j\|$ with matrix Bernstein inequality. We first check the conditions of Theorem \ref{bernstein} with Lemma \ref{lemma:sigma} and Corollary \ref{cor:r_i}:

    (I) We first bound the norm of $Z_j$. We have
    $$
    \|Z_j\|\lesssim 2|\sigma(\bw_j^{(1)}\cdot \bx_i)-\sigma(\bw_j^\top \bx_i)|\|V^\top \bw_j\|^3\le\tilde{O}(\frac{B^{5/2}}{m})
    $$
    with probability $1-\frac{\delta}{2Bm}$.

    (II) We have
\begin{align*}
    \max&\left\{\left\|\mathbb{E}\left[Z^\top Z\right]\right\|,\left\|\mathbb{E}\left[Z Z^\top\right]\right\|\right\}\le\mathbb{E}\left[\left\|Z\right\|^2\right]\\
        &\lesssim \E |\sigma(\bw^{(1)}\cdot \bx_i)-\sigma(\bw^\top \bx_i)|^2\|V^\top\bw\|^6\lesssim \frac{B^5}{m^2}.
\end{align*}

(III) We have
$$
    \max_{\|\ba\|=\|\bb\|=1}(\mathbb{E}(\ba^\top Z \bb)^2)^{1/2}\le\left(\mathbb{E}\left[\left\|Z\right\|^2\right]\right)^{1/2}\lesssim\frac{B^{5/2}}{m}.
$$

Moreover, $\|\E Z_j\|\lesssim\frac{B^{5/2}}{m}$ by Eq. \eqref{eq:T mean norm}. Then by Theorem \ref{bernstein}, we have
        $$
        \|\frac{1}{m}\sum_{j=1}^m Z_j-\E Z_j\|\le\tilde{O}(\frac{B^{5/2}}{m})
        $$
        with probability $1-\frac{\delta}{2B}$.
\end{proof}

\begin{lemma}
    \label{lemma:P2}
    If the learning rate $\eta_w=O(1)$ and $\eta_a=O(\frac{1}{m^2})$, and $m\ge \tilde{\Omega}(B)$, we have
    \begin{equation}
        \left\|\frac{1}{m}\sum_{j=1}^m\left(r_i^{(1)}-r_i\right)\sigma(\bw_j^\top\bx_i)(\bw_j\bw_j^\top-I)\right\|\le \tilde{O}(\frac{Bd}{m})
    \end{equation}
    with probability $1-\frac{\delta}{2B}$ for any $i$.
\end{lemma}
\begin{proof}
    We define $Z_j=(r_i^{(1)}-r_i)\sigma(\bw_j^\top\bx_i)(\bw_j\bw_j^\top-I)$, then we have
    $$
    \|\frac{1}{m}\sum_{j=1}^m Z_j\|\le \|\E Z_j\|+\|\frac{1}{m}\sum_{j=1}^m Z_j-\E Z_j\|.
    $$

    We first bound $\|\E Z_j\|$. By Lemma \ref{lemma:z} we have
    \begin{equation}
        \label{eq:P mean norm 2}
        \begin{split}
    \|\E Z_j\|&\le\E |r_i^{(1)}-r_i| \|\sigma(\bw_j^\top \bx_i)(\bw_j\bw_j^\top-I)\|\\
    &\le \left(\E \left[|r_i^{(1)}-r_i|^2\right] \right)^{1/2} \left(\E \|\sigma(\bw_j^\top \bx_i)(\bw_j\bw_j^\top-I)\|^2\right)^{1/2}\\
    &\le \tilde{O}(\frac{Bd}{m}).
        \end{split}
    \end{equation}

    Next we bound $\|\frac{1}{m}\sum_{j=1}^m Z_j-\E Z_j\|$ with matrix Bernstein inequality. We first check the conditions of Theorem \ref{bernstein}:

    (I) We first bound the norm of $Z_j$. By Lemma \ref{lemma:z}, we have
    $$
    \|Z_j\|\le |r_i^{(1)}-r_i||\sigma(\bw_j^\top \bx_i)|(\|\bw_j\|^2+1)\le\tilde{O}(\frac{Bd}{m})
    $$
    with probability $1-\frac{\delta}{2Bm}$.

    (II) By Lemma \ref{lemma:z} we have
\begin{align*}
    \max&\left\{\left\|\mathbb{E}\left[Z^\top Z\right]\right\|,\left\|\mathbb{E}\left[Z Z^\top\right]\right\|\right\}=\left\|\mathbb{E}\left[Z^2\right]\right\|\\
    &\le \left(\E \left[|r_i^{(1)}-r_i|^4\right]\right)^{1/2} \left(\E\left\|\sigma(\bw^\top \bx_i)(\bw\bw^\top-I)\right\|^4\right)^{1/2}\\
    &\le\tilde{O}(\frac{B^2d^2}{m^2}).
\end{align*}

(III) For $\max_{\|\ba\|=\|\bb\|=1}(\mathbb{E}(\ba^\top Z \bb)^2)^{1/2}$, it reaches the maximal when $\ba=\bb$ since $Z$ is a symmetric matrix. Thus, we have
\begin{align*}
    \mathbb{E}(\ba^\top Z \ba)^2&=\mathbb{E}|r_i^{(1)}-r_i|^2(\ba^\top\sigma(\bw^\top\bx_i)(\bw\bw^\top-I)\ba)^2\\
    &\le \left(\E \left[|r_i^{(1)}-r_i|^4\right]\right)^{1/2} \left(\E\left\|\sigma(\bw^\top \bx_i)(\bw\bw^\top-I)\right\|^4\right)^{1/2}\\
    &\lesssim\frac{B^2d^2}{m^2}.
\end{align*}
Then $\max_{\|\ba\|=\|\bb\|=1}(\mathbb{E}(\ba^\top Z \bb)^2)^{1/2}\lesssim\frac{Bd}{m}$.

Moreover, $\|\E Z_j\|\lesssim\frac{Bd}{m}$ by Eq. \eqref{eq:P mean norm 2}. Then by Theorem \ref{bernstein}, we have
        $$
        \|\frac{1}{m}\sum_{j=1}^m Z_j-\E Z_j\|\le\tilde{O}(\frac{Bd}{m})
        $$
        with probability $1-\frac{\delta}{2B}$.
\end{proof}

\begin{lemma}
    \label{lemma:T2}
    If the learning rate $\eta_w=O(1)$ and $\eta_a=O(\frac{1}{m^2})$, and $m\ge \tilde{\Omega}(B)$, we have
    \begin{equation}
        \left\|\frac{1}{m}\sum_{j=1}^m\left(r_i^{(1)}-r_i\right)\sigma(\bw_j^\top\bx_i)(\bw_j^{\otimes3}-\bw_j\tilde{\otimes}I)(V,V,V)\right\|\le \tilde{O}(\frac{B^{5/2}}{m})
    \end{equation}
    with probability $1-\frac{\delta}{2B}$ for any $i$.
\end{lemma}
\begin{proof}
    We define $\tilde{\bZ_j}=(r_i^{(1)}-r_i)\sigma(\bw_j^\top\bx_i)(\bw_j^{\otimes3}-\bw_j\tilde{\otimes}I)(V,V,V)$ and $Z_j$ be the flatten of $\tilde{\bZ}_j$ along the first dimension. Then we have
    $$
    \|\frac{1}{m}\sum_{j=1}^m \tilde{\bZ}_j\|\le \|\frac{1}{m}\sum_{j=1}^m Z_j\|\le \|\E Z_j\|+\|\frac{1}{m}\sum_{j=1}^m Z_j-\E Z_j\|.
    $$

    We first bound $\|\E Z_j\|$. By Lemma \ref{lemma:z} we have
    \begin{equation}
        \label{eq:T mean norm 2}
        \begin{split}
    \|\E Z_j\|&\lesssim\E |r_i^{(1)}-r_i| |\sigma(\bw_j^\top \bx_i)|\|V^\top \bw_j\|^3\\
    &\le \left(\E \left[|r_i^{(1)}-r_i|^2\right] \right)^{1/2} \left(\E |\sigma(\bw_j^\top \bx_i)|^2\|V^\top\bw_j\|^6\right)^{1/2}\\
    &\le \tilde{O}(\frac{B^{5/2}}{m}).
        \end{split}
    \end{equation}

    Next we bound $\|\frac{1}{m}\sum_{j=1}^m Z_j-\E Z_j\|$ with matrix Bernstein inequality. We first check the conditions of Theorem \ref{bernstein}:

    (I) We first bound the norm of $Z_j$. By Lemma \ref{lemma:z}, we have
    $$
    \|Z_j\|\le |r_i^{(1)}-r_i||\sigma(\bw_j^\top \bx_i)|\|V^\top\bw_j\|^3\le\tilde{O}(\frac{B^{5/2}}{m})
    $$
    with probability $1-\frac{\delta}{2Bm}$.

    (II) By Lemma \ref{lemma:z} we have
\begin{align*}
    \max&\left\{\left\|\mathbb{E}\left[Z^\top Z\right]\right\|,\left\|\mathbb{E}\left[Z Z^\top\right]\right\|\right\}\\
    &\lesssim \left(\E \left[|r_i^{(1)}-r_i|^4\right]\right)^{1/2} \left(\E\left[\sigma(\bw^\top \bx_i)\right]^4\left\|V^\top \bw\right\|^{12}\right)^{1/2}\\
    &\le\tilde{O}(\frac{B^5}{m^2}).
\end{align*}

(III) We have
\begin{align*}
    \max_{\|\ba\|=\|\bb\|=1}(\mathbb{E}(\ba^\top Z \bb)^2)^{1/2}&\lesssim\left(\mathbb{E}(r_i^{(1)}-r_i)^2\left[\sigma(\bw^\top \bx_i)\right]^2\left\|V^\top \bw\right\|^6\right)^{1/2}\\
    &\le \left(\E \left[|r_i^{(1)}-r_i|^4\right]\right)^{1/2} \left(\E\left[\sigma(\bw^\top \bx_i)\right]^4\left\|V^\top \bw\right\|^{12}\right)^{1/2}\\
    &\lesssim\frac{B^{5/2}}{m}.
\end{align*}

Moreover, $\|\E Z_j\|\lesssim\frac{B^{5/2}}{m}$ by Eq. \eqref{eq:T mean norm 2}. Then by Theorem \ref{bernstein}, we have
        $$
        \|\frac{1}{m}\sum_{j=1}^m Z_j-\E Z_j\|\le\tilde{O}(\frac{B^{5/2}}{m})
        $$
        with probability $1-\frac{\delta}{2B}$.
\end{proof}

Now we are ready to prove Proposition \ref{prop:twosteps} and Proposition \ref{prop:diff batch}.
\begin{proof}[Proof of Proposition \ref{prop:twosteps}]
    We define $g^{(0)}_j= \frac{1}{m}\sum_{i=1}^B r_i^{(0)}\sigma(\bw_j^{(0)}\cdot \bx_i)$, $g^{(1)}_j=\frac{1}{m} \sum_{i=1}^B r_i^{(1)}\sigma(\bw_j^{(1)}\cdot \bx_i)$ and $\tilde{g}_j=g^{(0)}_j+g^{(1)}_j$, then the target vector $\tilde{g}=\sum_{j=1}^m\tilde{g}_j$. We also define $\hat{P}_0=\sum_{j=1}^m g^{(0)}_j(\bw_j\bw_j^\top-I)$, $\hat{P}_1=\sum_{j=1}^m g^{(1)}_j(\bw_j\bw_j^\top-I)$, $\hat{\bT}_0=\sum_{j=1}^m g^{(0)}_j(\bw_j^{\otimes3}-\bw_j\tilde{\otimes}I)$, $\hat{\bT}_1=\sum_{j=1}^m g^{(1)}_j(\bw_j^{\otimes3}-\bw_j\tilde{\otimes}I)$, $\hat{P}=\hat{P}_0+\hat{P}_1$ and $\hat{\bT}=\hat{\bT}_0+\hat{\bT}_1$. Let $$P=\E \left[\sum_{i=1}^B r_i^*\sigma''(\bw^\top\bx_i)\bx_i\bx_i^\top\right]$$ and $$\bT=\E \left[\sum_{i=1}^B r_i^*\sigma^{(3)}(\bw^\top\bx_i)\bx_i^{\otimes 3}\right].$$ We will bound $\|\hat{P}_0+\hat{P}_1-2P\|$ and $\|(\hat{\bT}_0+\hat{\bT}_1-2\bT)(V,V,V)\|$.

    \textbf{Error bound of $P$.} We have
    \begin{align*}
        \|\hat{P}_0+\hat{P}_1-2P\|&\le 2\|\hat{P}_0-P\|+\|\hat{P}_1-\hat{P}_0\|\\
        &\le \tilde{O}(B\sqrt{\frac{d}{m}})+\|\hat{P}_1-\hat{P}_0\|
    \end{align*}
    with probability $1-\frac{\delta}{2}$, where the last inequality is by Lemma \ref{lemma:P bound}. Now we only need to bound $\|\hat{P}_1-\hat{P}_0\|$. By Lemma \ref{lemma:P1} and Lemma \ref{lemma:P2},
    \begin{align*}
        \|\hat{P}_1-&\hat{P}_0\|=\left\|\frac{1}{m}\sum_{i=1}^B\sum_{j=1}^mr_i^{(1)}\sigma(\bw_j^{(1)}\cdot\bx_i)\left(\bw_j\bw_j^\top-I\right)-\frac{1}{m}\sum_{i=1}^B\sum_{j=1}^mr_i\sigma(\bw_j^\top\bx_i)\left(\bw_j\bw_j^\top-I\right)\right\|\\
        &\le \sum_{i=1}^B\left\|\frac{1}{m}\sum_{j=1}^mr_i^{(1)}\left|\sigma(\bw_j^{(1)}\cdot\bx_i)-\sigma(\bw_j^\top\bx_i)\right|(\bw_j\bw_j^\top-I)\right\|\\
        &\quad\quad+\sum_{i=1}^B\left\|\frac{1}{m}\sum_{j=1}^m\left(r_i^{(1)}-r_i\right)\sigma(\bw_j^\top\bx_i)(\bw_j\bw_j^\top-I)\right\|\\
        &\le \tilde{O}(\frac{B^2d}{m})+\tilde{O}(\frac{B^2d}{m})
    \end{align*}
    with probability $1-\frac{\delta}{2}$. Since $m\ge\tilde{\Omega}(B^2d)$, $\|\hat{P}_1-\hat{P}_0\|\le O(B\sqrt{\frac{d}{m}})$. Thus, 
    $$\|\hat{P}-2 P\|=\|\hat{P}_0+\hat{P}_1-2P\|\le \tilde{O}(B\sqrt{\frac{d}{m}})$$ with probability $1-\delta$.

    \textbf{Error bound of $\bT(V,V,V)$.} We have
    \begin{align*}
        \|(\hat{\bT}_0+\hat{\bT}_1-2\bT)(V,V,V)\|&\le 2\|\hat{\bT}_0(V,V,V)-\bT(V,V,V)\|+\|\hat{\bT}_1(V,V,V)-\hat{\bT}_0(V,V,V)\|\\
        &\le \tilde{O}(\frac{B^{5/2}}{m})+\|\hat{\bT}_1(V,V,V)-\hat{\bT}_0(V,V,V)\|
    \end{align*}
    with probability $1-\frac{\delta}{2}$, where the last inequality is by Lemma \ref{lemma:T bound}. Now we only need to bound $\|\hat{\bT}(V,V,V)_1-\hat{\bT}_0(V,V,V)\|$. By Lemma \ref{lemma:T1} and Lemma \ref{lemma:T2},
    \begin{align*}
        &\|\hat{\bT}_1(V,V,V)-\hat{\bT}_0(V,V,V)\|\\
        =&\left\|\frac{1}{m}\sum_{i=1}^B\sum_{j=1}^mr_i^{(1)}\sigma(\bw_j^{(1)}\cdot\bx_i)\left(\bw_j^{\otimes3}-\bw_j\tilde{\otimes}I\right)(V,V,V)\right.\\
        &\quad\quad-\left.\frac{1}{m}\sum_{i=1}^B\sum_{j=1}^mr_i\sigma(\bw_j^\top\bx_i)\left(\bw_j^{\otimes3}-\bw_j\tilde{\otimes}I\right)(V,V,V)\right\|\\
        \le& \sum_{i=1}^B\left\|\frac{1}{m}\sum_{j=1}^mr_i^{(1)}\left|\sigma(\bw_j^{(1)}\cdot\bx_i)-\sigma(\bw_j^\top\bx_i)\right|(\bw_j^{\otimes3}-\bw_j\tilde{\otimes}I)(V,V,V)\right\|\\
        &\quad\quad+\sum_{i=1}^B\left\|\frac{1}{m}\sum_{j=1}^m\left(r_i^{(1)}-r_i\right)\sigma(\bw_j^\top\bx_i)\left(\bw_j^{\otimes3}-\bw_j\tilde{\otimes}I\right)(V,V,V)\right\|\\
        \le& \tilde{O}(\frac{B^{7/2}}{m})+\tilde{O}(\frac{B^{7/2}}{m})
    \end{align*}
    with probability $1-\frac{\delta}{2}$. Since $m\ge\tilde{\Omega}(B^2d)$, $\|\hat{\bT}_1-\hat{\bT}_0\|\le O(\frac{B^{5/2}}{m})$. Thus, 
    $$\|\hat{\bT}-2\eta \bT\|=\eta\|\hat{\bT}_0+\hat{\bT}_1-2\bT\|\le \eta\tilde{O}(\frac{B^{5/2}}{m})$$ with probability $1-\delta$.

    We have that $P$'s smallest component $\nu_{\mathrm{min}}\ge|\nu|$, $\bT$'s smallest component $\lambda_{\mathrm{min}}\ge |\nu|$ and $\kappa= \frac{\max|r_i^*|}{\min|r_i^*|}=1$. Since $0<|y_i|\le 1$, then $|r_i^*|$ is lower bounded. By Theorem \ref{prop:tensorerror}, 
$$
\|\bx_i-\hat{\bx_i}\| \leq  \frac{1}{\min\{|\nu|,|\lambda|\} \pi_{\min}^2}\tilde O(B\sqrt{\frac{d}{m}} )\
$$
for all $i=1,\dots,B$ with high probability.
\end{proof}

\begin{proof}[Proof of Proposition \ref{prop:diff batch}]
    We define $g^{(0)}_j= \frac{1}{m}\sum_{i=1}^B r_i^{(0)}\sigma(\bw_j^{(0)}\cdot \bx_i)$, $g^{(1)}_j=\frac{1}{m} \sum_{i=B+1}^N r_i^{(1)}sigma(\bw_j^{(1)}\cdot \bx_i)$, $\tilde{g}^{(1)}_j=\frac{1}{m} \sum_{i=B+1}^N r_i^{(0)}sigma(\bw_j^{(1)}\cdot \bx_i)$ and $\tilde{g}_j=g^{(0)}_j+g^{(1)}_j$, then the target vector $\tilde{g}=\sum_{j=1}^m\tilde{g}_j$. We also define $\hat{P}_0=\sum_{j=1}^m g^{(0)}_j(\bw_j\bw_j^\top-I)$, $\hat{P}_1=\sum_{j=1}^m g^{(1)}_j(\bw_j\bw_j^\top-I)$, $\tilde{P}_1=\sum_{j=1}^m \tilde{g}_j^{(1)}(\bw_j\bw_j^\top-I)$, $\hat{\bT}_0=\sum_{j=1}^m g^{(0)}_j(\bw_j^{\otimes3}-\bw_j\tilde{\otimes}I)$, $\hat{\bT}_1=\sum_{j=1}^m g^{(1)}_j(\bw_j^{\otimes3}-\bw_j\tilde{\otimes}I)$, $\tilde{\bT}_1=\sum_{j=1}^m \tilde{g}^{(1)}_j(\bw_j^{\otimes3}-\bw_j\tilde{\otimes}I)$, $\hat{P}=\hat{P}_0+\hat{P}_1$ and $\hat{\bT}=\hat{\bT}_0+\hat{\bT}_1$. Let $$P=\E \left[\sum_{i=1}^N r_i^*\sigma''(\bw^\top\bx_i)\bx_i\bx_i^\top\right]$$ and $$\bT=\E \left[\sum_{i=1}^N r_i^*\sigma^{(3)}(\bw^\top\bx_i)\bx_i^{\otimes 3}\right].$$ We will bound $\|\hat{P}_0+\hat{P}_1-2P\|$ and $\|(\hat{\bT}_0+\hat{\bT}_1-2\bT)(V,V,V)\|$.

    \textbf{Error bound of $P$.} We have
    \begin{align*}
        \|\hat{P}_0+\hat{P}_1-P\|&\le \|\hat{P}_0+\Tilde{P}_1-P\|+\|\hat{P}_1-\Tilde{P}_1\|\\
        &\le \tilde{O}(N\sqrt{\frac{d}{m}})+\|\hat{P}_1-\Tilde{P}_1\|
    \end{align*}
    with probability $1-\frac{\delta}{2}$, where the last inequality is by Lemma \ref{lemma:P bound}. Now we only need to bound $\|\hat{P}_1-\Tilde{P}_1\|$. By Lemma \ref{lemma:P1} and Lemma \ref{lemma:P2},
    \begin{align*}
        \|\hat{P}_1-&\Tilde{P}_1\|=\left\|\frac{1}{m}\sum_{i=B+1}^N\sum_{j=1}^mr_i^{(1)}\sigma(\bw_j^{(1)}\cdot\bx_i)\left(\bw_j\bw_j^\top-I\right)\right.\\
        &\quad\quad\quad\quad\left.-\frac{1}{m}\sum_{i=B+1}^N\sum_{j=1}^mr_i\sigma(\bw_j^\top\bx_i)\left(\bw_j\bw_j^\top-I\right)\right\|\\
        &\le \sum_{i=B+1}^N\left\|\frac{1}{m}\sum_{j=1}^mr_i^{(1)}\left|\sigma(\bw_j^{(1)}\cdot\bx_i)-\sigma(\bw_j^\top\bx_i)\right|(\bw_j\bw_j^\top-I)\right\|\\
        &\quad\quad+\sum_{i=B+1}^N\left\|\frac{1}{m}\sum_{j=1}^m\left(r_i^{(1)}-r_i\right)\sigma(\bw_j^\top\bx_i)(\bw_j\bw_j^\top-I)\right\|\\
        &\le \tilde{O}(\frac{(N-B)^2d}{m})+\tilde{O}(\frac{(N-B)^2d}{m})
    \end{align*}
    with probability $1-\frac{\delta}{2}$. Since $m\ge\tilde{\Omega}(B^2d)$, $\|\hat{P}_1-\hat{P}_0\|\le O(N\sqrt{\frac{d}{m}})$. Thus, 
    $$\|\hat{P}-P\|=\|\hat{P}_0+\hat{P}_1-P\|\le \tilde{O}(N\sqrt{\frac{d}{m}})$$ with probability $1-\delta$.

    \textbf{Error bound of $\bT(V,V,V)$.} We have
    \begin{align*}
        \|(\hat{\bT}_0+\hat{\bT}_1-\bT)(V,V,V)\|&\le \|(\hat{\bT}_0+\Tilde{\bT}-\bT)(V,V,V)\|+\|\hat{\bT}_1(V,V,V)-\Tilde{\bT}_1(V,V,V)\|\\
        &\le \tilde{O}(\frac{N^{5/2}}{m})+\|\hat{\bT}_1(V,V,V)-\Tilde{\bT}_1(V,V,V)\|
    \end{align*}
    with probability $1-\frac{\delta}{2}$, where the last inequality is by Lemma \ref{lemma:T bound}. Now we only need to bound $\|\hat{\bT}(V,V,V)_1-\hat{\bT}_0(V,V,V)\|$. By Lemma \ref{lemma:T1} and Lemma \ref{lemma:T2},
    \begin{align*}
        &\|\hat{\bT}_1(V,V,V)-\Tilde{\bT}_1(V,V,V)\|\\
        =&\left\|\frac{1}{m}\sum_{i=B+1}^N\sum_{j=1}^mr_i^{(1)}\sigma(\bw_j^{(1)}\cdot\bx_i)\left(\bw_j^{\otimes3}-\bw_j\tilde{\otimes}I\right)(V,V,V)\right.\\
        &\quad\quad-\left.\frac{1}{m}\sum_{i=B+1}^N\sum_{j=1}^mr_i\sigma(\bw_j^\top\bx_i)\left(\bw_j^{\otimes3}-\bw_j\tilde{\otimes}I\right)(V,V,V)\right\|\\
        \le&\!\! \sum_{i=B+1}^N\!\left\|\frac{1}{m}\sum_{j=1}^mr_i^{(1)}\!\left|\sigma(\bw_j^{(1)}\cdot\bx_i)-\sigma(\bw_j^\top\bx_i)\right|(\bw_j^{\otimes3}-\bw_j\tilde{\otimes}I)(V,V,V)\right\|\\
        &\quad\quad+\sum_{i=B+1}^N\left\|\frac{1}{m}\sum_{j=1}^m\left(r_i^{(1)}-r_i\right)\sigma(\bw_j^\top\bx_i)\left(\bw_j^{\otimes3}-\bw_j\tilde{\otimes}I\right)(V,V,V)\right\|\\
        \le& \tilde{O}(\frac{(N-B)^{7/2}}{m})+\tilde{O}(\frac{(N-B)^{7/2}}{m})
    \end{align*}
    with probability $1-\frac{\delta}{2}$. Since $m\ge\tilde{\Omega}(B^2d)$, $\|\hat{\bT}_1-\hat{\bT}_0\|\le O(\frac{N^{5/2}}{m})$. Thus, 
    $$\|\hat{\bT}-2\bT\|=\|\hat{\bT}_0+\hat{\bT}_1-2\bT\|\le \tilde{O}(\frac{N^{5/2}}{m})$$ with probability $1-\delta$.

    We have that $P$'s smallest component $\nu_{\mathrm{min}}\ge|\nu|$, $\bT$'s smallest component $\lambda_{\mathrm{min}}\ge |\nu|$ and $\kappa =\frac{\max|r_i^*|}{\min|r_i^*|}=1$. Since $0<|y_i|\le 1$, then $|r_i^*|$ is lower bounded. By Theorem \ref{prop:tensorerror}, 
$$
\|\bx_i-\hat{\bx_i}\| \leq  \frac{1}{\min\{|\nu|,|\lambda|\} \pi_{\min}^2}\tilde O(N\sqrt{\frac{d}{m}} )\
$$
for all $i=1,\dots,N$ with high probability.
\end{proof}

%% file: lowerbound.tex
\section{Analysis of Reconstruction Lower Bound}
\label{append:lower bound}
\subsection{The Minimax Risk}
We start with the following definition of an information-theoretical minimax risk:
\begin{definition}
\label{defminimax}
The minimax risk with batch size $n$ is defined as
\begin{equation}
\label{eq:lower 2}
    R_L=\left(\min_{\hat{S}=\hat{S}(O)}\max_{S\subset \cX^B}\min_{\pi}\E\left[d(S,\pi(\hat{S}))\right]\right)^{1/2}.
\end{equation}
Here $d(S,\pi(\hat{S}))=\frac{1}{B}\sum_{i=1}^B\|S_i-\hat{S}_{\pi(i)}\|^2$, where $\pi$ is a permutation of $[B]$. 
$\hat{S}=\hat{S}(O)$ is ranged over all algorithms represented as functions of $O$,
$O=O(S)$ is a random variable related with S, $S$ is ranged over all possible input data and $\pi\in S_n$
is ranged over all permutations with rank $n$.
\end{definition}
% \begin{remark}
% The minimum over permutation $\min\limits_{\pi}$ is added because batched gradient descent adds gradient together in an unordered manner, so the reconstruction algorithm could only identify the set of data points and cannot identify their correspondence.
% \end{remark}
The risk defined above provides a lower bound for the minimax expected risk of all attacking algorithms. As an information-theoretical or statistical minimax risk, we model the observation with a random noise. In our following analysis, $O$ is interpreted
as model gradients with a random Gaussian noise $O(S)=\nabla L(S;\Theta)+\epsilon$. The expectation in \ref{eq:lower 2} is taken over the random noise $\epsilon$.
% In cases where batch size equals 1, we get a simplified version of the minimax lower bound
% (same as the above definition):
% \begin{definition}
%     The minimax risk is defined as
%     \[R_{full}=\min\limits_{\hat{S}=\hat{S}(O)}\max\limits_{S}\E_{O(S)}\left[
%     d(S,\hat{S})\right]\]
%     where $d(S,\hat{S})=\|S-\hat{S}_{\pi}\|^2$.
%     $\hat{S}=\hat{S}(O)$ is ranged over all reconstruction algorithms represented as functions of $O$,
%     $O$ is a random variable related with S, $S$ is ranged over all possible input data.
% \end{definition}

\subsection{Analysis of 2-Layer Neural Networks}
In our setup of data reconstruction, we consider reconstructing training data based on the model gradients. Then for a 2-layer neural network $f(\bx;\Theta)$ with data and lables $\{(\bx_i,y_i)\}_{i=1}^B$ and square loss function $\ell$, $S=\{\bx_1,\dots,\bx_B\}$ is the input data and $O(S)=\sum_{i=1}^B\nabla_\Theta \ell(f(\bx_i),y_i)$ is the gradient. We further assume we already have information about the learning rate and the response variables. 

% Denote the model as $f_\theta(\bx)$ and data points as $(x_i,y_i)$ where $i=1,\dots,B$. Then in the following section 
% the correspondence of notations in Definition \ref{defminimax} are $S={x_1,\dots,x_B}$, $S_i=x_i$, $O(S)=\sum_{i=1}^{B}\nabla_\theta \ell (f(\bx_i),y_i)$

% \begin{lemma}
% \label{crminimax}
% Under certain regularity conditions and using the loss function $L (S,\hat{S})=\frac{1}{n}\norm{S-\hat{S}}$, the minimax risk is the Cramer-Rao lower bound.
% \end{lemma}
% {\red Yuxiao: I'm still not sure how to use the C-R bound to show our bound, more explanation is to be added}

%This section still has problems with formulation since there is no rigorous support for the bound
%An abuse of notation exists: activation function and noise scale are both sigma

Since we view data reconstruction as a statistical problem to estimate $S$ from $O(S)$, the expectation of $d(S,\hat{S})$ can be lower bounded by Cramer-Rao lower bound \citep{cramer1999mathematical,rao1992information}. In this problem, we define the Jacobian of the observation $O$ by $J(S)=\nabla_S O(S)$. Then the minimax risk
\begin{equation}
\label{eq:cr bound}
    R_L^2\ge \tr((J(S) J(S)^\top)^{-1})\sigma^2.
\end{equation}

To analyze the lower bound of minimax risk in data reconstruction, we first introduce some lemmas. For ease of calculation, we loosen the minimax risk as below:
\begin{lemma}
For positive-definite symmetric matrix $\mM\in\R^{d\times d}$, we have $\tr(\mM^{-1})\ge\frac{d^2}{\tr(\mM)}$
\end{lemma}
\begin{proof}
Let $\lambda_1,\dots,\lambda_d$ be the eigenvalues of $\mM$. Since $\mM$ is positive definite, the eigenvalues are real and positive.
Notice that $\tr(\mM)=\sum_{i=1}^{d}\lambda_i$ and $\tr(\mM^{-1})=\sum_{i=1}^{d}\frac{1}{\lambda_i}$, by Cauchy's inequality we have
\[\tr(\mM)\tr(\mM^{-1})=\sum_{i=1}^{d}\lambda_i\sum_{i=1}^{d}\frac{1}{\lambda_i}\le d^2.\]
Rearranging the inequality yields the desired result.
\end{proof}
\begin{remark}
\label{rmk:cr bound}
Applying the lower bound to $R_L$ in Eq. (\ref{eq:cr bound}), the bound is now 
\begin{equation}
    R_L^2\ge\frac{d^2}{\tr(J(S) J(S)^\top)}\sigma^2.
\end{equation}
\end{remark}

To analyze 2-layer neural networks, we make some assumptions about the model and how the model parameters are generated for analyzing scaling.
\begin{assumption}
\label{assumption2layer}
Let $\sigma$ be an activation function $\R\to\R$ such that $|\sigma'(x)|<C$ and $|\sigma''(x)|<C$ for any $x$ for some constant $C$.
Let the 2-layer network $f(\bx)$ be defined as $f(\bx)=\sum_{j=1}^{m}a_j\sigma(\bw_j^\top \bx)$, where $a_j\in\R$ are chosen from $\cN(0,\frac{1}{m^2})$ and $\bw_j\in\R^d$ are chosen from $\cN(0,I_d)$.
Let data points $(\bx_i,y_i)$ satisfy $\|\bx_i\|=1$ and
$|y_i|\in \{\pm 1\}$.
\end{assumption}
\begin{remark}
Popular activation functions including $ReLU$, $Sigmoid$, $LeakyReLU$ and $Tanh$ all satisfy the assumptions about the activation function.
\end{remark}
\begin{lemma}
\label{rmk:lbscaling}
With high probability we have the following scaling for all $i=1,\dots,B$ and $j=1,\dots,m$:
\begin{itemize}
    \item $a_j=\Tilde{O}(\frac{1}{m})$
    \item $\norm{\bw_j}=\Tilde{O}(\sqrt{d})$
    \item $|\bx_i^\top\bw_j|=\Tilde{O}(1)$
    \item $f(\bx_i)=\Tilde{O}(1)$.
\end{itemize}
Here the log terms are omitted.
\end{lemma}
The proof is trivial with concentration bounds.

We first prove the result using batch size 1 and then generalize to batch size higher than 1.

\begin{lemma}
\label{lem:lbsimple}
    Under Assumption \ref{assumption2layer} and with high probability,
    the lower bound of minimax risk $R_L$ with the observation
    \[
    O(S)=\nabla_{\theta}(y-f(\bx))^2+\vepsilon
    \]
    is of the scale $R_L\ge \Tilde{\Omega}(\sigma\sqrt{\frac{d}{m}})$, where $\epsilon$ follows $\cN(0,\sigma^2\mI_{md+m})$.
\end{lemma}

\begin{proof}
We first calculate the model gradients:
\[
\begin{cases}
    \nabla_{a_j}(y-f(\bx))^2=2(y-f(\bx))\sigma(\bw_j^\top \bx)\\
    \nabla_{\bw_j}(y-f(\bx))^2=2(y-f(\bx))a_j\sigma'(\bw_j^\top \bx)\bx.
\end{cases}
\]
We then calculate the Jacobian of the gradients on $x$:
\[
\begin{cases}
    \nabla_{\bx}\nabla_{a_j}(y-f(\bx))^2=2(y-f(\bx))\sigma'(\bw_j^\top \bx)\bw_j-2\sigma(\bw_j^\top \bx)\vh\\
    \nabla_{\bx}\nabla_{\bw_j}(y-f(\bx))^2=2(y-f(\bx))a_j\left(\sigma''(\bw_j^\top \bx)\bw_j\bx^\top+\sigma'(\bw_j^\top \bx)\mI\right)-2a_j\sigma'(\bw_j^\top \bx)\vh \bx^\top,
\end{cases}
\]
where $\vh=\nabla_\bx f(x)=\sum_{j=1}^{m}(a_j\sigma'(\bw_j^\top \bx)\bw_j)$. Since $a_j=\Tilde{O}(\frac{1}{m})$ with high probability and $\bw_j=\Tilde{O}(\sqrt{d})$ with high probability,
we have with high probability that
\[\norm{\vh}\le C\sum_{j=1}^{m}a_j\norm{\bw_j}=\Tilde{O}(\sqrt{d}).\]

Denote $J(\bx)$ as the Jacobian of model gradients on the input data. Notice that $y-f(\bx)=\Tilde{O}(1)$,
to calculate the lower bound, we first calculate $\tr(J(\bx)^\top J(\bx))$:
\begin{align*}
\tr(J(\bx)^\top J(\bx))=& \sum_{j=1}^{m}\norm{\nabla_{\bx}\nabla_{a_j}(y-f(\bx))^2}^2+\sum_{j=1}^{m}\norm{\nabla_{\bx}\nabla_{\bw_j}(y-f(\bx))^2}_F^2\\
\le&8(y-f(\bx))^2\sum_{j=1}^{m}\left[\norm{\sigma'(\bw_j^\top \bx)\bw_j}^2+\norm{a_j\left(\sigma''(\bw_j^\top \bx)\bw_j\bx^\top+\sigma'(\bw_j^\top \bx)\mI\right)}_F^2\right]\\
&+8\sum_{j=1}^{m}\left[\sigma^2(\bw_j^\top \bx)\norm{\vh}^2+a_j^2\sigma'^2(\bw_j^\top \bx)\norm{\vh}^2\norm{\bx}^2\right]
\end{align*}
where $\norm{\cdot}_F$ is the Frobenius norm. Note that the inequality used the fact that $\norm{A+B}\le 2(\norm{A}+\norm{B})$ for any $A$ and $B$.

For the first part, we have with high probability that
\begin{align*}
A_j:=&\norm{\sigma'(\bw_j^\top \bx)\bw_j}^2+\norm{a_j\left(\sigma''(\bw_j^\top \bx)\bw_j\bx^\top+\sigma'(\bw_j^\top \bx)\mI\right)}_F^2\\
=&\sigma'^2(\bw_j^\top \bx)\norm{\bw_j}^2+a_j^2\sigma''^2(\bw_j^\top \bx)\norm{\bw_j}^2\norm{\bx}^2\\
&+2a_j^2\sigma'(\bw_j^\top \bx)\sigma''(\bw_j^\top \bx)\bw_j^\top \bx+da_j^2\sigma'^2(\bw_j^\top \bx)\\
\le &C^2\left[\norm{\bw_j}^2+a_j^2\norm{\bw_j}^2+2a_j^2|\bw_j^\top \bx|+da_j^2\right]\\
= &C^2\left[\Tilde{O}(d)+\Tilde{O}(\frac{1}{m^2})\Tilde{O}(d)+2\Tilde{O}(\frac{1}{m^2})+d\Tilde{O}(\frac{1}{m^2})\right]\\
= &\Tilde{O}(d).
\end{align*}
For the second part, we have with high probability that
\begin{align*}
B_j:=&\sigma^2(w_j^\top x)\norm{\vh}^2+a_j^2\sigma'^2(w_j^\top x)\norm{\vh}^2\norm{x}^2\\
\le&C^2\Tilde{O}(d)+C^2\Tilde{O}(\frac{1}{m^2})\Tilde{O}(d)\\
=&\Tilde{O}(d).
\end{align*}

Therefore with a high probability,
\begin{align*}
\tr(J(\bx)^\top J(\bx))\le &8(y-f(\bx))^2\sum_{j=1}^{m}A_j+8\sum_{j=1}^{m}B_j\\
=&8O(1)\sum_{j=1}^{m}O(d)\\
=&O(md)
\end{align*}

By Remark \ref{rmk:cr bound}, we have $R_L^2\ge\frac{d^2}{\tr(J(\bx)^\top J(\bx))}\sigma^2=\Tilde{\Omega}(\frac{d}{m})\sigma^2$. Then $R_L\ge \Tilde{\Omega}(\sigma\sqrt{\frac{d}{m}})$.
\end{proof}

We now analyze the minimax risk of gradients summed over batch size $B$: $\vg(\bx_{1:B})=\sum_{i=1}^{B}\vg(\bx_i)$.
\begin{theorem}[Full version of Theorem \ref{thm:batchedmain short}]
\label{batchedmain}
Under Assumption \ref{assumption2layer} and with high probability,
    the lower bound of minimax risk $R_L$ with the observation with $B$ input samples
    \[
    O(S)=\sum_{i=1}^B\nabla_{\theta}(y-f(\bx_i))^2+\epsilon
    \]
    is of the scale $R_L\ge \Tilde{\Omega}(\sigma\sqrt{\frac{d}{m}})$, where $\epsilon$ follows $\cN(0,\sigma^2\mI_{md+m})$.
\end{theorem}

\begin{proof}
We first notice that the Jacobian $J_{1:B}(\bx_{1:B})$ equals a concatenate of $J_{1}(\bx_1),\dots,J_{B}(\bx_B)$:
\[
J_{1:B}(\bx_{1:B})=\begin{pmatrix}
J_{1}(\bx_1), \dots, J_{B}(\bx_B)
\end{pmatrix}.
\]
Therefore 
\begin{align*}
\tr(J_{1:B}(\bx_{1:B})^\top J_{1:B}(\bx_{1:B}))=&\sum_{i=1}^{B}\tr(J_{i}(\bx_i)^\top J_{i}(\bx_i))\\
=&O(mdB),
\end{align*}
Note that in Definition \ref{defminimax}, a scaling of $\frac{1}{B}$ was added to the sum of risks. By Remark \ref{rmk:cr bound},
we have the lower bound
$R_L^2\ge\frac{1}{B}\frac{B^2d^2}{\tr(J_{1:B}(\bx_{1:B})^\top J_{1:B}(\bx_{1:B}))}\sigma^2=\Tilde{\Omega}(\frac{d}{m})\sigma^2$. Then $R_L\ge \Tilde{\Omega}(\sigma\sqrt{\frac{d}{m}})$.
\end{proof}

\section{Proofs with Lower Bounds with Defenses}
\subsection{Dropout}
Dropout refers to randomly choosing certain entries of the model gradient and setting them to 0.
Suppose the gradient is dropped out with a certain probability $p$, then we only obtain $1-p$ of gradient information. In other words
each row of $J(x)$ is similar to the Jacobian matrix in Theorem \ref{batchedmain} but removed with probability $p$. Therefore our lower bound under high probability is
\[R_{full}\ge \sqrt{\frac{1}{B}\frac{\sigma^2 d^2 B^2}{(1-p)\tr(J(\vx)^\top J(\vx))}}=\Tilde{\Omega}\left(\sigma\sqrt{\frac{d}{(1-p)m}}\right)\]
where $J(\vx)$ is the complete Jacobian matrix without deleted rows.

\subsection{Gradient Clipping}
Clipping reduces the size of the gradient when the gradient size is too large. The noisy gradient could be written as 
\[\min\left(\norm{\nabla_\theta \sum_{i=1}^{B}(y-f(\bx_i))^2},C\right)\frac{\nabla_\theta \sum_{i=1}^{B}(y-f(\bx_i))^2}{\norm{\nabla_\theta \sum_{i=1}^{B}(y-f(\bx_i))^2}}+\vepsilon,\]
where $\theta$ stands for the model parameters and $\vepsilon$ is random noise with distribution $N(0,\frac{\sigma^2}{m^2}\mI_{md+m})$.

When the norm of the gradient is smaller than the constant $C$, the lower bound remains unchanged.
When the norm is higher than the constant $C$, the lower bound is of scale $\Omega(\frac{d\norm{\nabla_\theta \sum_{i=1}^{B}(y-f(\bx_i))^2}^2}{mC^2})\sigma^2$. (When the number of model parameters is high we could ignore the effect of changes in the norm on the bound.)
\begin{lemma}
    Under Assumption \ref{assumption2layer}, and for some constant $C$, when the norm of the model gradients $\norm{G}\ge C$, with high probability the lower bound of
minimax risk $R_L$ with the observation
    \[
    O(S)=C\frac{G}{\norm{G}}+\vepsilon=C\frac{\nabla_\theta \sum_{i=1}^{B}(y-f(\bx_i))^2}{\norm{\nabla_\theta \sum_{i=1}^{B}(y-f(\bx_i))^2}}+\vepsilon
    \]
    is of the scale $R_L\ge\Tilde{\Omega}(\sigma\sqrt{\frac{d}{m}}\frac{\norm{G}}{C})$, where $\vepsilon$ follows $\cN(0,\sigma^2\mI_{md+m})$ and $G=\nabla_\theta\sum_{i=1}^{B}(y-f(\bx_i))^2$.
\end{lemma}
\begin{proof}
Denote $\bg_\theta=C\frac{G}{\norm{G}}$, then for each $\bx_k$,
\begin{align*}
\nabla_{\bx_k} \bg_\theta=&C\frac{\nabla_{\bx_k}G}{\norm{G}}+C\left(\nabla_{\bx_k}\frac{1}{\norm{G}}\right)G^\top\\
=&C\frac{\nabla_{\bx_k}G}{\norm{G}}-C\left(\frac{\nabla_{\bx_k}G}{\norm{G}^3}G\right)G^\top
\end{align*}
Using the sub-multiplicative property of the Frobenius norm and scaling property with high probability of $\norm{\nabla_{\bx_k}G}=\norm{\nabla_{\bx_k}\nabla_\theta (y-f(\bx_k))^2}_F^2=\Tilde{O}(md)$ proved in Lemma \ref{lem:lbsimple}, we have that
\begin{align*}
\norm{\nabla_{\bx_k} \bg_\theta}_F^2\le &2C^2\norm{\frac{\nabla_{\bx_k}G}{\norm{G}}}_F^2+2C^2\norm{\frac{\nabla_{\bx_k}G}{\norm{G}^3}}_F^2\norm{GG^\top}_F^2\\
\le&2C^2\frac{\Tilde{O}(md)}{\norm{G}^2}+2C^2\frac{\Tilde{O}(md)}{\norm{G}^6}\norm{GG^T}_F^2\\
=&\frac{\Tilde{O}(C^2md)}{\norm{G}^2}(1+\norm{G}^{-4}\norm{G}^4)\\
=&\frac{\Tilde{O}(C^2md)}{\norm{G}^2}.
\end{align*}
Therefore we have with high probability that
\begin{align*}
\norm{\nabla_{\bx} \bg_\theta}_F^2=&\sum_{i=1}^B\norm{\nabla_{\bx_i} \bg_\theta}_F^2\\
=&\sum_{i=1}^B\frac{\Tilde{O}(C^2md)}{\norm{G}^2}\\
=&\frac{\Tilde{O}(C^2mdB)}{\norm{G}^2}
\end{align*}
Therefore our lower bound is given as
\[R_L\ge\Tilde{\Omega}(\sqrt{\frac{1}{B}\frac{B^2d^2}{\norm{\nabla_{\bx} \bg_\theta}_F^2}\sigma^2})=\Tilde{\Omega}(\sqrt{\frac{d}{m}}\frac{\norm{G}}{C}),\]
which is the desired result.
\end{proof}

%\begin{remark}
%Notice that the gradient norm $\norm{\nabla_\theta \sum_{i=1}^{B}(y-f(\bx_i))^2}$ scales $\Tilde{O}(\sqrt{m}B)$ with high probability. So if we keep $C$ constant with scaling, then the bound becomes $\Omega{dB^2}$.
%\end{remark}
%\begin{proof}
%This is because
%\[
%\begin{cases}
%\nabla_{a_j} \sum_{i=1}^{B}(y-f(\bx_i))^2=\sum_{i=1}^{B}2(y-f(\bx_i))\sigma(\bw_j^\top \bx_i)=\Tilde{O}(B)\\
%\norm{\nabla_{\bw_j}\sum_{i=1}^{B}(y-f(\bx_i))^2}^2=\norm{\sum_{i=1}^{B}2(y-f(\bx_i))a_j\sigma'(\bw_j^\top \bx_i)\bx_i}^2=\Tilde{O}(\frac{B^2}{m}).
%\end{cases}
%\]
%And 
%\[\norm{\nabla_\theta \sum_{i=1}^{B}(y-f(\bx_i))^2}^2=\sum_{j=1}^{m}|\nabla_{a_j} \sum_{i=1}^{B}(y-f(\bx_i))^2|^2+\sum_{j=1}^{m}\norm{\nabla_{\bw_j}\sum_{i=1}^{B}(y-f(\bx_i))^2}^2=\Tilde{O}(mB^2)\]
%\end{proof}

\subsection{Local Aggregation}
In local aggregation, we calculate two updates of model parameters using two different batches of size $B$. 
% Denote the loss function $\ell (\vz,\vq,y)=(y-f_\vq(\vz))^2$
% where $\vz$ is the input and $\vq$ are the model parameters. For the sake of simplicity, we make a slight abuse of notation and hide the term $y$ from the loss function.
In local aggregation, we use different learning rates for $a_j$ and $\bw_j$ since their scalings are different. Specifically, we take $\eta_a=O(\frac{1}{m^2d})$ and $\eta_w=O(\frac{1}{d})$. Then $a_j^{(1)}=a_j-\eta_a \sum_{i=1}^B\nabla_{a_j} \ell$ and $\bw_j^{(1)}=\bw_j-\eta_w\sum_{i=1}^B \nabla_{\bw_j} \ell$ for all $j$. Denote $\theta^{(1)}=(a^{(1)},\bw^{(1)})$ and we can similarly define $a^{(2)}$, $\bw^{(2)}$ and $\theta^{(2)}$. Denote $\ell(\bx_i,\theta,y_i)=(y_i-f_\theta(\bx_i))$, the direct observation from the local aggregation gradients is (with $y_j$ omitted)
$$
a_j^{(2)}-a_j=-\eta_a\left(\sum_{i=1}^B\nabla_{a_j}\ell(\bx_i,\theta)+\sum_{i=B+1}^{2B}\nabla_{a_j}\ell(\bx_i,\theta)\right)
$$
and
$$
\bw_j^{(2)}-\bw_j=-\eta_w\left(\sum_{i=1}^B\nabla_{\bw_j}\ell(\bx_i,\theta)+\sum_{i=B+1}^{2B}\nabla_{\bw_j}\ell(\bx_i,\theta)\right).
$$
In the setting of reconstruction attack, both of the learning rates are known so we can directly observe the gradients for two steps with a Gaussian noise.
\begin{equation}
\label{eq:aggre obs}
    O(S)=\sum_{i=1}^{B}\nabla_\theta \ell(\bx_i,\theta)+\sum_{i=B+1}^{2B}\nabla_{\theta^{(1)}} \ell\left(\bx_i,\theta^{(1)}\right)+\epsilon,
\end{equation}
where $\epsilon\sim\cN(0,\sigma^2I_{md+d})$.
% Suppose random noise is only added to the second
% step of update and learning rate is $\eta $. Let $\theta_1=\theta+\eta \frac{1}{B}\sum_{i=1}^{B}\nabla_q \ell(\bx_i,\theta)$ be the model parameters after the first update,
% then the observed update is
% \[
% \theta_{new}=\theta+\eta \left(\frac{1}{B}\sum_{i=1}^{B}\nabla_\theta \ell(\bx_i,\theta)+\frac{1}{B}\sum_{i=B+1}^{2B}\nabla_{\theta_1} \ell\left(\bx_i,\theta_1\right)+\vepsilon\right).
% \]
We want to reconstruct all $2B$ data points from the update. Intuitively, this is approximately equivalent to a regular update with batch size $2B$ and learning rate $2\eta$,
so the lower bound should remain unchanged. We now prove the bound rigorously.
\begin{proposition}[Full version of Proposition \ref{prop: aggre lower bound short}]
\label{prop:aggre lower bound}
If satisfy assumption \ref{assumption2layer}, $m>B$ and the learning rate for the two layers $\eta_w=O(1)$ and $\eta_a=O(\frac{1}{m^{3/2}})$, then
with high probability the lower bound of minimax risk $R_L$ with the observation Eq. (\ref{eq:aggre obs}) is of the scale $R_L\ge \Tilde{\Omega}(\sigma\sqrt{\frac{d}{m}})$.
\end{proposition}

\begin{proof}
% We first notice that the information we actually obtain is
% \[\frac{1}{B}\sum_{i=1}^{B}\nabla_\theta \ell(\bx_i,\theta)+\frac{1}{B}\sum_{i=B+1}^{2B}\nabla_{\theta_1} \ell\left(\bx_i,\theta_1\right)+\vepsilon.\]
% Note that $\theta_1$ is a function of $\theta$ and $\bx_1,\dots,\bx_B$. 
Define
\[g(\bx_{1:2B})=\sum_{i=1}^{B}\nabla_\theta \ell(\bx_i,\theta)+\sum_{i=B+1}^{2B}\nabla_{\theta^{(1)}} \ell\left(\bx_i,\theta^{(1)}\right),\]
our lower bound is given as
\[\frac{1}{2B}\frac{B^2d^2}{\sum_{i=1}^{2B}\norm{\nabla_{\bx_i}g(\bx_{1:2B})}^2}\sigma^2.\]

Notice that all entries of $\nabla_a \ell(\bx_i,\theta)$ are $\Tilde{O}(1)$ and $\nabla_w\ell(\bx_i,\theta)$ are $\Tilde{O}(\frac{1}{m})$:
\[
\begin{cases}
    \nabla_{a_j}\ell=2\sum_{i=1}^B(y-f(\bx_i))\sigma(\bw_j^\top \bx_i)\\
    \nabla_{\bw_j}\ell=2\sum_{i=1}^B(y-f(\bx_i))a_j\sigma'(\bw_j^\top \bx_i)\bx_i.
\end{cases}
\]
Therefore taking $\eta_a \le O(\frac{1}{B})$ and $\eta_w\le O(\frac{m}{B})$ would make the updated $a_j^{(1)}=a_j-\eta_a\nabla_{a_j}\ell$ and $\bw_j^{(1)}=\bw_j-\eta_w\nabla_{\bw_j}\ell$ still satisfy scaling in Lemma \ref{rmk:lbscaling}.

Therefore for $i\ge n+1$, with high probability
\[\norm{\nabla_{\bx_i}g(\bx_{1:2B})}_F^2=\norm{\nabla_{\bx_i}\nabla_{\theta^{(1)}} \ell\left(\bx_i,\theta_1\right)}_F^2=\Tilde{O}(md).\]

For $i\le n$, with high probability
\begin{equation}
\label{eq:start 2}
    \left\|\nabla_{\bx_i}\nabla_{a_j} \ell\left(\bx_i,\theta\right)\right\|^2 =\Tilde{O}(d),
\end{equation}
\begin{equation}
    \left\|\nabla_{\bx_i}\nabla_{\bw_j} \ell\left(\bx_i,\theta\right)\right\|_F^2=\Tilde{O}(\frac{d}{m^2}),
\end{equation}
\begin{equation}
\left|\nabla_{a_k}\nabla_{a_j}\ell\left(\bx_i,\theta\right)\right|=\Tilde{O}(1)
\end{equation}
for any $j,k$,
\begin{equation}
\left\|\nabla_{a_j}\nabla_{\bw_j}\ell\left(\bx_i,\theta\right)\right\|^2\le \Tilde{O}(1),
\end{equation}
\begin{equation}
    \left\|\nabla_{a_k}\nabla_{\bw_j}\ell\left(\bx_i,\theta\right)\right\|^2\le \Tilde{O}(\frac{1}{m^2})
\end{equation}
for $j\neq k$,
\begin{equation}
\left\|\nabla_{\bw_j}\nabla_{\bw_j}\ell\left(\bx_i,\theta\right)\right\|_F^2\le \Tilde{O}(\frac{1}{m^2})
\end{equation}
and
\begin{equation}
\label{eq:end 2}
\left\|\nabla_{\bw_k}\nabla_{\bw_j}\ell\left(\bx_i,\theta\right)\right\|_F^2\le \Tilde{O}(\frac{1}{m^4})
\end{equation}
for $j\neq k$.
Therefore, by Eq. (\ref{eq:start 2}) to (\ref{eq:end 2}) we have

\begin{align*}
\norm{\nabla_{\bx_i}g(\bx_{1:2B})}^2=&\left\|\nabla_{\bx_i}\nabla_{\theta^{(1)}} \ell\left(\bx_i,\theta\right)-\sum_{i=B+1}^{2B}\left(\eta_w \nabla_{\bx_i}\nabla_{\bw^{(1)}} \ell\left(\bx_i,\theta\right)\nabla_{\bw^{(1)}}\nabla_{a^{(1)}} \ell\left(\bx_i,\theta^{(1)}\right)\right.\right.\\
&\left.+\eta_w \nabla_{\bx_i}\nabla_{\bw^{(1)}} \ell\left(\bx_i,\theta\right)\nabla_{\bw^{(1)}}\nabla_{\bw^{(1)}} \ell\left(\bx_i,\theta^{(1)}\right)\right.\\
&\left.+\eta_a \nabla_{\bx_i}\nabla_{a^{(1)}} \ell\left(\bx_i,\theta\right)\nabla_{a^{(1)}}\nabla_{\bw^{(1)}} \ell\left(\bx_i,\theta^{(1)}\right)\right.\\
&\left.\left.+\eta_a \nabla_{\bx_i}\nabla_{a^{(1)}} \ell\left(\bx_i,\theta\right)\nabla_{a^{(1)}}\nabla_{a^{(1)}} \ell\left(\bx_i,\theta^{(1)}\right)\right)\right\|_F^2\\
\lesssim&\norm{\nabla_{\bx_i}\nabla_{\theta^{(1)}} \ell\left(\bx_i,\theta\right)}_F^2+\sum_{i=1}^B\left(\eta_w\norm{\nabla_{\bx_i}\nabla_{\bw^{(1)}} \ell\left(\bx_i,\theta\right)\nabla_{\bw^{(1)}}\nabla_{a^{(1)}} \ell\left(\bx_i,\theta^{(1)}\right)}_F^2\right.\\
&+\eta_w\norm{\nabla_{\bx_i}\nabla_{\bw^{(1)}} \ell\left(\bx_i,\theta\right)\nabla_{\bw^{(1)}}\nabla_{\bw^{(1)}} \ell\left(\bx_i,\theta^{(1)}\right)}_F^2\\
&+\eta_a\norm{\nabla_{\bx_i}\nabla_{a^{(1)}} \ell\left(\bx_i,\theta\right)\nabla_{a^{(1)}}\nabla_{\bw^{(1)}} \ell\left(\bx_i,\theta^{(1)}\right)}_F^2\\
&+\left.\eta_a\norm{\nabla_{\bx_i}\nabla_{a^{(1)}} \ell\left(\bx_i,\theta\right)\nabla_{a^{(1)}}\nabla_{a^{(1)}} \ell\left(\bx_i,\theta^{(1)}\right)}_F^2\right)\\
\lesssim&\Tilde{O}(md)+\sum_{i=1}^B\left(\eta_w^2 m \Tilde{O} (\frac{d}{m^2})(m\Tilde{O}(1)+m(m-1)\Tilde{O}(\frac{1}{m^2})\right.\\
&+\eta_w^2 m\Tilde{O}(\frac{d}{m^2})(m\Tilde{O}(\frac{1}{m^2})+m(m-1)\Tilde{O}(\frac{1}{m^4}))\\
&+\eta_a^2 m\Tilde{O}(d)(m\Tilde{O}(1)+m(m-1)\Tilde{O}(\frac{1}{m^2}))\\
&\left.+\eta_a^2 m\Tilde{O}(d)m^2\Tilde{O}(1)\right)\\
\le&\Tilde{O}(md)+\eta_w^2\Tilde{O}(Bd)+\eta_a^2\Tilde{O}(Bm^3d)\\
=&\Tilde{O}(md)
\end{align*}
if $\eta_w=O(\sqrt{m}{B})$ and $\eta_a=O(\frac{1}{m\sqrt{B}})$.

Therefore the lower bound is \[R_{L}^2\ge\frac{1}{B}\frac{d^2B^2}{\sum_{i=1}^{2B}\norm{\nabla_{x_i}g(x_{1:2B})}^2}\frac{1}{m^2}\sigma^2=\Tilde{\Omega}(\frac{d}{m})\sigma^2.\]
Thus, we have $R_L\ge \Tilde{\Omega}(\sigma\sqrt{\frac{d}{m}})$.
\end{proof}

\subsection{Gradient Pruning}
Gradient pruning refers to changing small entries of the model gradient to zero before updating the model. If a model gradient was pruned, the entries corresponding to this term in the Jacobian matrix
would become 0. Denote the Jacobian of pruned parameters' gradient on the input $\bx$ the matrix $J_0(\bx)$ (which is part of the full Jacobian $\tr(J(\bx)$), we will demonstrate in this section that gradient pruning has a lower bound at least as good as the regular bound in Proposition \ref{batchedmain}.
\begin{lemma}
\label{prop:prune}
    Under Assumption \ref{assumption2layer}, and for any constant $\gamma$, with high probability the lower bound of minimax risk $R_L$ with the observation
    \[
    P\left(\sum_{i=1}^{B}\nabla_{\theta}(y-f(\bx_i))^2\right)+\vepsilon
    \]
    is of scale $R_L\ge\Tilde{\Omega}(\sigma\sqrt{\frac{d}{m(1-\hat{p})}})$, where $P(\va)$ is the function that for all $i=1,\dots,\dim(\va)$,
    \[
    P_i(\va)=
    \begin{cases} 
        a_i & \text{if } |a_i| \ge \gamma \\
        0 & \text{if } |a_i| < \gamma
    \end{cases},
    \]
    and $$\hat{p}=\frac{\norm{J_0(\bx)}_F^2}{\norm{J(\bx)}_F^2}.$$
\end{lemma}
\begin{proof}
As the points that are not differentiable after applying this defense method are of zero measure, our analysis could ignore these cases.

 The Jacobian $\nabla_\bx P\left(\sum_{i=1}^{B}\nabla_{\theta}(y-f(\bx_i))^2\right)$ is $\nabla_\bx \sum_{i=1}^{B}\nabla_{\theta}(y-f(\bx_i))^2$ with the columns corresponding to columns in $J_0$ being set to 0. Our lower bound is given as
\[R_L\ge\Tilde{\Omega}(\sigma\sqrt{\frac{1}{B}\frac{d^2B^2}{\norm{J(\bx)}_F^2-\norm{J_0{\bx}}_F^2}})=\Tilde{\Omega}(\sigma\sqrt{\frac{1}{B}\frac{d^2B^2}{\norm{J(\bx)}_F^2}\frac{1}{1-\hat{p}}}).\]

Which gives us the desired result.
\end{proof}
\begin{remark}
Since entries of $G_{\bw_j}$ scale smaller compared to $G_{a_j}$, gradients from the first layer easier to be pruned. Therefore until the ratio of pruned gradients reaches $\frac{d}{d+1}$, the pruned gradients gives a small $\hat{p}$ and does not change the lower bound much. But higher than this ratio we start to prune gradients of $a_j$, so we predict that the bound will greatly increase after the ratio reaches $\frac{d}{d+1}$.
\end{remark}

%% file: exp_details.tex
\section{Additional experimental details for our proposed attack}
\label{sec:exp_setup}
% Code at \url{https://github.com/shengliu66/AttackDefenseEval}.
\paragraph{Image Prior Network} The architecture of our deep image prior network is structured on the foundation of PixelCNN++ as per \citep{salimans2017pixelcnn++}, %Salimans et al. (2017)
and it utilizes a U-Net as depicted by \citep{ronneberger2015u}, %Ronneberger et al. (2015), 
which is built upon a Wide-ResNet according to \citep{zagoruyko2016wide},
%Zagoruyko and Komodakis (2016).
To simplify the implementation, we substituted weight normalization \citep{salimans2016weight}
%(Salimans and Kingma, 2016)
with group normalization \citep{wu2018group}.
%(Wu and He, 2018). 
The models operating on $32 \times 32$ utilize four distinct feature map resolutions, ranging from $32 \times 32$ to $4 \times 4$. Each resolution level of the U-Net comprises two convolutional residual blocks along with self-attention blocks situated at the $16 \times 16$ resolution between the convolutional blocks. When available, the reconstructed features are incorporated into each residual block.

\paragraph{Backbone network for federated learning}
We adopt ResNet-18~\citep{he2016deep} architecture as our backbone network. The network is pre-trained for one epoch for better results of gradient inversion. Note that this architecture results in a more challenging setting given its minimal information retention capacity and is observed to cause a slight drop in quantitative performance and large variance~\citep{jeon2021gradient}.  We further add an extra linear layer before the classification head for feature reconstruction. In order to satisfy the concentration bound $\Tilde{O}(B\sqrt{\frac{d}{m}})$ of feature reconstruction, this linear layer has $512 \times 32 = 16384$ neurons. 

We adopt a special design for the last but one layer in favor of both the gradient matching procedure as well as intermediate feature reconstruction step. Specifically, we find out that a pure random net will degrade the performance of gradient matching compared to a moderately trained network. Therefore we split the neurons into two sets; we set the first set (first 4096 neurons) as trained parameters from ResNet-18 and set the weights associated with the rest neurons as random weights as required for the feature reconstruction step.

\paragraph{Hyper-parameters for training} For all experiments, we train the backbone ResNet-18 for 200 epochs, with a batch size of 64. We use SGD with a momentum of 0.9 as the optimizer. The initial learning rate is set to 0.1 by default. We decay the learning rate by a factor of 0.1 every 50 epochs. 

\paragraph{Hyper-parameters of the attack} The attack minimize the objective function given in Eq.~\ref{eq:final_objective}. We
search $\alpha_{\rm TV}$ in $\{0, 0.001, 0.005, 0.01, 0.05, 0.1, 0.5\}$ and $\alpha_{\rm BN}$ in $\{0, 0.0005, 0.001, 0.01, 0.05, 0.01\}$ for attack without defense, and apply the best choices $\alpha_{\rm TV} = 0.01$ and $\alpha_{\rm BN} = 0.001$ for all defenses (we did not tune $\alpha_{\rm TV}$ and $\alpha_{\rm BN}$ for each defense as we observe the proposed attack method is robust in terms of these two hyper-parameters). We optimize
the attack for 10,000 iterations using Adam~\citep{kingma2014adam}, with the initial learning rate set to $0.00001$. We decay the learning rate by a factor of $0.1$ at $3/8$, $5/8$, and $7/8$ of the optimization. For feature matching loss strength $\alpha_f$, we search in $\{0.01,0.05,0.1,0.5\}$ for each defense and find that $0.1$ is the optimal choice. 

\paragraph{Batch size of the attack} ~\citealp{zhu2019deep,geiping2020inverting} suggests that small batch size is vital for the success of the attack. We therefore intentionally evaluate the attack with three small
batch sizes $2$, $4$, and $8$ to test the upper bound of privacy leakage, and the minimum (and unrealistic) batch
size $1$, and two small but realistic batch sizes $16$ and $32$.

\paragraph{Gradient 
Reweighting}
The original gradient inversion loss, which calculates the cosine similarity of the gradients, can be expressed as
\begin{equation*}
    \sum_{i=1}^p\frac{\left\langle\nabla_{\theta_i} \mathcal{L}_{\theta_i}(\hat{x}, y), \nabla_{\theta_i} \mathcal{L}_{\theta_i}\left(x^{*}, y^{*}\right)\right\rangle}{\left\|\nabla _ { \theta_i } \mathcal { L } _ { \theta_i } ( \hat{x} , y ) \left|\left\|\mid \nabla_{\theta_i} \mathcal{L}_{\theta_i}\left(x^{*}, y^{*}\right)\right\|\right.\right.},
\end{equation*}
Here, $p$ represents the number of parameters in the federated neural network. To make the optimization of this loss easier, since some parameters can have much larger gradients and therefore overshadow the effect of others, we reweight the loss:
\begin{equation*}
    \sum_{i=1}^p w_i \frac{\left\langle\nabla_{\theta_i} \mathcal{L}_{\theta_i}(\hat{x}, y), \nabla_{\theta_i} \mathcal{L}_{\theta_i}\left(x^{*}, y^{*}\right)\right\rangle}{\left\|\nabla _ { \theta_i } \mathcal { L } _ { \theta_i } ( \hat{x} , y ) \left|\left\|\mid \nabla_{\theta_i} \mathcal{L}_{\theta_i}\left(x^{*}, y^{*}\right)\right\|\right.\right.},
\end{equation*}
where $w_i = \parallel\theta_i\parallel_0$.

\paragraph{Technical details on feature reconstruction}
The feature reconstruction method suggested by~\citealp{wang2023reconstructing} always produces features that are normalized and may have a sign opposite to the actual features. To compare the similarity between the original and the reconstructed features, we use cosine similarity, as it doesn’t change with feature magnitude. We square this similarity to avoid any issues with mismatched signs.

In addressing the challenge of integrating a tensor-based method, which requires randomly initialized networks, into gradient inversion processes that typically see improved performance with a few steps of preliminary training, we have designed a strategic approach for weight initialization. Since tensor-based method is only applied to the final fully connected layers, denoted as $A\sigma(Wx)$, where $A\in \mathbb{R}^{K\times m},\; W\in\mathbb{R}^{m\times d}$, with \(K\) representing the number of classes, \(m\) the number of hidden nodes, and \(d\) the input dimension. We consider the original pre-trained weights to be partitioned into two sets: \(W = [W_1, W_2]\) and \(A = [A_1; A_2]\). Here, \(W_2 \in \mathbb{R}^{v \times d}\) is initialized as a random Gaussian matrix, essential for the tensor-based method to function effectively. Concurrently, each row \(i\) of \(A_2 \in \mathbb{R}^{K \times v}\) is initialized to a constant value of \(i/v\), as suggested by the tensor-based approach. The remaining dimensions of the weight matrices, \(W_1 \in \mathbb{R}^{(m-v) \times d}\) and \(A_1 \in \mathbb{R}^{K \times (m-v)}\), retain their pre-trained weights, aligning with the requirements for successful gradient matching. For all of our experiments, the dimension allocated for gradient matching, \(m - v\), is set as 2048. This strategic initialization approach not only accommodates the prerequisites of the tensor-based method but also maintains the efficacy of the gradient-matching process.
% When dealing with multiple features at once (batch size greater than 1), it’s not always obvious in which order they were reconstructed. This can tempt one to try out all possible orders to find the correct one, especially when errors might be intentionally induced. But we've used a simpler (but greedy) way. We compare each reconstructed feature with the original ones in sequence to figure out the correct order, ensuring each feature is only paired once. This prevents a situation where one reconstructed feature is incorrectly matched with several original features. Experimental results suggest this way works well.

%  \paragraph{Challenges in order recovery of the features.}
 
When deploying the tensor-based method for feature recovery, a challenge arises due to the algorithm's inability to precisely maintain the order and sign of features when dealing with batches larger than a single input. This challenge complicates the task of accurately matching the recovered features to their corresponding inputs within a batch. In theory, one approach to circumvent this issue is to consider every possible permutation of the recovered features against the batch inputs, aiming to identify the correct alignment. However, this method is computationally intensive and impractical for larger batches. To streamline the process, we adopt a more efficient strategy that involves sequentially comparing each recovered feature against the actual input features using cosine similarity. This comparison continues until the most accurate order of features is determined based on their similarity scores. Although a single recovered feature might have high cosine similarity with several features corresponding to different inputs, we empirically observe that this greedy matching method achieves similar performance with the method that considers all possible pairs. 
% \paragraph{Experimental details for Figure~\ref{fig:feature}}
\section{Additional experimental details for other methods}
For GradientInversion~\citep{geiping2020inverting}, we adopted the code from~\citep{huang2021evaluating}, and the hyperparameters provided\footnote{\url{https://github.com/Princeton-SysML/GradAttack/}}. For Rob The Fed~\citep{fowl2021robbing}, we follow strictly the setting mentioned in their paper and the code\footnote{\url{https://github.com/lhfowl/robbing_the_fed}}, and introduced defenses to the gradients. For CPA~\citep{kariyappa2023cocktail}, since the settings and model architectures used in their paper are very different from ours, we incorporate their algorithm for feature recovery into our settings. i.e. using the latent features recovered by CPA rather than our proposed method. We perform hyperparameter selection through grid searches, but due to limited time, the hyperparameter searching process may not be comprehensive. But this should not significantly deviate our main goal -- comparison of the effectiveness of different defense methods, (not the attack algorithms).  
\label{sec:exp_setup_others}
\section{Hardware and software}
We run most of the experiments on clusters using NVIDIA V100s. All experiments are implemented using PyTorch library. Overall we estimated that a total of 800 GPU hours were consumed.

\section{Additional experimental results} 
\label{sec:add_results}
\paragraph{Effectiveness of gradient reweighting.} As observed in Table~\ref{tab:ablation_gr}, when dealing with a large batch-size, adjusting the gradients notably enhances the quality of the final image reconstruction. We hypothesize that this improvement is due to the adjustment potentially modifying the loss landscape, aiding in the optimization process.
\paragraph{Effectiveness of feature matching.} From Table~\ref{tab:ablation_gr},  it’s evident that feature matching offers a minor enhancement in final image reconstruction, particularly with smaller batch sizes. This indicates that feature matching plays a more crucial role when dealing with larger batch sizes. Additionally, feature matching proves to be beneficial when defenses are in place. As suggested by Figure~\ref{fig:defense}, feature matching can bring a notable improvement in reconstruction quality when applying defenses like additive noise. The comparison between Table~\ref{tab:w/ feature} and Table~\ref{tab:w/o feature} also indicates the effectiveness of feature matching in improving the robustness under defenses. For the defenses invalidating feature reconstruction and matching, for example, gradient pruning and large gradient noise, the performance has little improvement. On the contrary, for the defenses with no effect on feature reconstruction like gradient clipping, our method with feature matching improves a lot.
\paragraph{Effectiveness of random initialization.} We conduct additional experiments by attacking the model at different time steps. Specifically, the RMSE for attacks after 1, 2, and 3 epochs of training are as follows: 0.15 (0.04), 0.17 (0.09), and 0.22 (0.08), respectively. These results indicate that the attack becomes progressively more challenging as the training continues. We use the attack at random initialization in all other experiments, which is the strongest setting.

\begin{figure}
\vspace{-3mm}
    \centering
    \includegraphics[width=0.5\linewidth]{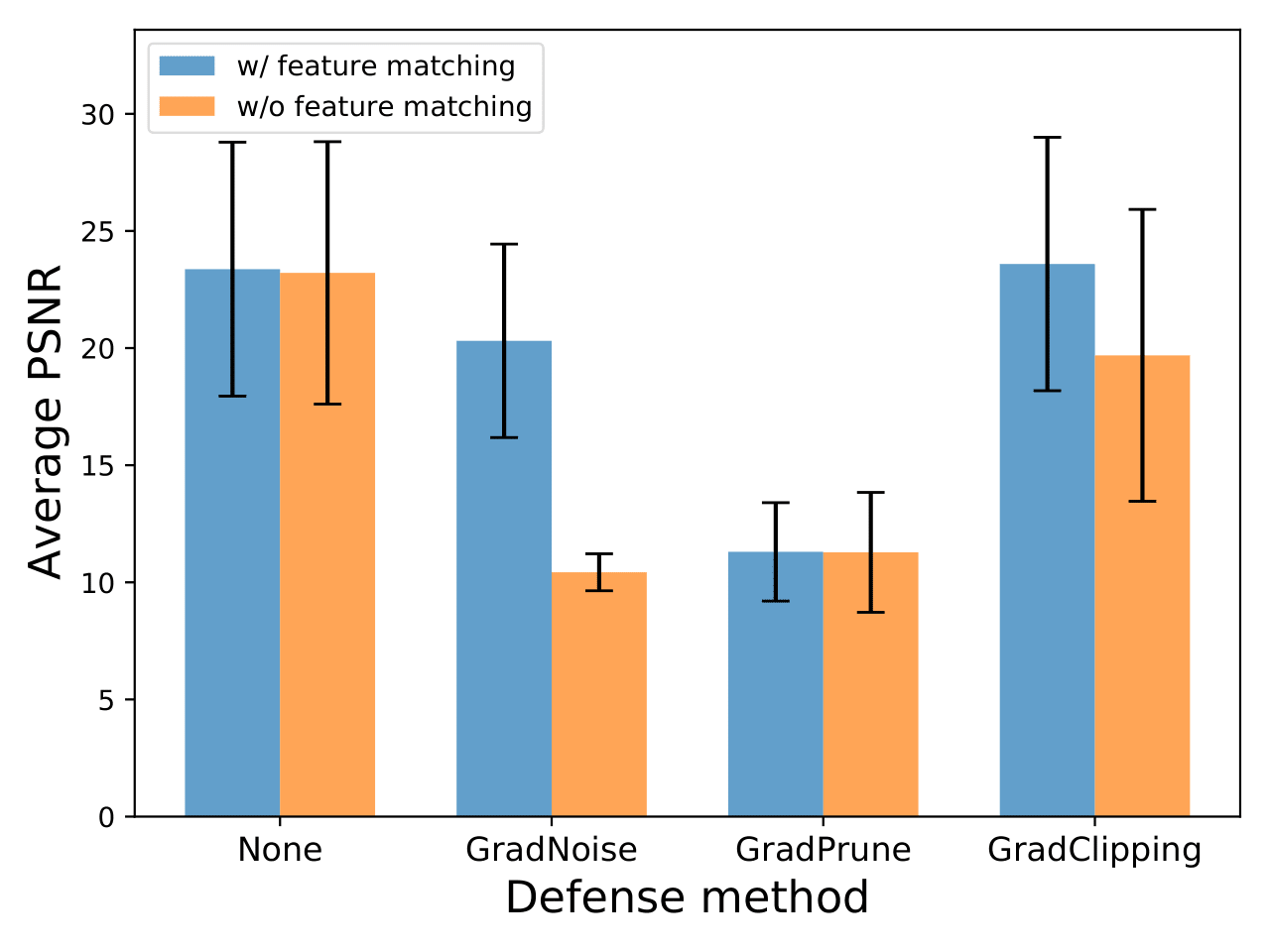}
    \caption{\small PSNR of reconstructed data using the method with and without feature matching. The difference indicates that feature matching improves robustness against defenses.}
    \label{fig:defense}
    \vspace{-3mm}
\end{figure}

% \paragraph{Final training loss.} In Table~\ref{tab:final loss}, we compute the final training loss of the federated learning model to show the interference brought by the defending methods. The smaller final training loss represents less interference to the training process. It shows that even gradient noise with small variance leads to larger interference than gradient pruning with large parameters, indicating gradient pruning is a less harmful defense.

\begin{table} 
  
  \setlength{\tabcolsep}{2.8pt}
  \renewcommand{\arraystretch}{0.95}
  \resizebox{1.\columnwidth}{!}{
  \begin{tabular}{l|c|c|c|c|c|c|c|c|c|c|c|c}
  \toprule
   batch size & \multicolumn{3}{c|}{{\bf 4 }} & \multicolumn{3}{c|}{{\bf 8 }} & \multicolumn{3}{c|}{{\bf 16 }} & \multicolumn{3}{c}{{\bf 32 }}\\
  \midrule
   {\bf Method}  & w/o GR and FM & w/o FM & Ours & w/o GR and FM & w/o FM & Ours & w/o GR and FM & w/o FM & Ours & w/o GR and FM & w/o FM & Ours \\
   \midrule
    {\bf PSNR $\uparrow$}  & 21.25 (6.02)  & 22.41 (5.84)  & 23.59 (5.39)  & 20.71 (5.96)  & 23.21 (5.60)  & 23.37 (5.42)  & 17.68 (5.61) & 20.01 (5.33)  & 21.90 (5.18) & 15.05 (4.05)  & 19.84 (5.30)  & 20.69 (4.88) \\
   {\bf RMSE $\downarrow$}  & 0.16(0.05)  & 0.16(0.04)  & 0.15 (0.03)  & 0.18 (0.05)  & 0.16 (0.04)  & 0.15 (0.04)  & 0.21 (0.06)  & 0.18 (0.06)  & 0.16 (0.04)  & 0.23 (0.06)  & 0.18 (0.05)  &  0.17 (0.05)\\
    {\bf LPIPS $\downarrow$}  & 0.12 (0.08)  & 0.09 (0.08)   & 0.09 (0.07)  & 0.15 (0.12)  & 0.11 (0.09)  & 0.10 (0.09) & 0.24 (0.15)  & 0.16 (0.13)  & 0.14 (0.11)  & 0.30 (0.12)  & 0.18 (0.12)  & 0.15 (0.11)\\
  \bottomrule
  \end{tabular}}
  \caption{\small Results of ablation study on gradient reweighting (GR) and feature matching (FM). Both feature matching and gradient reweighting contribute the the final performance boost. When batch size is large (e.g. $=16/32)$), gradient reweighting significantly improves the performance.} 
  \label{tab:ablation_gr}
  \vspace{-0mm}
\end{table}

\begin{table}[H]
  \scriptsize
  \setlength{\tabcolsep}{2.8pt}
  \renewcommand{\arraystretch}{0.95}
  \resizebox{.99\textwidth}{!}{
  \begin{tabular}{l|c|c|c|c|c|c|c|c|c|c|c|c|c|c|c}
  \toprule
   &  \bf None & \multicolumn{4}{c|}{\bf GradNoise ($\sigma$)} & \multicolumn{5}{c|}{\bf GradPrune ($p$)} & \multicolumn{3}{c|}{\bf GradClipping ($C$)} & \multicolumn{2}{c}{\bf Local Aggregation (steps)}\\
   % \multicolumn{11}{c}{\bf Attack batch size $= 1$} \\
   \midrule
   {\bf Parameter}  & - & 0.001 & 0.01 & 0.05 & 0.1 & 0.3 & 0.5 & 0.7 & 0.9 & 0.99 & 2 & 4 & 8 & 3 & 5 \\
   % {\bf RMSE $\downarrow$}  & 0.16(0.02)  & 0.19  & 0.24 (0.01)  & 0.25 (0.00)  & 0.42 & 0.22 (0.11)  & 0.21 (0.11)  & 0.22 (0.10)  & 0.16 (0.11)  & 0.17 (0.11)  & 0.17 (0.11) \\
   %  {\bf PSNR $\uparrow$}  & 12.18 (0.59)  & 0.19  & 12.01 (0.52)  & 11.58 (0.59)  & 0.42 & 19.55 (6.08) & 19.23 (5.49)  & 18.04 (5.12)  & 23.63 (5.47)  & 23.18 (5.52)  & 23.57 (5.94)\\
   \midrule
   \multicolumn{16}{c}{\bf Attack batch size $= 2$} \\
   \midrule
   {\bf RMSE $\downarrow$}  & 0.15(0.04)  & 0.17 (0.05)  & 0.22 (0.07)  & 0.31 (0.06)  & 0.28 (0.08)  & 0.16 (0.08)  & 0.17 (0.07) & 0.20 (0.09) & 0.26 (0.09) & 0.27 (0.10) & 0.17 (0.08)  & 0.17 (0.08)  & 0.19 (0.08) & 0.23 (0.09) & 0.25 (0.10)  \\
    {\bf PSNR $\uparrow$}  & 23.69 (5.69)  & 23.31 (3.04)  &  21.32 (3.52)  & 14.58 (1.59)  &  
  13.11 (4.26) & 22.63 (6.15)  & 23.33 (6.03)  & 19.40 (6.02) & 16.39 (6.56) & 14.72 (6.45)  & 23.26 (6.10)  & 23.40 (6.00)  & 23.01 (6.28) & 18.65 (6.75) & 18.77 (7.33)\\
   \midrule
   \multicolumn{16}{c}{\bf Attack batch size $= 4$} \\
   \midrule
   {\bf RMSE $\downarrow$}  & 0.15 (0.03)  & 0.19 (0.07)  & 0.24 (0.08)  & 0.29 (0.03)  & 0.27 (0.07)  & 0.16 (0.04)  & 0.16 (0.03)  & 0.21 (0.04) & 0.27 (0.07) & 0.27 (0.08) & 0.16 (0.03)  & 0.16 (0.04)  & 0.16 (0.04) & 0.28 (0.08) & 0.26 (0.08)\\
    {\bf PSNR $\uparrow$}  & 23.59 (5.39)  & 20.12 (3.58)  & 16.41 (4.64)  & 11.27 (0.33)  & 12.28 (4.08)  & 23.16 (5.48)  & 24.08 (5.38) & 18.73 (5.63) & 13.90 (3.77) & 12.43 (2.73) & 22.88 (5.31)  & 23.93 (5.55)  & 24.04 (5.58)  & 13.84 (5.38) & 14.21 (5.41)\\
   \midrule
   \multicolumn{16}{c}{\bf Attack batch size $= 8$} \\
   \midrule
   {\bf RMSE $\downarrow$}  & 0.15 (0.04)  & 0.19 (0.05)  & 0.29 (0.05)  & 0.30 (0.03)  & 0.30 (0.06)  & 0.15 (0.03)  & 0.16 (0.04) & 0.20 (0.05) & 0.29 (0.05) & 0.29 (0.05) & 0.16 (0.03)  & 0.16 (0.04)  & 0.16 (0.04) & 0.29 (0.04) & 0.30 (0.04)\\
    {\bf PSNR $\uparrow$}  & 23.37 (5.42)  & 20.31 (4.13)  & 14.56 (0.88)  & 11.27 (0.89)  & 11.27 (1.83)  & 23.74 (5.21)  & 23.32 (5.13) & 18.25 (5.22) & 11.30 (2.10) & 11.13 (2.15) & 23.36 (5.14)  & 23.59 (5.41) & 23.47 (5.23) & 11.25 (2.56)  & 10.82 (2.20) \\
  \bottomrule
  \end{tabular}}
  \caption{\small Our method evaluated with different defense methods. With feature matching, our method performs well against most defenses. Specifically, gradient pruning shows its effectiveness under this stronger attack as well.} %the performance under relatively small gradient noise largely improved comparing with the method without feature matching.} 
  \label{tab:w/ feature}
\end{table}

\begin{table}[H]
  \scriptsize
  \setlength{\tabcolsep}{2.8pt}
  \renewcommand{\arraystretch}{0.95}
  \resizebox{.99\textwidth}{!}{
  \begin{tabular}{l|c|c|c|c|c|c|c|c|c|c|c|c|c|c|c}
  \toprule
   &  \bf None & \multicolumn{4}{c|}{\bf GradNoise ($\sigma$)} & \multicolumn{5}{c|}{\bf GradPrune ($p$)} & \multicolumn{3}{c|}{\bf GradClipping ($C$)} & \multicolumn{2}{c}{\bf Local Aggregation (steps)}\\
  \midrule
   {\bf Parameter}  & - & 0.001 & 0.01 & 0.05 & 0.1 & 0.3 & 0.5 & 0.7 & 0.9 & 0.99 & 2 & 4 & 8 & 3 & 5 \\
   \midrule
   % \multicolumn{11}{c}{\bf Attack batch size $= 1$} \\
   % \midrule
   % {\bf RMSE $\downarrow$}  & 0.16(0.02)  & 0.19  & 0.24 (0.01)  & 0.25 (0.00)  & 0.42 & 0.22 (0.11)  & 0.21 (0.11)  & 0.22 (0.10)  & 0.16 (0.11)  & 0.17 (0.11)  & 0.17 (0.11) \\
   %  {\bf PSNR $\uparrow$}  & 12.18 (0.59)  & 0.19  & 12.01 (0.52)  & 11.58 (0.59)  & 0.42 & 19.55 (6.08) & 19.23 (5.49)  & 18.04 (5.12)  & 23.63 (5.47)  & 23.18 (5.52)  & 23.57 (5.94)\\
   % \midrule
   \multicolumn{16}{c}{\bf Attack batch size $= 2$} \\
   \midrule
   {\bf RMSE $\downarrow$}  & 0.16 (0.05)  & 0.30 (0.05)  & 0.32 (0.07)  & 0.32 (0.06)  & 0.35 (0.07)  & 0.19 (0.08)  & 0.20 (0.07) & 0.24 (0.09) & 0.28 (0.09) & 0.27 (0.11) & 0.19 (0.07)  & 0.23 (0.09)  & 0.25 (0.08) & 0.25 (0.10) & 0.25 (0.10)\\
    {\bf PSNR $\uparrow$}  & 22.51 (4.72)  & 12.27 (3.27)  &  12.12 (3.45)  & 11.48 (1.53)  &  
  11.77 (2.00) & 20.31 (4.35)  & 20.32 (7.03)  & 19.23 (7.02) & 15.73 (6.70) & 14.01 (5.95) & 21.25 (4.10)  & 21.40 (6.00)  & 21.01 (4.28) & 17.52 (6.88) & 16.24 (6.58)\\
   \midrule
   \multicolumn{16}{c}{\bf Attack batch size $= 4$} \\
   \midrule
   {\bf RMSE $\downarrow$}  & 0.16 (0.04)  & 0.31 (0.01)  & 0.31 (0.03)  & 0.31 (0.03)  & 0.31 (0.07)  & 0.19 (0.04)  & 0.18 (0.03)  & 0.23 (0.04) & 0.29 (0.07) & 0.28 (0.07) & 0.21 (0.03)  & 0.22 (0.04)  & 0.23 (0.03)  & 0.28 (0.078) & 0.25 (0.09)\\
    {\bf PSNR $\uparrow$}  & 22.41 (5.84)  & 12.12 (3.58)  & 10.53 (1.03)  & 10.51 (1.03)  & 9.98 (1.35)  & 20.16 (4.26)  & 21.08 (5.28) & 17.03 (4.59) & 11.88 (2.66) & 12.64 (4.04)  & 19.02 (4.26)  & 19.64 (4.33)  &  21.38 (4.27) & 13.08 (5.16) & 14.31 (6.02)\\
   \midrule
   \multicolumn{16}{c}{\bf Attack batch size $= 8$} \\
   \midrule
   {\bf RMSE $\downarrow$}  & 0.16 (0.04)  & 0.30 (0.02)  & 0.31 (0.02)  & 0.31 (0.02)  & 0.31 (0.71)  & 0.20 (0.03)  & 0.21 (0.04) & 0.25 (0.04) & 0.30 (0.05) & 0.30 (0.05) & 0.22 (0.04)  & 0.21 (0.04)  & 0.21 (0.04) & 0.31 (0.04) &  0.31 (0.04)\\
    {\bf PSNR $\uparrow$}  & 23.21 (5.60)  & 10.43 (0.79)  & 10.15 (0.80)  & 10.46 (0.87)  & 10.54 (1.38)  & 21.56 (4.32)  & 21.46 (6.23) & 17.25 (4.59) & 11.28 (2.56) & 11.17 (2.19) & 19.36 (5.14)  & 19.69 (6.23)  & 20.43 (4.27)  & 10.01 (2.22) & 10.26 (2.03)\\
  \bottomrule
  \end{tabular}}
  \caption{\small Results of the proposed method without feature matching loss with different defenses. Under defenses having no effect on feature reconstruction, the performance of the method without feature matching is much worse than that with feature matching.} 
  \label{tab:w/o feature}

\end{table}